\documentclass[sigconf]{acmart}
\usepackage{epsfig}
\usepackage{amssymb,amsmath,amsthm,amsfonts}
\usepackage{url}
\usepackage{tikz}
\usepackage{graphicx}
\usepackage{thmtools, thm-restate}
\usepackage{mathpartir}
\usepackage{wrapfig}
\usepackage{booktabs}
\usepackage{tikz}
\usetikzlibrary{trees}
\usetikzlibrary{fit}
\usetikzlibrary{positioning}
\usepackage{tikz-cd}
\usepackage{xcolor}
\usepackage[all,color]{xy}
\usepackage[T1]{fontenc}
\usepackage{stmaryrd}
\usepackage{array}
\usepackage{mathrsfs}
\usepackage{xspace}
\usepackage{leftidx}
\usepackage{wrapfig}
\usepackage{microtype}
\usepackage{balance}
\usepackage{utf8math}
\usepackage{cleveref}
\usepackage{endnotes}
\usepackage{flushend}

\makeatletter
\def\@secfont{\bfseries\Large\section@raggedright}
\makeatother

\usepackage{pifont}
\makeatletter
\newcommand\Pimathsymbol[3][\mathord]{%
  #1{\@Pimathsymbol{#2}{#3}}}
\def\@Pimathsymbol#1#2{\mathchoice
  {\@Pim@thsymbol{#1}{#2}\tf@size}
  {\@Pim@thsymbol{#1}{#2}\tf@size}
  {\@Pim@thsymbol{#1}{#2}\sf@size}
  {\@Pim@thsymbol{#1}{#2}\ssf@size}}
\def\@Pim@thsymbol#1#2#3{%
  \mbox{\fontsize{#3}{#3}\Pisymbol{#1}{#2}}}
\makeatother
\newcommand{\myneq}{\Pimathsymbol[\mathrel]{mathb}{"88}}
\DeclareFontFamily{U}{mathb}{\hyphenchar\font45}
\DeclareFontShape{U}{mathb}{m}{n}{
      <5> <6> <7> <8> <9> <10> gen * mathb
      <10.95> mathb10 <12> <14.4> <17.28> <20.74> <24.88> mathb12
      }{}

\renewcommand{\sfdefault}{cmss}

\newcommand{\IRS}{\ensuremath{\mathrm{IRS}}\xspace}
\newcommand{\clio}{\textsc{Clio}\xspace}
\newcommand{\lioS}{$\mathrm{LIO}$\xspace}

\newcommand{\flows}{\sqsubseteq}
\newcommand{\flowsstrict}{\myneq}
\newcommand{\lub}{\sqcup}
\newcommand{\glb}{\sqcap}

\newcommand{\lto}{\longrightarrow}

\newcommand{\lcur}[1][l]{\ensuremath{{#1}_{\textrm{cur}}}}
\newcommand{\ccurr}{\ensuremath{l_{\textrm{clr}}}}
\newcommand{\lclr}{\ccurr}
\newcommand{\gtau}{\boldsymbol{\underline{\tau}}}

\newcommand{\true}{\textsf{true}}
\newcommand{\false}{\textsf{false}}
\newcommand{\unit}{\textsf{()}}
\newcommand{\fix}[1]{\textsf{fix}\ #1}
\newcommand{\Labeled}[2]{\ensuremath{\text{\guilsinglleft}{#2}\!:\!{#1}\text{\guilsinglright}}}
\newcommand{\LabeledO}[2]{\Labeled{#1}{#2}}

\newcommand{\LIO}[1]{{#1}^\textsf{LIO}}
\newcommand{\CLIO}[1]{{#1}^\textsf{CLIO}}

\newcommand{\ltrip}[3]{\ensuremath{\langle {#1} , {#2}, {#3} \rangle}}

\newcommand{\ifel}[3]{\textsf{if}\ #1\ \textsf{then}\ #2\ \textsf{else}\ #3}
\newcommand{\join}[2]{#1 \lub #2}
\newcommand{\meet}[2]{#1 \glb #2}
\newcommand{\cjoin}[2]{#1 \lub^C #2}
\newcommand{\cmeet}[2]{#1 \glb^C #2}
\newcommand{\ijoin}[2]{#1 \lub^I #2}
\newcommand{\imeet}[2]{#1 \glb^I #2}
\newcommand{\ajoin}[2]{#1 \lub^A #2}
\newcommand{\ameet}[2]{#1 \glb^A #2}
\newcommand{\canFlowTo}[2]{#1 \flows #2}
\newcommand{\cantFlowTo}[2]{#1 \not\flows #2}
\newcommand{\canFlowToC}[2]{#1 \flows^C #2}
\newcommand{\cantFlowToC}[2]{#1 \not\flows^C #2}
\newcommand{\canFlowToI}[2]{#1 \flows^I #2}

\newcommand{\canFlowToA}[2]{#1 \flows^A #2}

\newcommand{\canFlowToStrict}[2]{#1 \flowsstrict #2}
\newcommand{\canFlowToStrictA}[2]{#1 \flowsstrict^A #2}
\newcommand{\canFlowToStrictC}[2]{#1 \flowsstrict^C #2}
\newcommand{\canFlowToStrictI}[2]{#1 \flowsstrict^I #2}

\newcommand{\canTrustTo}[2]{#1 \trusts #2}

\newcommand{\return}[1]{\textsf{return}\ #1}
\newcommand{\bindSymbol}{\ensuremath{\gg\!\!=}}
\newcommand{\bind}[2]{#1 \bindSymbol #2}
\newcommand{\getLabel}{\textsf{getLabel}}
\newcommand{\getClearance}{\textsf{getClearance}}
\newcommand{\llabel}[2]{\textsf{label}\ #1\ #2}
\newcommand{\labelOf}[1]{\textsf{labelOf}\ #1}
\newcommand{\unlabel}[1]{\textsf{unlabel}\ #1}
\newcommand{\toLabeled}[2]{\textsf{toLabeled}\ #1\ #2}
\newcommand{\lowerClr}[1]{\textsf{lowerClearance}\ #1}

\newcommand{\block}[4]{\ensuremath{{}^{#1}_{#2}\{^{#3}\, {#4}\, \}}}
\newcommand{\store}[2]{\textsf{store}\ {#1}\ {#2}}

\newcommand{\fetchvar}[3]{\textsf{fetch}_{\,{#3}}\ {#1}\ {#2}}
\newcommand{\fetch}[2]{\fetchvar{#1}{#2}{\tau}}

\newcommand{\BoolT}{\textsf{Bool}}
\newcommand{\unitT}{\textsf{()}}
\newcommand{\lio}{\ensuremath{\mathit{LIO}}}
\newcommand{\LIOT}[1]{\textsf{LIO}\ #1}
\newcommand{\CLIOT}[1]{\textsf{CLIO}\ #1}
\newcommand{\LabelT}{\textsf{Label}}
\newcommand{\LabeledT}[1]{\textsf{Labeled}\ #1}

\newcommand{\typeOf}[1]{\textsf{typeOf}\ #1}

\newcommand{\gv}{\ensuremath{\underline{v}}}

\newcommand{\lowEquiv}{\ensuremath{=_{\normalfont{\lowlabel}}}}

\newcommand{\lowEquivC}{\ensuremath{=^C_{\normalfont{\lowlabel}}}}

\newcommand{\attacker}{\ensuremath{\mathcal{A}}\xspace}
\newcommand{\drawnFrom}{\ensuremath{\leftarrow}}

\newcommand{\probabilityOfCond}[2]{\ensuremath{\mathrm{Pr}[#1~|~#2]}}

\newcommand{\conf}[3]{\ensuremath{\langle #1, #2 \mid #3\rangle}}

\newcommand{\randomExprL}{\ensuremath{\{\!\!\,|~}}
\newcommand{\randomExprR}{\ensuremath{~|\!\!\,\}}}
\newcommand{\randomExpr}[2]{\ensuremath{\{\!|\ {#1}\ |\ {#2}\ |\!\}}}

\newcommand{\trusts}{\ensuremath{\preceq}}

\newcommand{\projC}[1]{\ensuremath{\mathbb{C}({#1})}}
\newcommand{\projI}[1]{\ensuremath{\mathbb{I}({#1})}}
\newcommand{\projA}[1]{\ensuremath{\mathbb{A}({#1})}}

\newcommand{\event}{\ensuremath{\alpha}}

\newcommand{\lowlabel}{\ensuremath{\ell}\xspace}

\newcommand{\textor}{\ensuremath{\mathrm{~or~}}}
\newcommand{\textif}{\ensuremath{\mathrm{if~}}}
\newcommand{\textand}{\ensuremath{\mathrm{~and~}}}

\newcommand{\storecmd}[2]{\ensuremath{\textsf{put}~{#2}~\textsf{at}~{#1}}}

\newcommand{\fetchcmdvar}[3]{\ensuremath{\textsf{got}_{\,#3}~{#2}~\textsf{at}~{#1}}}
\newcommand{\fetchcmd}[2]{\fetchcmdvar{{#1}}{{#2}}{\tau}}
\newcommand{\missingcmd}[1]{\ensuremath{\textsf{nothing-at}~{#1}}}
\newcommand{\missing}{\ensuremath{\bot}\xspace}

\newcommand{\ltol}[1]{\ensuremath{\xrightarrow{#1~}~}}

\newcommand{\iconf}[2]{\ensuremath{\langle #1, #2 \rangle}}
\newcommand{\istore}{\ensuremath{\sigma} }
\newcommand{\istoredist}{\ensuremath{{\varmathbb{XXX}} }}
\newcommand{\ito}{\ensuremath{\rightsquigarrow}}
\newcommand{\emptystore}{\ensuremath{\emptyset}}

\newcommand{\concat}[2]{\ensuremath{#1 ~\cdot~ #2}}

\newcommand{\version}{\ensuremath{n}}
\newcommand{\incrementnm}{\ensuremath{\mathrm{increment}}}
\newcommand{\increment}[1]{\ensuremath{\incrementnm({#1})}}
\newcommand{\interactionsdist}{\ensuremath{{\varmathbb{R}}}}
\newcommand{\interaction}{\ensuremath{R}}
\newcommand{\interactionsvar}[1]{\ensuremath{\overline{#1}}}
\newcommand{\interactions}{\ensuremath{\overline{\interaction}}}
\newcommand{\interactionsprime}{\ensuremath{\overline{\interaction'}}}
\newcommand{\bitstream}{\ensuremath{b}\xspace}
\newcommand{\versions}{\ensuremath{\mathbf{V}}}
\newcommand{\rconf}[4]{\ensuremath{\langle #1, #3, #4 \rangle}}
\newcommand{\rto}[1]{\ensuremath{\rightsquigarrow_{#1}}}
\newcommand{\keystore}{\ensuremath{\mathcal{P}}}
\newcommand{\serializenm}{\ensuremath{\mathsf{serialize}}}
\newcommand{\serialize}[2]{\ensuremath{\serializenm_{\keystore}({#1}, {#2})}}
\newcommand{\deserializenm}{\ensuremath{\mathsf{deserialize}}}
\newcommand{\deserializevar}[3]{\ensuremath{\deserializenm_{\keystore}({#1}, {#2}, {#3})}}
\newcommand{\deserialize}[2]{\deserializevar{{#1}}{{#2}}{\tau}}

\newcommand{\fetchcknm}{\ensuremath{\textsf{fetch\_ck}_{\keystore}}}
\newcommand{\fetchck}[2]{\ensuremath{\fetchcknm({#1}, {#2})}}

\newcommand{\initializecknm}{\ensuremath{\textsf{initialize\_ck}_{\keystore}}}

\newcommand{\createcknm}{\ensuremath{\textsf{create\_ck}_{\keystore}}}

\newcommand{\Lowstep}[1]{\ensuremath{\curvearrowright_{#1}}}
\newcommand{\pcOf}[1]{\ensuremath{\mathrm{PC}({#1})}}

\newcommand{\adversary}{\ensuremath{\mathcal{A}}\xspace}
\newcommand{\strategy}{\ensuremath{\mathcal{S}}\xspace}
\newcommand{\principals}{\ensuremath{\tilde{p}}\xspace}
\newcommand{\INDvar}[7]{\mathrm{#7}_{{#1}}({#2}, {#3}, {#4}, {#5}, {#6})}
\newcommand{\IND}[6]{\INDvar{#1}{#2}{#3}{#4}{#5}{#6}{IND}}

\newcommand{\distfam}[1]{\ensuremath{\Big\{~{#1}~\Big\}_n}}
\newcommand{\cequiv}{\ensuremath{\approx}}
\newcommand{\cryptosys}{\ensuremath{\Pi}\xspace}
\newcommand{\nat}{\ensuremath{\varmathbb{N}}}
\newcommand{\Gen}{\textsf{Gen}}
\newcommand{\Sign}{\textsf{Sign}}
\newcommand{\Verify}{\textsf{Verify}}
\newcommand{\Enc}{\textsf{Enc}}
\newcommand{\Dec}{\textsf{Dec}}

\newcommand{\zeroes}{\ensuremath{\Sigma_0}}
\newcommand{\pubkeys}[1]{\ensuremath{\mathrm{pub}({#1})}}

\newcommand{\authorityOf}[1]{\ensuremath{\mathsf{authorityOf}({#1})}}
\newcommand{\Start}[1]{\ensuremath{\mathsf{Start}({#1})}}
\newcommand{\Clr}[1]{\ensuremath{\mathsf{Clr}({#1})}}
\newcommand{\Min}[1]{\ensuremath{\mathsf{min}({#1})}}
\newcommand{\Max}[1]{\ensuremath{\mathsf{max}({#1})}}
\newcommand{\multistepnm}{\textsf{step}}
\newcommand{\multistep}[4]{\ensuremath{\multistepnm^{#1}_{\lowlabel}({#2}, {#3}, {#4})}}

\newcommand{\skipcmd}{\textsf{skip}\xspace}
\newcommand{\corruptcmd}[1]{\textsf{corrupt}~{#1}}
\newcommand{\category}{\ensuremath{C}}
\newcommand{\categorykey}{\ensuremath{ck}}

\newcommand{\translate}[4]{\ensuremath{\mathcal{T}_{{#1}}({#2}, {#3}, {#4})}}
\newcommand{\iinteraction}{\ensuremath{I}\xspace}
\newcommand{\iinteractionsvar}[1]{\ensuremath{\overline{#1}}}
\newcommand{\iinteractions}{\ensuremath{\overline{\iinteraction}}}

\newcommand{\rulename}[1]{{\footnotesize{\text{\textsc{#1}}}}}

\newcommand{\figLowEquivMain}{
	$\gv_1$ & $\lowEquiv$ & $\gv_2$ & if $\gv_1 = \gv_2$ \\
	$\LabeledO{l_1}{\gv_1}$ & $\lowEquiv$ & $\LabeledO{l_1}{\gv_2}$& where \\
	
\multicolumn{4}{>{\hspace{1.5em}}l}{
 $ 
\left\{   
\begin{array}{lll}
	\gv_1 = \gv_2 & \mathrm{if} & \projC{l_1} \flows^C \projC{\lowlabel} \\ &&
	\mathrm{and~} \projI{l_1} \flowsstrict^I \projI{\lowlabel}  \\ &&
	 \mathrm{and~} \projA{l_1} \flowsstrict^A \projA{\lowlabel}, \\
	 \typeOf{\gv_1} = \typeOf{\gv_2} & \multicolumn{2}{l}{\mathrm{otherwise} }
\end{array}
\right.
$} \\
	\\
	$(v_1, v_2)$ & $\lowEquiv$ & $(v_1', v_2')$& if $v_1 \lowEquiv v_1'$ and $v_2 \lowEquiv v_2'$ \\
	$t_1\ t_2$ & $\lowEquiv$ & $t_1'\ t_2'$& if $t_1 \lowEquiv t_1'$ and $t_2 \lowEquiv t_2'$\\
}

\newcommand{\figLowEquivFull}{
	\multicolumn{3}{l}{{\bf Terms:}}\\
	\figLowEquivMain

	$\fix{t}$ & $\lowEquiv$ & $\fix{t'}$& if $t \lowEquiv t'$ \\
	$\ifel{t_1}{t_2}{t_3}$ & $\lowEquiv$ & $\ifel{t_1'}{t_2'}{t_3'}$& if $t_1 \lowEquiv t_1'$ and $t_2 \lowEquiv t_2'$ and $t_3 \lowEquiv t_3'$\\
	$\return{t}$ & $\lowEquiv$ & $\return{t'}$& if $t \lowEquiv t'$ \\
	$\bind{t_1}{t_2}$ & $\lowEquiv$ & $\bind{t_1'}{t_2'}$& if $t_1 \lowEquiv t_1'$ and $t_2 \lowEquiv t_2'$\\
	$\llabel{t_1}{t_2}$ & $\lowEquiv$ & $\llabel{t_1'}{t_2'}$& if $t_1 \lowEquiv t_1'$ and $t_2 \lowEquiv t_2'$\\
	$\labelOf{t}$ & $\lowEquiv$ & $\labelOf{t'}$& if $t \lowEquiv t'$ \\
	$\unlabel{t}$ & $\lowEquiv$ & $\unlabel{t'}$& if $t \lowEquiv t'$ \\
	$\getLabel$ & $\lowEquiv$ & $\getLabel$ & \\
	$\getClearance$ & $\lowEquiv$ & $\getClearance$ & \\
	$\lowerClr{t}$ & $\lowEquiv$ & $\lowerClr{t'}$& if $t \lowEquiv t'$ \\
	$\toLabeled{t_1}{t_2}$ & $\lowEquiv$ & $\toLabeled{t_1'}{t_2'}$& if $t_1 \lowEquiv t_1'$ and $t_2 \lowEquiv t_2'$\\
	$\block{l_1}{l_2}{l_3}{t}$ & $\lowEquiv$ & $\block{l_1'}{l_2'}{l_3}{t'}$& if $l_1 \lowEquiv l_1'$ and $l_2 \lowEquiv l_2'$ and $t \lowEquiv t'$\\
	$\store{t_1}{t_2}$ & $\lowEquiv$ & $\store{t_1'}{t_2'}$& if $t_1 \lowEquiv t_1'$ and $t_2 \lowEquiv t_2'$\\
	$\fetch{t_1}{t_2}$ & $\lowEquiv$ & $\fetch{t_1'}{t_2'}$& if $t_1 \lowEquiv t_1'$ and $t_2 \lowEquiv t_2'$\\
	\\ 
	\multicolumn{3}{l}{{\bf Configurations:}}\\
	$\conf{\lcur}{\lclr}{t}$ & $\lowEquiv$ & $\conf{\lcur'}{\lclr'}{t'}$ & if $t \lowEquiv t'$ and $\lcur = \lcur'$ and  $\lclr = \lclr'$ \\
	$\conf{\lcur}{\lclr}{t}$ & $\lowEquiv$ & $\conf{\lcur'}{\lclr'}{t'}$ & if ($\lcur \not\flows \lowlabel$ and $\lcur' \not\flows \lowlabel$) and ($\lclr \not\flows \lowlabel$ and $\lclr' \not\flows \lowlabel$)\\

	\\

}

\newcommand{\figRealClioSyntax}{
	\begin{alignat*}{2}
	\textrm{Keystores:} & \quad
	  \keystore & :~~~ &
		p \rightarrow (b, b_\bot) \\
	\textrm{Bit strings:} & \quad
	  \bitstream & \in~~~ &
	 	\{0,1\}^* \\
	\textrm{Stores:} &  \quad
	 \istore & :~~~ & 
	 	(\gv \rightarrow \LabeledO{l}{\bitstream}_{\!\bot}) \cup (C \rightarrow ck_{\!\bot})\\
	\textrm{Versions:} & \quad
	 \versions & :~~~ &
	    \gv \rightarrow \version \\
	\textrm{Interactions:} &  \quad
	 \interactions & \Coloneqq~ &
	 	\concat{\interaction}{\interactions} \ ~|~\ \interaction \\
	& \quad \interaction & \Coloneqq~ &
	 	\skipcmd = \lambda \istore. ~ \istore \\
	&& \!\!|\ ~~ &  \storecmd{\category}{\categorykey} = \lambda \istore. ~ \istore[\category \mapsto \categorykey]\\
	&& \!\!|\ ~~ & \storecmd{\gv_k}{\Labeled{l}{\bitstream}} = \lambda \istore. ~ \istore[\gv_k \mapsto \Labeled{l}{\bitstream}]\\ 
	\textrm{Strategies:} & \quad
		\strategy & :~~~ & \interactionsdist \rightarrow \interactionsdist \\
	\textrm{Category Keys:} & \quad
	 \categorykey & \Coloneqq~ &
	 	(b,~esk,~\bitstream') \\
	\textrm{Encrypted Keys:} & \quad
	 esk & :~~~ &
	 	p \rightarrow \bitstream
	\end{alignat*}

	\begin{tabular}{lcl}
	\multicolumn{3}{c}{{\bf keystore Label Functions}}\\
	$\Min{\keystore}$ & $=$ & $p_1 \vee ... \vee p_i \vee ... \vee p_n$ for all $p_i \in dom(\keystore)$\\
	$\Max{\keystore}$ & $=$ & $p_1 \wedge ... \wedge p_i \wedge ... \wedge p_n$ for all $p_i \in dom(\keystore)$ \\ 
	\multicolumn{3}{r}{and $sk_i \neq \bot$ where $\keystore(p_i) = (pk_i, sk_i) $}\\
	$\Start{\keystore}$ & $=$ & $\ltrip{\Min{\keystore}}{\Max{\keystore}}{\Max{\keystore}} \rangle$\\
	$\Clr{\keystore}$ & $=$ & $\ltrip{\Max{\keystore}}{\Min{\keystore}}{\Min{\keystore}} \rangle$\\
	\multicolumn{3}{l}{
		$\authorityOf{\keystore} =\ltrip{\Max{\keystore}}{\Max{\keystore}}{\Max{\keystore}} $
	}
	\end{tabular}
}

\newcommand{\figClioSemanticsBaseMain}{
	\rulename{LabelOf}\hfill
	\vspace{-1em}\\
	\inferrule
	{
	}
	{
	  \conf{\lcur}{\lclr}{\labelOf{(\LabeledO{l_1}{t})}} \lto \conf{\lcur}{\lclr}{l}
	}
	\vspace{-1em}
	\\
	\rulename{Return}\hfill
	\vspace{-1em}\\
	\inferrule
	{ }
	{
	  \conf{\lcur}{\lclr}{\return{t}} \lto \conf{\lcur}{\lclr}{\CLIO{t}}
	}
	\vspace{-1em}
	\\
	\rulename{Bind}\hfill
	\vspace{-1em}\\
	\inferrule
	{ }
	{
	\conf{\lcur}{\lclr}{\bind{\CLIO{t_1}}{t_2}} \lto \conf{\lcur}{\lclr}{(t_2\ t_1)}
	}
	\vspace{-1em}
	\\
	\rulename{Label}\hfill
	\vspace{-1em}\\
	\inferrule
	{ \canFlowTo{\lcur}{l_1}  \and
	  \canFlowTo{l_1}{\lclr}
	}
	{
	\conf{\lcur}{\lclr}{\llabel{l_1}{\gv}} \lto \conf{\lcur}{\lclr}{\return{(\LabeledO{l_1}{\gv})}}
	}
	\vspace{-1em}
	\\
	\rulename{Unlabel}\hfill
	\vspace{-1em}\\
	\inferrule
	{ \join{\lcur}{l_1} = l_2 \and \canFlowTo{l_2}{\lclr}}
	{
	\conf{\lcur}{\lclr}{\unlabel{(\LabeledO{l_1}{t_2})}} \lto \conf{l_2}{\lclr}{\return{t_2}}
	}
	\vspace{-0.5em}
	\\
	\rulename{ToLabeled}\hfill
	\vspace{-1em}\\
	\inferrule
	{\canFlowTo{\lcur}{l_1}  \and
	  \canFlowTo{l_1}{\lclr} \and
	}
	{
	\conf{\lcur}{\lclr}{\toLabeled{l_1}{t}}\lto \conf{\lcur}{\lclr}{\block{\lcur}{\lclr}{l_1}{t}}
	}
	\vspace{-1em}
	\\
	\rulename{Reset}\hfill
	\vspace{-1em}\\
	\inferrule
	{\canFlowTo{\lcur}{l_2}
	}
	{
	\conf{\lcur}{\lclr}{\block{l_1}{l_3}{l_2}{\CLIO{t}}}\lto \conf{l_1}{l_3}{\llabel{l_2}{t}}
	}
}

\newcommand{\figClioSemanticsStoreMain}{
	\rulename{Store}\hfill
	\vspace{-1em}\\
	\inferrule
	{\canFlowTo{\lcur}{\lowlabel} \and
	\canFlowTo{\lcur}{l_1} \and
	\event = \storecmd{\gv_k}{\Labeled{{l_1}}{\gv}}
	}
	{
	\conf{\lcur}{\lclr}{\store{\gv_k}{\Labeled{l_1}{\gv}}}\ltol{\event} 
	    \conf{\lcur}{\lclr}{\return{\unitT}}
	}
	\vspace{-1em}
	\\
	\rulename{Fetch-Valid}\hfill
	\vspace{-1em}\\
	\inferrule
	{
	\canFlowToA{\projA{\lowlabel}}{\projA{l_d}} \\\\
	\event = \fetchcmd{\gv_k}{\Labeled{l}{\gv}} \and
	\canFlowTo{l}{{l_d}} 
	}
	{
	\conf{\lcur}{\lclr}{\fetch{\gv_k}{\Labeled{l_d}{\gv_d}}}\ltol{\event} 
	    \conf{\lcur}{\lclr}{{\return{\Labeled{l_d}{\gv}}}}
	}
	\vspace{-1em}
	\\
	\rulename{Fetch-Invalid}\hfill
	\vspace{-1em}\\
	\inferrule
	{
	\canFlowToA{\projA{\lowlabel}}{\projA{l_d}}
	\\\\
	(\event = \missingcmd{\gv_k})
		 \mathrm{~or~} 
	(\event = \fetchcmd{\gv_k}{\Labeled{l}{\gv}} \mathrm{~and~} \cantFlowTo{l}{{l_d}} )
	}
	{
	\conf{\lcur}{\lclr}{\fetch{\gv_k}{{\Labeled{l_d}{\gv_d}}}}\ltol{\event} 
	    \conf{\lcur}{\lclr}{\return{\Labeled{l_d}{\gv_d}}}
	}
}

\newcommand{\figLioSemanticsFull}{
	\inferrule[app]
	{
	}
	{
	  \conf{\lcur}{\lclr}{(\lambda x. t_1)\ t_2} \lto \conf{\lcur}{\lclr}{[t_2/x ]~t_1}
	}
	\and
	\inferrule[fix]
	{
	}
	{
	  \conf{\lcur}{\lclr}{\fix{(\lambda x. t)}} \lto \conf{\lcur}{\lclr}{[\fix{(\lambda x. t)}/x ]~t}
	}
	\and
	\inferrule[ifTrue]
	{
	}
	{
	  \conf{\lcur}{\lclr}{\ifel{\true}{t_2}{t_3}} \lto \conf{\lcur}{\lclr}{t_2}
	}
	\and
	\inferrule[getLabel]
	{ }
	{
	  \conf{\lcur}{\lclr}{\getLabel} \lto \conf{\lcur}{\lclr}{\return{\lcur}}
	}
	\and
	\inferrule[ifFalse]
	{
	}
	{
	  \conf{\lcur}{\lclr}{\ifel{\false}{t_2}{t_3}} \lto \conf{\lcur}{\lclr}{t_3}
	}
	\and
	\inferrule[getClearance]
	{ }
	{
	  \conf{\lcur}{\lclr}{\getClearance} \lto \conf{\lcur}{\lclr}{\return{\lclr}}
	}
	\and
	\inferrule[lowerClearance]
	{ \canFlowTo{\lcur}{l_1} \and
	  \canFlowTo{l_1}{\lclr} }
	{
	\conf{\lcur}{\lclr}{\lowerClr{l_1}} \lto \conf{\lcur}{l_1}{\return{()}}
	}
	\and
	\inferrule[LabelOp]
	{
	\otimes \in \{ \glb, \lub, \flows \} \and
	\gv = l_1 \otimes l_2
	}
	{
	  \conf{\lcur}{\lclr}{l_1 \otimes l_2} \lto \conf{\lcur}{\lclr}{\gv}
	}
	\and
	\inferrule[Step]
	{
          \conf{\lcur}{\lclr}{t} \lto \conf{\lcur'}{\lclr'}{t'}
	}
	{
          \conf{\lcur}{\lclr}{E[t]} \lto \conf{\lcur'}{\lclr'}{E[t']}
	}

}

\newcommand{\figLioTypingRules}{
	\inferrule[bool]
	{
		b \in \{ \true, \false \}
	}{
		\Gamma \vdash b : \BoolT
	} \and
	\inferrule[unit]
	{
	}{
		\Gamma \vdash \unit : \unitT
	} \and
	\inferrule[label]
	{
	}{
		\Gamma \vdash l : \LabelT
	} \and
	\inferrule[pair]
	{
		\Gamma \vdash t_1 : \tau_1
		\and \Gamma \vdash t_2 : \tau_2
	}{
		\Gamma \vdash (t_1, t_2) : (\tau_1, \tau_2)
	} \and
	\inferrule[labeled]
	{
		\Gamma \vdash v : \tau
	}{
		\Gamma \vdash \Labeled{l}{v} : \LabeledT{\tau}
	}
	\and
	\inferrule[var]
	{
		\Gamma(x) = \tau
	}{
		\Gamma \vdash x : \tau
	} \and
	\inferrule[abstraction]
	{
		\Gamma[x \mapsto \tau_1] \vdash t_1 :  \tau_2 
	}{
		\Gamma \vdash \lambda x.\ t_1 : \tau_1 \rightarrow \tau_2
	} \and
	\inferrule[app]
	{
	\Gamma \vdash t_1 : \tau_2 \rightarrow \tau_1
	\and
	\Gamma \vdash t_2 : \tau_2
	}
	{
	   \Gamma \vdash {t_1\ t_2} : \tau_1
	}
	\and
	\inferrule[fix]
	{
	\Gamma \vdash  t : (\tau_1 \rightarrow \tau_2) \rightarrow \tau_2
	}
	{
	  \Gamma \vdash \fix{t} : \tau_1 \rightarrow \tau_2
	}
	\and
	\inferrule[LIO]
	{
		\Gamma \vdash t : \tau
	}{
		\Gamma \vdash \LIO{t} : \LIOT{\tau}
	}
	\and
	\inferrule[if]
	{
	\Gamma \vdash t_1 : \BoolT \and
	\Gamma \vdash t_2 : \tau \and
	\Gamma \vdash t_3 : \tau 
	}
	{
	  \Gamma \vdash {\ifel{t_1}{t_2}{t_3}} : \tau
	}
	\and
	\inferrule[LabelOp]
	{
	\otimes \in \{ \glb, \lub, \flows \} \and
	\Gamma \vdash t_1 : \LabelT \and
	\Gamma \vdash t_2 : \LabelT
	}
	{
	  \Gamma \vdash {t_1 \otimes t_2} : \LabelT
	}
	\and
	\inferrule[getLabel]
	{ }
	{
	  \Gamma \vdash {\getLabel} : \LabelT
	}
	\and
	\inferrule[getClearance]
	{ }
	{
	  \Gamma \vdash {\getClearance} : \LabelT
	}
	\and
	\inferrule[lowerClearance]
	{ \Gamma \vdash t : \LabelT
	}{
	\Gamma \vdash \lowerClr{t} : \LIOT{\unitT}
	}
	\and
	\inferrule[labelOf]
	{
		\Gamma \vdash t : \LabeledT{\tau}
	}
	{
	  \Gamma \vdash \labelOf{t} : \LabelT
	}
	\and
	\inferrule[return]
	{ \Gamma \vdash t : \tau }
	{
	  \Gamma \vdash \return{t} : \LIOT{\tau}
	}
	\and
	\inferrule[bind]
	{ 
		\Gamma \vdash t_1 : \LIOT{\tau_1} \and
		\Gamma \vdash t_2 : \tau_1 \rightarrow \LIOT{\tau_2}
	}
	{
		\Gamma \vdash \bind{t_1}{t_2} : \LIOT{\tau_2}
	}
	\and
	\inferrule[label]
	{ \Gamma \vdash t_1 : \LabelT \and
	\Gamma \vdash t_2 : \tau
	}
	{
	\Gamma \vdash \llabel{t_1}{t_2} : \LabeledT{\tau}
	}
	\and
	\inferrule[unlabel]
	{ \Gamma \vdash t : \LabeledT{\tau} }
	{
	\Gamma \vdash \unlabel{t} : \LIOT{\tau}
	}
	\and
	\inferrule[toLabeled]
	{\Gamma \vdash t_1 : \LabelT \and
	\Gamma \vdash t_2 : \LIOT{\tau}
	}
	{
	\Gamma \vdash \toLabeled{t_1}{t_2} : \LabeledT{\tau}
	}
	\and
	\inferrule[reset]
	{ \Gamma \vdash t : \LIOT{\tau}
	}
	{
	\Gamma \vdash \block{l_1}{l_3}{l_2}{{t}} : \LabeledT{\tau}
	}
	\and
	\inferrule[store]
	{
		\Gamma \vdash t_1 : \gtau' \and
		\Gamma \vdash t_2 : \LabeledT{\gtau}
	}
	{
	\Gamma \vdash \store{t_1}{t_2} : \LIOT{\unitT}
	}
	\vspace{-0.5em}
	\and
	\inferrule[fetch]
	{
		\Gamma \vdash t_1 : \gtau' \and
		\Gamma \vdash t_2 : \LabeledT{\gtau}
	}
	{
		\Gamma \vdash \fetchvar{t_1}{t_2}{\gtau} : \LIOT{(\LabeledT{\gtau})}
	}
}

\newcommand{\figLioEvalContextMain}{}
\newcommand{\figLioEvalContextFull}{
	\begin{align*}
           && E   \Coloneqq~ &
                     [\cdot]
               \ |\  E\ t
               \ |\  \ifel{E}{t}{t}
               \ |\  \bind{E}{t}
               \ |\  \llabel{E}{t}
 \\ && &	   \ |\  \join{E}{t} \ |\ \join{l}{E}
 			   \ |\  \meet{E}{t} \ |\ \join{l}{E}
 			   \ |\  \canFlowTo{E}{t} \ |\ \canFlowTo{l}{E}
 \\ && &              \ |\  \llabel{\gv}{E}
               \ |\  \labelOf{E}
               \ |\  \unlabel{E}
 \\ && &              \ |\  \lowerClr{E}
               \ |\  \toLabeled{E}{t}
               \ |\  \block{l}{l}{l}{E}
 \\ && &       \ |\  \store{E}{t}
               \ |\  \store{\gv}{E}
               \ |\  \fetch{E}{t}
               \ |\  \fetch{\gv}{E}
	\end{align*}
}

\newcommand{\figLioSyntaxMain}{
	\textrm{Ground Value:}
	                       && \gv   \Coloneqq~ &
	                                 \true
	                           \ |\  \false
	                           \ |\  \unit
	                           \ |\   l
	                           \ |\   (\gv,\gv)
	\\
	\textrm{Value:}         && v    \Coloneqq~ &
	                                 \gv
	                           \ |\   (v,v)
	                           \ | \  x
	                           \ |\  \lambda x.t
	                           \ |\  \LIO{t}
	                           \ |\  \LabeledO{l}{v}
	\\
	\textrm{Term:}    &&             t  \Coloneqq~ &
	                                 v
	                           \ |\  (t,t)
	                           \ |\  t\ t
	                           \ |\  \fix{t}
	                           \ |\  \ifel{t}{t}{t}
	                \\&& &
	                         \ |\  \join{t_1}{t_2}
	                         \ |\  \meet{t_1}{t_2}
	                         \ |\  \canFlowTo{t_1}{t_2}
	                \\&& &     \ |\  \return{t}
	                           \ |\  \bind{t}{t}
	                \\&& &     \ |\  \llabel{t}{t}
	                           \ |\  \labelOf{t}
	                           \ |\  \unlabel{t}
	                \\&& &     \ |\  \getLabel
	                           \ |\  \getClearance
	                           \ |\  \lowerClr{t}
	                \\&& &     \ |\  \toLabeled{t}{t}
	                           \ |\  \block{l}{l}{l}{t}
	                \\&& &     \ |\  \store{t}{t}
	                           \ |\  \fetch{t}{t}
	\\
	\textrm{Ground Type:}
	                  && \gtau \Coloneqq~&
	                                 \BoolT
	                           \ |\  \unitT
	                           \ |\  \LabelT
	                           \ |\  (\gtau, \gtau)
	\\
	\textrm{Type:}    && \tau \Coloneqq~&
	                                 \gtau
	                           \ |\  (\tau, \tau)
	                           \ |\  \tau\to\tau
	                           \ |\  \LIOT{\tau}
	                          \  |\  \LabeledT{\tau}
}

\newcommand{\figCLioSyntaxMain}{
	\textrm{Ground Value:}
	                       && \gv   \Coloneqq~ &
	                                 \true
	                           \ |\  \false
	                           \ |\  \unit
	                           \ |\   l
	                           \ |\   (\gv,\gv)
	\\
	\textrm{Value:}         && v    \Coloneqq~ &
	                                 \gv
	                           \ |\   (v,v)
	                           \ | \  x
	                           \ |\  \lambda x.t
	                           \ |\  \CLIO{t}
	                           \ |\  \LabeledO{l}{v}
	\\
	\textrm{Term:}    &&             t  \Coloneqq~ &
	                                 v
	                           \ |\  (t,t)
	                           \ |\  t\ t
	                           \ |\  \fix{t}
	                           \ |\  \ifel{t}{t}{t}
	                \\&& &
	                         \ |\  \join{t_1}{t_2}
	                         \ |\  \meet{t_1}{t_2}
	                         \ |\  \canFlowTo{t_1}{t_2}
	                \\&& &     \ |\  \return{t}
	                           \ |\  \bind{t}{t}
	                \\&& &     \ |\  \llabel{t}{t}
	                           \ |\  \labelOf{t}
	                           \ |\  \unlabel{t}
	                \\&& &     \ |\  \getLabel
	                           \ |\  \getClearance
	                           \ |\  \lowerClr{t}
	                \\&& &     \ |\  \toLabeled{t}{t}
	                           \ |\  \block{l}{l}{l}{t}
	                \\&& &     \ |\  \store{t}{t}
	                           \ |\  \fetch{t}{t}
	\\
	\textrm{Ground Type:}
	                  && \gtau \Coloneqq~&
	                                 \BoolT
	                           \ |\  \unitT
	                           \ |\  \LabelT
	                           \ |\  (\gtau, \gtau)
	\\
	\textrm{Type:}    && \tau \Coloneqq~&
	                                 \gtau
	                           \ |\  (\tau, \tau)
	                           \ |\  \tau\to\tau
	                           \ |\  \CLIOT{\tau}
	                          \  |\  \LabeledT{\tau}
}

\newcommand{\figLioSyntaxFull}{
	\figLioSyntaxMain
}

\newcommand{\figIdealClioMain}{
	\rulename{Internal-Step}\hfill
	\vspace{-1em}\\
	\inferrule{
	  c \lto c'
	}{
	  \iconf{c}{\istore} \ito \iconf{c'}{\istore}
	}
	\vspace{-1em}
	\\
	\rulename{Store}\hfill
	\vspace{-1em}\\
	\inferrule{
	  c \ltol{\storecmd{\gv_k}{\Labeled{~l_1}{\gv~}}} c' \and
	  \istore' = \istore[\gv_k \mapsto \Labeled{l_1}{\gv}]
	}{
	  \iconf{c}{\istore} \ito \iconf{c'}{\istore'}
	}
	\vspace{-1em}
	\\
	\rulename{Fetch-Exists}\hfill
	\vspace{-1em}\\
	\inferrule{
	  c \ltol{\fetchcmd{\gv_k}{\Labeled{~l_1}{\gv~}}} c' \and
	  \Labeled{l_1}{\gv} = \istore(\gv_k)
	}{
	  \iconf{c}{\istore} \ito \iconf{c'}{\istore}
	}
	\vspace{-1em}
	\\
	\rulename{Fetch-Missing}\hfill
	\vspace{-1em}\\
	\inferrule{
	  c \ltol{\missingcmd{\gv_k}} c' \and
	  \istore(\gv_k) = \missing
	}{
	  \iconf{c}{\istore} \ito \iconf{c'}{\istore}
	}
}

\newcommand{\figIdealClioFull}{
	\figIdealClioMain
}

\newcommand{\figRealClioMain}{
	\rulename{Internal-Step}\hfill
	\vspace{-1em}\\
	\inferrule{
		c \lto c'
	}{
		\rconf{c}{\istoredist}{\interactionsdist}{\versions} \rto{1} \rconf{c'}{\istoredist}{\interactionsdist}{\versions}
	}
	\vspace{-1em}
	\\
	\rulename{Store}\hfill
	\vspace{-1em}\\
	\inferrule{
		c \ltol{\storecmd{\gv_k}{\Labeled{~l_1}{\gv~}}} c' \\\\
		\version = \increment{\versions(\gv_k)} \and
		\versions' = \versions[\gv_k \mapsto \version] \\\\
		\begin{array}{l<{\hspace{-2mm}}c<{\hspace{-2mm}}l<{\hspace{-2mm}}l}
		\interactionsdist' &=& \randomExprL ~ & \concat{\concat{\storecmd{\gv_k}{\Labeled{l_1}{\bitstream}}}{\interactionsprime}}{\interactions} ~ ~ \Big|~ ~
			\interactions \drawnFrom \interactionsdist;  \arcr
				&&&
				(\interactionsprime, \Labeled{l_1}{\bitstream}) \drawnFrom \serialize{\istore}{\Labeled{l_1}{(\gv, \gv_k, \version)}}
				 ~ \randomExprR
		\end{array} 
	}{
		\rconf{c}{\istoredist}{\interactionsdist}{\versions} \rto{1} 
			\rconf{c'}{\istoredist}{\interactionsdist'}{\versions'}
	}
	\vspace{-1em}
	\\
	\rulename{Fetch-Exists}\hfill
	\vspace{-1em}\\
	\inferrule{
		c \ltol{\fetchcmd{\gv_k}{\Labeled{~l_1}{\gv~}}} c' \and \version \not< \versions(\gv_k)\\\\
		(\istore, p) \in \randomExprL ~\interactions(\emptyset) ~|~ \interactions \drawnFrom \interactionsdist ~ \randomExprR \and p > 0 
		\\\\
		\Labeled{l_1}{(\gv, \gv_k, \version)} = \deserialize{\istore}{\istore(\gv_k)}
	}{
		\rconf{c}{\istoredist}{\interactionsdist}{\versions} \rto{p}
			\rconf{c}{\istoredist}{\interactionsdist}{\versions}
	}
	\vspace{-1em}
	\\
	\rulename{Fetch-Missing}\hfill
	\vspace{-1em}\\
	\inferrule{
		c \ltol{\missingcmd{\gv_k}} c' \\\\
		(\istore, p) \in \randomExprL ~\interactions(\emptyset) ~|~ \interactions \drawnFrom \interactionsdist ~ \randomExprR \and p > 0 \\\\
			\deserialize{\istore}{\istore(\gv_k)} \mathrm{~undefined}
	}{
		\rconf{c}{\istoredist}{\interactionsdist}{\versions} \rto{p}
			\rconf{c}{\istoredist}{\interactionsdist}{\versions}
	}
	\vspace{-1em}
	\\
	\rulename{Fetch-Replay}\hfill
	\vspace{-1em}\\
	\inferrule{
		c \ltol{\missingcmd{\gv_k}} c' \and 	\gv_k' \neq \gv_k \textor \version < \versions(\gv_k)\\\\
		(\istore, p) \in \randomExprL ~\interactions(\emptyset) ~|~ \interactions \drawnFrom \interactionsdist ~ \randomExprR \and p > 0 \\\\
		\Labeled{l_1}{(\gv, \gv_k', \version)} = \deserialize{\istore}{\istore(\gv_k)} 
	}{
		\rconf{c}{\istoredist}{\interactionsdist}{\versions} \rto{p}
			\rconf{c}{\istoredist}{\interactionsdist}{\versions}
	}
}

\newcommand{\figRealClioFull}{
	\figRealClioMain
}

\newcommand{\figRealClioLowStep}{
	\rulename{Low-Step}\hfill
	\vspace{-1em}\\
	\inferrule{
	\interactionsdist' = \randomExprL ~ \concat{\interactions_A}{\interactions} ~|~ \interactions_A \drawnFrom \interactionsdist_A;~ \interactions \drawnFrom \interactionsdist ~ \randomExprR \\\\
	  \rconf{c}{\istoredist'}{{\interactionsdist'}}{\versions} 
	  	\rto{p}
	  \rconf{c'}{\istoredist''}{\interactionsdist''}{\versions'} \\\\
	  \canFlowTo{\pcOf{c}}{\projC{\lowlabel}} \and
	  \canFlowTo{\pcOf{c'}}{\projC{\lowlabel}} 
	}{ 
		(\rconf{c}{\istoredist_0}{\interactionsdist}{\versions}, \interactionsdist_A)
			\Lowstep{p}
		\rconf{c'}{\istoredist''}{\interactionsdist''}{\versions'}
	}
	\vspace{-1em}
	\\
	\rulename{Low-to-High-to-Low-Step}\hfill
	\vspace{-1em}\\
	\inferrule
	{
	  \interactionsdist' = \randomExprL ~ \concat{\interactions_A}{\interactions} ~|~ \interactions_A \drawnFrom \interactionsdist_A;~ \interactions \drawnFrom \interactionsdist ~ \randomExprR \\\\
	  \rconf{c}{\istoredist'}{\interactionsdist'}{\versions} 
	  	\rto{p_0}
	  \rconf{c_0}{\istoredist_0}{\interactionsdist_0}{\versions)} \\\\
	  \rconf{c_0}{\istoredist_0}{\interactionsdist_0}{\versions_0} 
	  	\rto{p_1} ... \rto{p_j}
	  \rconf{c_j}{\istoredist_j}{\interactionsdist_j}{\versions_j} \\\\
	  \forall_{0 \leq i < j}.~\cantFlowTo{\pcOf{c_i}}{\lowlabel} \and
	  \canFlowTo{\pcOf{c_j}}{\lowlabel} \and
	  p = \Pi_{0 \leq i \leq j}\ p_i
	}{
	  (\rconf{c}{\istoredist}{\interactionsdist}{\versions}, \interactionsdist_A) 
			\Lowstep{p}
	  \rconf{c_j}{\istoredist_j}{\interactionsdist_j}{\versions_j}
	}
}

\newcommand{\figIdealClioLowStep}{
	\inferrule[Low-Step]
	{
	  \iconf{c}{\iinteractionsvar{\iinteraction}(\istore)}
	  	\ito
	  \iconf{c'}{\istore'} \\\\
	  \canFlowTo{\pcOf{c}}{\lowlabel} \and
	  \canFlowTo{\pcOf{c'}}{\lowlabel}
	}{ 
		(\iconf{c}{\istore}, \iinteractionsvar{\iinteraction})
			\Lowstep{}
		\iconf{c'}{\istore'}
	}
        \hfill
	\inferrule[Low-to-High-to-Low-Step]
	{
	  \iconf{c}{\iinteractionsvar{\iinteraction}(\istore)}
	  	\ito
	  \iconf{c_0}{\istore_0} \\\\
	  \iconf{c_0}{\istore_0}
	  	\ito ... \ito
	  \iconf{c_j}{\istore_j} \\\\
	  \forall_{0 \leq i < j}.~\cantFlowTo{\pcOf{c_i}}{\lowlabel} \and
	  \canFlowTo{\pcOf{c_j}}{\lowlabel}
	}{
	  (\iconf{c}{\istore}, \iinteractionsvar{\iinteraction})
			\Lowstep{}
	  \iconf{c_j}{\istore_j}
	}
}

\newcommand{\figRealClioInteractions}{
	\begin{align*}
		&\textrm{Interaction:} &
		 & \interaction & \Coloneqq~ & 
			\skipcmd = \lambda\istore.~ \istore 
		\\&&&& \ | \ & \storecmd{\gv_k}{\Labeled{l_1}{\bitstream}} = \lambda\istore.~ \istore[\gv_k \mapsto \Labeled{l_1}{\bitstream}] 
		\\&&&& \ | \ & \storecmd{\category}{\categorykey} = \lambda\istore.~ \istore[\category \mapsto \categorykey] 
	\end{align*}
}

\newcommand{\figIdealClioInteractions}{
	\begin{alignat*}{2}
		  &
		   \iinteractions~ 
	   		& \Coloneqq~~ & \concat{\iinteraction}{\iinteractions} \ ~|~\ \iinteraction\\
		  & 
		 \iinteraction~ & \Coloneqq~~ & 
			\skipcmd=\lambda\istore. \istore 
		\\&& \ |~~ \ & \storecmd{\gv'}{\Labeled{l_1}{\gv}} = \lambda\istore.~ \istore[\gv' \mapsto \Labeled{l_1}{\gv}] \textrm{~s.t.~} \canFlowToI{\projI{\lowlabel}}{\projI{l_1}}
		\\&& \ |~~ \ & \corruptcmd{\gv_1, ... \gv_n}=\lambda\istore. \istore[\gv_1 \mapsto \missing; ... ~ \gv_n \mapsto \missing] 
	\end{alignat*}
}

\newcommand{\figStepFunctionDefLineA}{%
  \ensuremath{\multistep{\keystore}{c_0}{\strategy}{1} \hspace{4mm}  = \hspace{1mm} \Big\{
	  \big(\rconf{c_1}{\istoredist}{\interactionsdist_1}{\versions_1}, p_0 \cdot p_1\big)
	  ~\Big|~}
}
\newcommand{\figStepFunctionDefLineAa}{%
 \ensuremath{\hspace*{0.4cm}
	  (\rconf{c_0}{\{(\emptyset,1)\}}{ \{ (\skipcmd, 1) \} }{\zeroes}, \strategy(\{ (\skipcmd, 1) \}))
	  	\Lowstep{p_1}
	  \rconf{c_1}{\istoredist}{\interactionsdist_1}{\versions_1} ~
          \Big\}
        }
}
\newcommand{\figStepFunctionDefLineB}{%
  \ensuremath{
    \multistep{\keystore}{c_0}{\strategy}{j+1} = \hspace{1mm} \Big\{
    \big(\rconf{c_2}{\istoredist_2}{\interactionsdist_2}{\versions_2}, p_0 \cdot p_1 \big)
    ~\Big|~
    }
}
\newcommand{\figStepFunctionDefLineBa}{%
 \ensuremath{\hspace*{0.4cm}
   \big(\rconf{c_1}{\istoredist_1}{\interactionsdist_1}{\versions_1}, p_0\big) \in \multistep{\keystore}{c_0}{\strategy}{j};}
}
\newcommand{\figStepFunctionDefLineBb}{%
 \ensuremath{\hspace*{0.4cm}
	  (\rconf{c_1}{\istoredist_1}{\interactionsdist_1}{\versions_1}, \strategy(\interactionsdist_1))
	  	\Lowstep{p_1}
	  \rconf{c_2}{\istoredist_2}{\interactionsdist_2}{\versions_2} ~
          \Big\}}
}

\newcommand{\figStepFunctionDef}{%
  \ensuremath{
    \begin{array}{l}
      \figStepFunctionDefLineA\figStepFunctionDefLineAa \\
      \figStepFunctionDefLineB\figStepFunctionDefLineBa\figStepFunctionDefLineBb
    \end{array}
    }
}
\newcommand{\figStepFunctionDefNL}{%
	\ensuremath{\begin{array}{l}
         \figStepFunctionDefLineA \\
	 \figStepFunctionDefLineAa \\
         \figStepFunctionDefLineB \\
	 \figStepFunctionDefLineBa \\
         \figStepFunctionDefLineBb
        \end{array}}
}

\newcommand{\figCTAGame}{
	\figCTAGameVar{IND}{\interactionsvar{\interaction_b}}
}

\newcommand{\figCTAGameVar}[2]{
	\begin{array}{l}
	\hspace*{0cm}  \INDvar{b}{\keystore_0}{\adversary}{\principals}{j}{n}{#1} = \\
	\hspace*{0.4cm} \keystore' \drawnFrom \Gen(\principals, 1^n);  ~ 
	\keystore = \keystore_0 \uplus \keystore'; \\
	\hspace*{0.4cm} t, v_0, v_1, \strategy, \adversary_2 \drawnFrom 
	         \adversary(\pubkeys{\keystore})
	         	\text{~~such~that~~} v_0 \lowEquivC v_1 \\
	\hspace*{3.4cm} \text{~and~} \vdash t : \LabeledT{\tau} \rightarrow \LIOT{\tau'} \\
	\hspace*{3.4cm} \text{~and~} \vdash v_0 : \LabeledT{\tau}  \\
	\hspace*{3.4cm} \text{~and~} \vdash v_1 : \LabeledT{\tau}  \\
	\hspace*{3.4cm} \text{~and~} \lowlabel = \authorityOf{\keystore_0}; \\
	\hspace*{0.4cm} \rconf{c}{\istoredist}{\interactionsdist_b}{\versions'} \drawnFrom \multistep{\keystore}{\conf{\Start{\keystore}}{\Clr{\keystore}}{(t\ v_b)}}{\strategy}{j}; \\
	\hspace*{0.4cm} \interactionsvar{\interaction_b} \drawnFrom \interactionsdist_b; ~ 
	\mathrm{Output}~\adversary_2({#2})
	\end{array}
}

\newcommand{\figCPAGame}{
	\begin{array}{r@{\;}l}
	\hspace*{0cm} \mathrm{IND}_b(\mathcal{A}, n) =
	&(pk, sk) \drawnFrom \Gen(1^n); \\
	& m_0, m_1, \mathcal{A}_2 \drawnFrom \mathcal{A}(pk) \text{ s.t. } |m_0| = |m_1|;\\
	& c \drawnFrom \Enc(pk, m_b);\\
	& \mathrm{Output~} \mathcal{A}_2(c)
	\end{array}
}

\newcommand{\figForgeryProb}{
	\begin{array}{l}
	\textbf{Pr}\Big[ ~ 
		\Labeled{l_1}{b} \in \mathrm{Values}_{\keystore}(\interactionsprime)
		\textand
		\Labeled{l_1}{b} \not\in \mathrm{Values}_{\keystore}(\interactions)
	\\ \hspace*{0.1cm} \Big|  
		\hspace*{0.3cm} 
		\keystore' \drawnFrom \Gen(\{p\}, 1^n); ~ 
		\keystore = \keystore_0 \uplus \keystore'; \\
		\hspace*{0.5cm} t, \strategy, \adversary_2 \drawnFrom \adversary(\pubkeys{\keystore});\\
		\hspace*{0.5cm} \rconf{c}{\istoredist}{\interactionsdist}{\versions} \drawnFrom \multistep{\keystore}{\conf{\Start{\keystore}}{\Clr{\keystore}}{t}}{\strategy}{j}; \\
		\hspace*{0.5cm}\interactions \drawnFrom \interactionsdist; \\
		\hspace*{0.5cm} 
		t', \strategy' \drawnFrom \adversary_2(\interactions);\\
		\hspace*{0.5cm} \rconf{c'}{\istoredist'}{\interactionsdist'}{\versions'} \drawnFrom \multistep{\keystore}{\conf{\Start{\keystore_0}}{\Clr{\keystore}}{t'}}{\strategy'}{j}; \\ 
		\hspace*{0.5cm}\interactionsprime \drawnFrom \interactionsdist' 
	\Big]
	\end{array}
}

\newcommand{\figInteractionLemmaIH}{
}

\newcommand{\figInteractionProbStatement}[2]{
	$ 
	\hspace*{1cm} \mathbf{Pr}\Big[  
		\rconf{c_1'}{\istoredist_1}{\interactionsdist_1}{\versions_1}
		\not\lowEquivC
		\rconf{c_2'}{\istoredist_2}{\interactionsdist_2}{\versions_2}
		\textor
		\versions_1 \neq \versions_2
		{#2}
	~~\Big| \\ \hspace*{2cm}
		 ~\keystore \drawnFrom \Gen(\principals, 1^n); ~
    	\keystore_1 = \keystore_0 \uplus \keystore; ~
		\rconf{c_1'}{\istoredist_1}{\interactionsdist_1}{\versions_1} \drawnFrom \multistep{\keystore_1}{c_1}{\strategy}{#1}; \\
	\hspace*{2cm}
		\keystore' \drawnFrom \Gen(\principals, 1^n); ~
    	\keystore_2 = \keystore_0 \uplus \keystore'; ~
		\rconf{c_2'}{\istoredist_2}{\interactionsdist_2}{\versions_2} \drawnFrom \multistep{\keystore_2}{c_2}{\strategy}{#1}
	\Big]$ \\ \\
	\hspace*{1cm} is negligible in $n$, and 
	\\ \\
	$\hspace*{2cm} \Big\{  ~~ \interactions_1 ~~ \Big| ~~
	 \keystore' \drawnFrom \Gen(\principals, 1^n); ~
	 \keystore = \keystore_0 \uplus \keystore'; ~
	 \rconf{c_1'}{\istoredist_1}{\interactionsdist_1}{\versions_1} \drawnFrom \multistep{\keystore}{c_1}{\strategy}{#1}; ~ \interactions_1 \drawnFrom \interactionsdist_1 ~ \Big\}_n \\
		\hspace*{7cm} \asymp \\
		\hspace*{2cm} \Big\{  ~~ \interactions_2 ~~ \Big| ~~
	 \keystore' \drawnFrom \Gen(\principals, 1^n);~ 
	 \keystore = \keystore_0 \uplus \keystore';~
	 \rconf{c_2'}{\istoredist_2}{\interactionsdist_2}{\versions_2} \drawnFrom \multistep{\keystore}{c_2}{\strategy}{#1}; ~ \interactions_2 \drawnFrom \interactionsdist_2 ~ \Big\}_n $
}

\newcommand{\figInteractionLemma}[1]{
	all keystores $\keystore_0$ where $\lowlabel = \authorityOf{\keystore_0}$, \lioS configurations $c_1, c_2$, strategies $\strategy$, and principals $\principals$, and $j \in \nat$, if $\cryptosys$ is CPA Secure and $c_1 \lowEquivC c_2$, then

	\medskip

	\noindent \figInteractionProbStatement{j}{{#1}}
}

\makeatletter 
\def\arcr{\@arraycr}
\makeatother

\newif\ifsoundness
\soundnessfalse

\newcommand{\todo}[1]{}

\newenvironment{CompactItemize}%
  {\begin{list}{$\ \ \bullet$}%
   {\leftmargin=9pt \itemsep=2pt \topsep=2pt
     \parsep=0pt \partopsep=0pt}}%
  {\end{list}}

\copyrightyear{2017}
\setcopyright{rightsretained}
\setcopyright{none}
\acmYear{2017}
\acmConference{CCS '17}{October 30-November 3, 2017}{Dallas, TX, USA}
\acmDOI{10.1145/3133956.3134036}
\acmISBN{978-1-4503-4946-8/17/10}

\renewcommand\footnotetextcopyrightpermission[1]{} 
\begin{document}
\title{
  Cryptographically Secure Information Flow Control on Key-Value Stores}
\titlenote{The conference version of this paper appears in CCS 2017~\cite{clio:conf}.}

\author{Lucas Waye}
\affiliation{%
  \institution{Harvard University}
  \city{Cambridge}
  \state{Massachusetts}
}
\email{lwaye@seas.harvard.edu}

\author{Pablo Buiras}
\affiliation{%
  \institution{Harvard University}
  \city{Cambridge}
  \state{Massachusetts}
}
\email{pbuiras@seas.harvard.edu}

\author{Owen Arden}
\affiliation{%
  \institution{University of California, Santa Cruz$^\dagger$\authornote{Work done while author was at Harvard University.}}
  \city{Santa Cruz}
  \state{California}
}
\email{owen@soe.ucsc.edu}

\author{Alejandro Russo}
\affiliation{%
  \institution{Chalmers University of Technology}
  \city{Gothenburg}
  \country{Sweden}
}
\email{russo@chalmers.se}

\author{Stephen Chong}
\affiliation{%
  \institution{Harvard University}
  \city{Cambridge}
  \state{Massachusetts}
}
\email{chong@seas.harvard.edu}

\begin{abstract}

%
We present \clio, an information flow control (IFC)
system that transparently incorporates cryptography
to enforce confidentiality and integrity policies on untrusted storage.
\clio insulates developers from explicitly manipulating keys and cryptographic
primitives by leveraging the policy language of the IFC system to
automatically use the appropriate keys and correct cryptographic operations.
We prove that \clio is secure with a novel proof technique that is
based on a proof style from cryptography together with standard programming 
languages results.
We present a prototype \clio implementation and a case study
that demonstrates \clio's practicality.

\end{abstract}


\maketitle

\section{Introduction}
\label{sec:intro}

Cryptography is critical for applications that securely store
and transmit data.
It enables the authentication of remote hosts, authorization of
privileged operations, and the preservation of confidentiality and integrity of
data.
However, applying cryptography is a subtle task, often involving setting up
configuration options and low-level details that users must get
right; even small mistakes can lead to major vulnerabilities~\cite{art1,art2}.
A common approach to address this problem is to raise the level of
abstraction.
For example, many libraries provide high-level interfaces for establishing
TLS~\cite{rfc5246} network connections (e.g.,
OpenSSL\footnote{\url{https://www.openssl.org/}}) that are very similar to the
interfaces for establishing unencrypted connections.
These libraries are useful (and popular) because they abstract many
configuration details, but they also make several assumptions about certificate
authorities, valid protocols, and client authentication.
Due in part to these assumptions, the interfaces
are designed for experienced cryptography programmers and as a result can be
used incorrectly by non-experts in spite of their high level of abstraction~\cite{Whitten}.
Indeed, crypto library misuse is a more prevelant security issue than
Cross-Site Scripting (XSS) and SQL Injection~\cite{veracode}.

Information flow control (IFC) is an attractive approach to building
secure applications because it addresses some of these issues.
There has been extensive work in developing expressive information flow policy
languages~\cite{conf/sp/MyersL98, stefan:2011:dclabels, FLAM} that help clarify
a programmer's intent.
Furthermore, many semantic guarantees offered by IFC languages are inherently
compositional from a security point of view
\cite{Goguen:Meseguer:Noninterference, robdecl}.
However, existing IFC languages (e.g., \cite{jif,
stefan:2011:flexible,Russo:2015:FP, hedin2014jsflow, Yip:2009:PBS,
DeGroef:2012,stefan:2014:protecting}) generally assume that critical components
of the system, such as persistent storage, are trustworthy---the components must enforce
the policies specified by the language abstraction.
This assumption makes most IFC systems a poor fit for many of the use-cases
that cryptographic mechanisms are designed for.

It is tempting to extend IFC guarantees to work with untrustworthy data
storage by simply ``plugging-in'' cryptography.
However, the task is not simple: the threat model of an IFC system extended
with cryptography differs from both
the standard cryptographic threat models and from standard IFC threat models.
Unlike most IFC security models, an attacker in this scenario may have
low-level abilities to access signatures and ciphertexts of sensitive data, and the ability to deny access to data by corrupting it (e.g.,
flipping bits in ciphertexts).

Attackers also have indirect access to the private cryptographic keys
through the trusted runtime.  An attacker may craft and run programs
that have access to the system's cryptographic keys in order to trick
the system into inappropriately decrypting or signing information.
Cryptographic security models often
account for the high-level actions of attackers using \emph{oracles} that
mediate what information an active attacker can learn through
interactions with the cryptosystem.  These oracles abstractly represent implementation
artifacts that could be used by the attacker to distinguish ciphertexts.  Ensuring that an
actual implementation constrains its behavior to that modeled by an oracle
is typically left to developers.

An attacker's actual interactions with a system often extend beyond
the semantics of specific cryptographic primitives and into
application-specific runtime behavior such as how a server responds
when a message fails to decrypt or a signature cannot be verified.  If
an attacker can distinguish this behavior, it may provide them with
information about secrets.  Building real
implementations that provide no additional information to attackers
beyond that permitted by the security model can be very challenging.

Therefore, to give developers better tools for building secure applications, we
need to ensure that system security is not violated by
combining attackers' low-level abilities and their ability to craft
their own programs. This requires extending the attacker's power beyond that
typically considered by IFC models, and representing the attacker's interactions with
the system more precisely than typical cryptographic security models.

This paper presents \clio,
a programming language
that reconciles IFC and cryptography models to provide guarantees on both
ephemeral data within \clio applications and persistent data on an untrusted
key-value store.
\clio extends the IFC-tool LIO~\cite{stefan:2011:flexible} with \emph{store}
and \emph{fetch} operations for interacting with a persistent key-value store.
Like LIO, \clio expresses confidentiality and integrity requirements using
\emph{security labels}: flows of information are
controlled throughout the execution of programs to ensure the policies
represented by the labels are enforced.
\clio encrypts and signs data as it leaves the \clio runtime,
and decrypts and verifies as it enters the system.
These operations are done automatically according to the security
labels---thus avoiding both the mishandling of sensitive data and the
misuse of cryptographic mechanisms.  Because the behavior of the
system is fully specified by the semantics of the \clio language, an
attacker's interactions with the system can be characterized
precisely. This results in a strong connection between the power of
the attacker in our formal security model and in actual \clio
programs.

\clio transparently maps security labels to cryptographic keys and leverages
the underlying IFC mechanisms to ensure that keys are not misused within the program.
Since we consider attackers capable of denying access to information
by corrupting data, \clio
extends \lioS
labels with an availability policy that tracks who can deny
access to information (i.e., who may corrupt the data).

Figure~\ref{fig:threat-model} presents an overview of the \clio threat
model.  At a high-level, a \clio program may be a malicious program written by the attacker.
All interactions between the runtime and the
store are visible to the attacker.  Only the (trusted) \clio runtime
has access to the keys used to protect information from the attacker,
but the attacker may have access to other ``low'' keys.  The \clio
runtime never exposes keys directly to program code: they are only used
implicitly to protect or verify data as it leaves or enters the \clio
runtime.

Attackers may also perform low-level
\texttt{fetch} and \texttt{store} operations directly on the key-value
store. Using these low-level operations, an attacker may corrupt ciphertexts to
make them invalid even when it does not possess the signing keys to make valid modifications.
We treat these actions as attacks on the \emph{availability} of data, rather than
on its integrity.
A low-availability store is vulnerable to availability
attacks, and thus should be prevented from storing data that requires high-availability.
\clio's information flow control mechanisms mediate the attacker's ability
to discover new information or modify signed values by interacting
with a \clio program through \texttt{fetch}s and \texttt{store}s to a
\clio store.

\begin{figure}
\begin{center}
\includegraphics[width=20em]{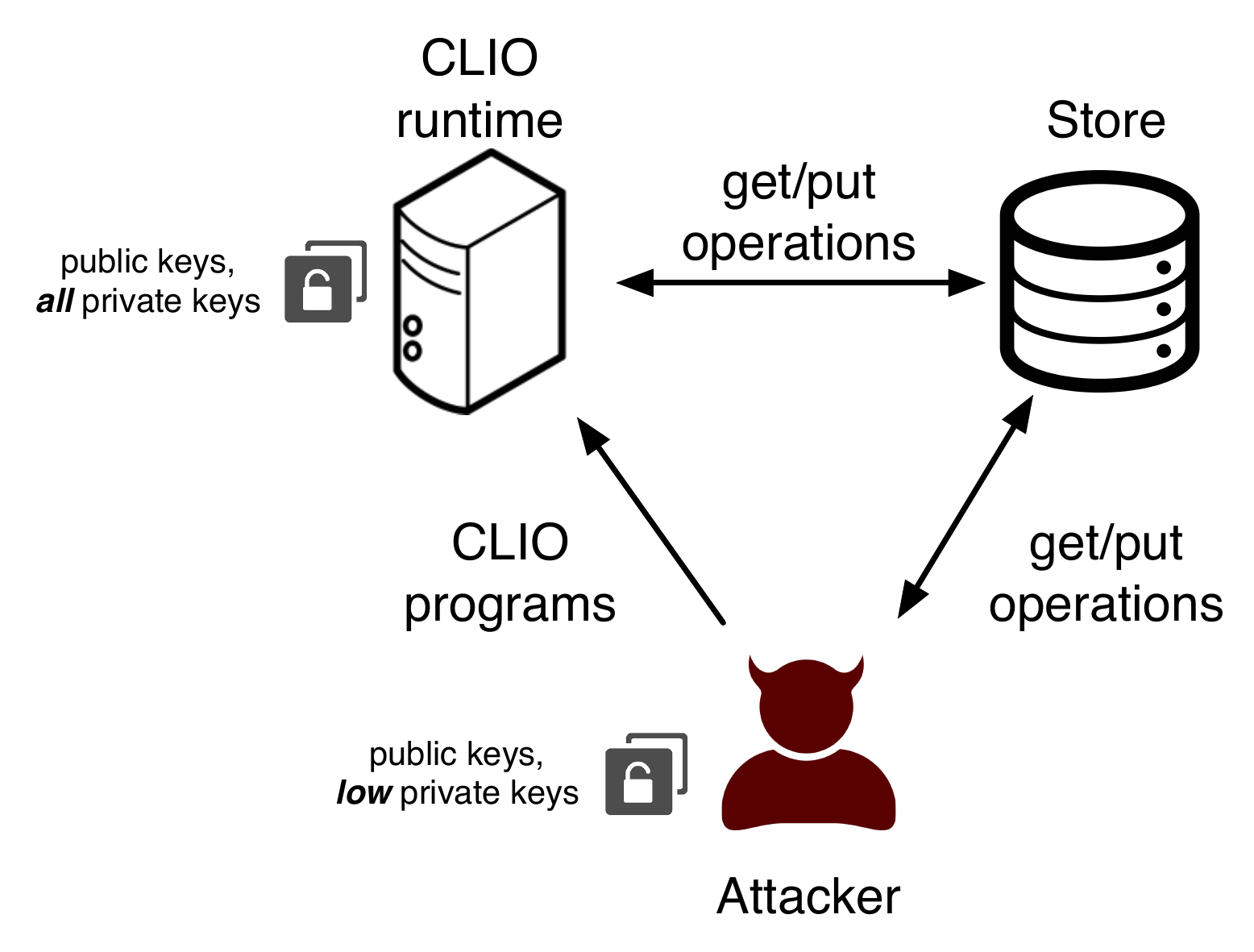}
\end{center}
\caption{\clio threat model. Attackers write \clio programs, read from and write to
the store, and observe the runtime's interactions.}
\label{fig:threat-model}
\end{figure}

This paper makes the following contributions:
\begin{CompactItemize}
\item A formalization of the \emph{ideal} semantics of \clio,
which models its security without
cryptography, and a \emph{real} semantics, which enforces
security cryptographically.
\item A novel proof technique that combines standard programming language
and cryptographic proof techniques.
Using this approach, we characterize the interaction between the
high-level security guarantees provided by information flow control
and the low-level guarantees offered by the cryptographic mechanisms.
\item For confidentiality, we have formalized these guarantees as \emph{chosen-term attack} (CTA)
security, an extension of \emph{chosen-plaintext attack} (CPA) security to systems where an
attacker may choose arbitrary \clio programs that encrypt
and decrypt information through the \clio runtime.  Though CTA security is predicated on the relatively weak guarantees of CPA
crypto primitives, CTA security provides stronger guarantees since it applies to the end-to-end flow of
information through the system, including the interactions an active, adaptive attacker might use to distinguish
ciphertexts.
\item For integrity, we have defined \emph{leveraged existential forgery}, an
extension of \emph{existential forgery} to systems where an attacker may choose and
execute a program to produce signed values.

\item  A prototype \clio implementation in the form of a Haskell library extending LIO.
Our prototype system employs the DC~labels model~\cite{stefan:2011:dclabels},
previously used in practical systems (e.g., Hails~\cite{giffin:2012:hails} and
COWL~\cite{stefan:2014:protecting}).
%
Our implementation extends DC-labels with an
availability component, which may be applicable to these existing
systems as well.
\end{CompactItemize}

Our approach uses a computational model of cryptography.
However, we do not rely on a formal definition
of computational noninterference~\cite{Laud:2001}.
Instead, we phrase security in terms
of an adversary-based game with a definition much closer to standard
cryptographic definitions of security such as CPA security~\cite{pass-shelat}. This
approach helps {to} model an active adversary on the store, something that
computational noninterference can not easily capture. Furthermore, we
incorporate the semantics of \clio programs and potential attacks
against them into the security model. This approach captures the power
of the attacker more precisely than cryptographic models for active
attackers like chosen-ciphertext attack (CCA)
security~\cite{pass-shelat}.

Our CTA model applies a game-based definition of security in a
language setting and is a novel aspect of this work. Computational
noninterference and related approaches consider attackers that can
only provide different secret inputs to the program. Thus a key
contribution of our work is capturing the abilities of an active
attacker (that can both supply code to execute and directly manipulate
the store) in a crypto-style game that goes beyond CPA security
 and standard IFC guarantees (noninterference, including
computational noninterference). Although our results are specific to
\clio, we expect our approach to be useful in proving the security of
cryptographic extensions of other information flow languages.

The rest of the paper is structured as follows.
Section~\ref{sec:background} introduces \lioS and
Section~\ref{sec:store-interaction} describes the extensions to it in order to
interact with an untrusted store. 
Section~\ref{sec:design} describes the computational model of
\clio with cryptography, and Section~\ref{sec:eval} shows the model's
formal security properties. Section~\ref{sec:impl} describes
the prototype implementation of \clio along with a case study.
And finally Section~\ref{sec:related} discusses related work.
and Section~\ref{sec:conclusion} concludes.

\section{Background}
\label{sec:background}


In this section, we describe the programming model of \clio. 
\clio is based on
LIO \cite{stefan:2011:flexible}, a dynamic IFC library 
implemented in Haskell.


%
%


LIO uses Haskell features 
to control how sensitive information is used and to restrict I/O
side-effects. In particular, it implements an embedded language and a
runtime monitor based on the notion of a \emph{monad}, an abstract
data type that represents sequences of actions (also known as
\emph{computations}) that may perform side-effects.
%
The basic interface of a monad consists in the fundamental operations
\textsf{return} and $(\bindSymbol)$ (read as ``bind'').
The expression $\return{x}$ denotes a computation that returns
the value denoted by $x$, performing no side-effects.
The function $(\bindSymbol)$ is used to \emph{sequence} 
computations.
Specifically, $\bind{t}{\lambda x.t'}$ takes the result produced by
$t$ and applies function $\lambda x.t'$ to it (which allows
computation $t'$ to depend on the value produced by $t$).  In order to
be useful, monads are usually extended with additional primitive
operations to selectively allow the desired side-effects. The \lio{} monad is a
specific instance of this pattern equipped with IFC-aware operations
that enforce security.
%



LIO, like many dynamic IFC approaches (e.g., \cite{Zeldovich:2008,
Roy:2009, Cheng2012}), employs a \emph{floating label}.
Security concerns are represented by labels
which form a \emph{lattice}, a partially-ordered ($\sqsubseteq$) set with
least upper bounds ($\sqcup$) and greatest lower bounds ($\sqcap$).
A runtime monitor maintains as part of its state a distinguished
label $\lcur$ known as the \emph{current label}.
The current label is similar to the program
counter ($\mathit{pc}$) label of static IFC systems (e.g.,~\cite{jif,FlowCaml}):
it restricts side-effects in the current
computation that may compromise the
confidentiality or integrity of data.
For example, a computation whose current label is secret cannot write to a public location.
LIO operations adjust this label when sensitive information enters the program
and use it to validate (or reject) outgoing flows.

When an LIO computation with current label $\lcur$ observes an
entity with label $l$, its current label is increased (if necessary)
to the least upper bound of the two labels, written
$\lcur \lub l$. Thus, the current label
``floats up'' in the security lattice, so that it is always
an upper bound on the security levels of information in the computation.  Similarly, before performing a side-effect
visible to label $l$, LIO ensures the current label flows to
$l$ ($\lcur \flows l$). 

Once the current label within a given computation is raised, it can
never be lowered. This can be very restrictive, since, for example, as
soon as confidential data is accessed by a computation, the
computation will be unable to output any public data. To address this
limitation, the $\textsf{toLabeled}$
operation allows evaluation of an LIO computation $m$ in a separate
\emph{compartment}: $\toLabeled{l}{m}$ will run $m$ to completion,
and produce a \emph{labeled value} $\LabeledO{l}{v}$, where $v$ is the result of
computation $m$, and $l$ is an over-approximation of the final current
label of $m$. Note that the current label of the enclosing computation
is not affected by executing $\toLabeled{l}{m}$.  In general, given a
labeled value $\LabeledO{l}{v}$, label $l$ is an upper bound on the information
conveyed by $v$. Labeled values can also be created from raw values
using operation \textsf{label}, and a labeled value can be read into
the current scope with operation \textsf{unlabel}.
Creating a
labeled value with label $l$ can be regarded as writing into a channel
at security level $l$. Similarly, observing (i.e., unlabeling) a
labeled value at $l$ is analogous to reading from a channel at $l$.


\emph{\lioS security guarantees.}
\label{sec:lio-security}
%
LIO provides a termination-insensitive \emph{noninterference-based}
security guarantee~\cite{Goguen:Meseguer:Noninterference}.
Intuitively, if a program is noninterfering with respect to
confidentiality, then the public outputs of a program reveal nothing
about the confidential inputs. More precisely, an attacker $\attacker$
that can observe inputs and outputs with confidentiality label at most
$l_\attacker$ learns nothing about any input to the program with label
$l$ such that $l \not\sqsubseteq l_\attacker$.  Similarly, a program
is noninterfering for integrity if an attacker that can control
untrusted inputs cannot influence trusted outputs.


\subsection{\clio}


\paragraph{\clio calculus}

\begin{figure}[t] 
\resizebox{\columnwidth}{!}{
\begin{minipage}{\columnwidth}
\begin{align*}
\figCLioSyntaxMain
\end{align*}
\end{minipage}
}
\caption{Syntax for \clio values, terms, and types.\label{fig:language}
}
\end{figure}

\clio{} is formalized as a typed $\lambda$-calculus with call-by-name
evaluation, in the same style as LIO~\cite{stefan:2011:flexible}.
Figure~\ref{fig:language} gives the syntax of \clio values, terms, and types.
In addition to standard $\lambda$-calculus features, \clio includes several
security-related extensions that mirror those in LIO, and two
operations for interacting with the key-value store, namely
$\texttt{store}$ and $\texttt{fetch}$. As those primitives have nontrivial
semantics that involve the external storage, we defer their discussion
to Section~\ref{sec:store-interaction}.
Security labels have type $\LabelT$ and labeled values have type $\LabeledT{\tau}$.
Computation on labeled values occur in the $\CLIOT$ monad using the $\return{}$ and
$(\bindSymbol)$ monadic operators.
The nonterminals $\CLIO{t}$ and
$\block{l}{l'}{l''}{t}$ are only generated by intermediate reduction steps and are
not valid source-level syntax.
For convenience, we also distinguish values that can be easily
serialized as \emph{ground values}, $\gv$. Ground values are all values
except functions and \clio{} computations. To facilitate our extension of LIO with cryptography, we require
labeled values to contain only ground values. 

Static type checking is performed in the standard way. 
We elide the typing rules $\vdash t : \tau$ since they are mostly
standard\footnote{Complete definitions given in \Cref{sec:typing}.}. 
LIO enforces information flow control dynamically, so it
does not rely on its type system to provide security guarantees.


The semantics is given by a small-step reduction relation $\lto$ over
\clio configurations (Figure~\ref{fig:sos-lio})\footnote{The rest can be found in \Cref{sec:fullsem}.}. Configurations are of the form
$\conf{\lcur}{\lclr}{t}$, where $\lcur$ is the current label and $t$
is the \clio{} term being evaluated.  Label $\lclr$ is the \emph{current
clearance} and is an upper bound on the current label $\lcur$. The
clearance allows a programmer to specify an upper bound for
information that a computation is allowed to access. We write $c \lto
c'$ to express that configuration $c$ can take a reduction step to
configuration $c'$. We define $\lto^*$ as the reflexive and transitive
closure of $\lto$. Given configuration $c=\conf{\lcur}{\lclr}{t}$ we write $\pcOf{c}$ for $\lcur$, the current label of $c$.

\begin{figure}[t] 
\begin{center}
\resizebox{.95\columnwidth}{!}{
\begin{minipage}{\columnwidth}
\begin{center}
\figLioEvalContextMain
\vspace{-1em}
\begin{mathpar}
\figClioSemanticsBaseMain
\end{mathpar}
\end{center}
\end{minipage}
}
\end{center}
\caption{\clio language semantics (selected rules). \label{fig:sos-lio}
}
\end{figure}


Rules $\rulename{return}$ and
$\rulename{bind}$ encode the core monadic operations. The intermediate
value $\CLIO{t}$ is used to represent a \clio computation which produces
the term $t$, without any further effects on the configuration.  In
rule $\rulename{label}$, the operation $\llabel{l}{\gv}$ returns a
labeled value with label $l$ holding $\gv$ ($\LabeledO{l}{\gv}$), provided
that the current label flows to $l$ ($\lcur \flows l$) and
$l$ flows to the current clearance ($l \flows \lclr$). Note
that we force the second argument to be a ground value, i.e. it should
be fully normalized.
Rule $\rulename{unlabel}$ expresses that, given a labeled value
$\mathit{lv}$ with label $l$, the operation $\unlabel{\mathit{lv}}$
returns the value stored in $\mathit{lv}$ and updates the current
label to $\lcur \lub l$, to capture the fact that a value with label $l$ has been read, provided that this new label flows to the
current clearance ($l \flows \lclr$).
The operations $\textsf{getLabel}$ and $\textsf{getClearance}$ can be
used to retrieve the current label and clearance respectively.

Rules $\rulename{toLabeled}$ and $\rulename{reset}$ deserve special
attention. To evaluate $\toLabeled{l_1}{t}$, we first check that $l_1$
is a valid target label
($\lcur \flows l_1 \flows \lclr$) and then wrap
$t$ in a compartment using the special syntactic form
$\block{\lcur}{\lclr}{l_1}{t}$, recording the
current label and clearance at the time of entering
$\textsf{toLabeled}$ and the target label of the operation,
$l_1$. Evaluation proceeds by reducing $t$ in the context of
the compartment to a value of the form $\CLIO{t_1}$. Next,
the rule $\rulename{reset}$ evaluates the term $\block{\lcur}{\lclr}{l_1}{\CLIO{t_1}}$,
first checking that the current label
flows to the target of the current $\textsf{toLabeled}$
($\lcur \flows l_2$). Finally, the compartment is replaced
by a normal $\textsf{label}$ operation and the current label and
clearance are restored to their saved values.

\paragraph{DC Labels}
LIO is parametric in the label format, but for the purposes of this
paper in \clio{} we use DC~labels~\cite{stefan:2011:dclabels} with
three components to model confidentiality, integrity, and availability
policies.
A label $\ltrip{l_c}{l_i}{l_a}$ represents a policy with
confidentiality $l_c$, integrity $l_i$, and availability $l_a$.
Information labeled with $\ltrip{l_c}{l_i}{l_a}$ can be read by $l_c$,
is vouched for by $l_i$, and is hosted by $l_a$.  We write
$\projC{l}$, $\projI{l}$, $\projA{l}$ for the
confidentiality, integrity, and availability components of $l$,
respectively. Each component is a conjunction of disjunctions of
principal names, i.e., a formula in conjunctive normal form. A
disjunction $A\vee B$ in the confidentiality component means that
either $A$ or $B$ can read the data; in the integrity component, it
means that one of $A$ or $B$ vouch for the data, but none of them take
sole responsibility; in terms of
availability, it means that one of $A$ or $B$ can deny access to the
data. Conjunctions $A \wedge B$ mean that only $A$ and $B$ together
can read the data (confidentiality), that they jointly vouch for
the data (integrity), or that they can jointly deny access to the data (availability).
%
%
Data may flow between differently labeled entities, but only those with more \emph{restrictive} policies:
those readable, vouched for, or hosted by fewer entities.
A label $\ltrip{l_c}{l_i}{l_{a}}$ can flow to any label where the
confidentiality component is at least as sensitive than $l_c$, the integrity
component is no more trustworthy than $l_i$, and the availability is no
more 
than $l_a$,  i.e.  $l₁ \flows l₂$
if and only if $\projC{l_2} \Longrightarrow \projC{l_1}$, 
$\projI{l_1} \implies \projI{l_2}$,
and  $\projA{l_1} \implies \projA{l_2}$. We use logical implication
because it matches the intuitive meaning of disjunctions and
conjunctions, e.g., data readable by $A\vee B$ is less confidential
than data readable only by $A$, and data vouched for by $A\wedge B$ is
more trustworthy than data vouched for only by $A$.
%
%
In the rest of the paper, we consider only \clio{} computations that
work on labels of this form.

\section{Interacting with an Untrusted Store}
\label{sec:store-interaction}

\clio extends LIO with a key-value store. The language is extended with two new commands:
$\store{t_k}{t_v}$ puts a labeled value $t_v$ in the
store indexed by key $t_k$; $\fetch{t_k}{t_v}$ command fetches the entry with key $t_k$ and
if it cannot be fetched, returns the labeled value $t_v$.
In both commands, $t_k$ must evaluate to a ground value and the labeled value
$t_v$ must evaluate to a labeled ground value with type $\tau$.

\begin{figure}[t] 
\begin{center}
\resizebox{.95\columnwidth}{!}{
\begin{minipage}{\columnwidth}
\begin{center}
\begin{mathpar}
\figClioSemanticsStoreMain
\end{mathpar}
\end{center}
\end{minipage}
}
\end{center}
\caption{\clio language semantics (store and fetch rules). \label{fig:sos-clio-store}
}
\end{figure}

Semantics for fetch and store are shown
in Figure~\ref{fig:sos-clio-store}.
%
We modify the semantics to be a labeled transition system, where step relation $\ltol{\alpha}$ is annotated with \emph{store events} $\event$. Store event $\event$ is one of:
\begin{CompactItemize}
\item
$\skipcmd$ (representing no interaction
with the store, i.e., an internal step; we typically elide $\skipcmd$ for clarity),

\item $\storecmd{\gv_k}{\Labeled{l}{\gv}}$ (representing putting a labeled ground value $\Labeled{l}{\gv}$ indexed by
$\gv_k$),
\item $\fetchcmd{\gv_k}{\Labeled{l}{\gv}}$ (representing reading
a labeled value from the store indexed by $\gv_k$), or
\item $\missingcmd{\gv_k}$
representing no value is indexed by $\gv_k$.
\end{CompactItemize}

Labeling transitions with store events allows us to cleanly factor out the implementation of the store, enabling us to easily use either an idealized (non-cryptographic) store, or a store that uses cryptography to help enforce security guarantees. We describe the semantics of store events in both these settings later.

We associate a label $\lowlabel$ with the store. Intuitively, store level $\lowlabel$ describes how trusted the store is: it represents the inherent protections provided
by the store and the inherent trust by the store in \clio. 
For example, the store may be behind an organization's firewall 
so data is accessible only to organization members due to an external access control
mechanism (i.e., the firewall), so \clio can safely store the organization's information there.
Dually, there may be integrity requirements that \clio is trusted to uphold when
writing to the store. For example, the store may be used as part of a larger system that 
uses the store to perform important operations (e.g., ship customer orders). Thus the integrity component of the store label is a bound on the untrustworthiness of information that \clio should write to the store (e.g., \clio should not put unendorsed shipping requests in the store).
The availability component of the store label specifies a bound on who is able to corrupt information in the store and thus make it unavailable. (Note that we are concerned with \emph{information availability} rather than \emph{system availability}.)
In general, this would describe all the principals who have direct and 
indirect write-access to the store.

Rule \rulename{Store} (Figure~\ref{fig:sos-lio}) is used to put a labeled value \Labeled{l}{\gv} in the store, indexed by key $\gv_k$.
We require that the current label $\lcur$ is bounded above by store level $\lowlabel$. In terms of confidentiality, this means that any information that may be revealed by performing the store operation (i.e., $\lcur$) is permitted to be learned by users
of the store. 
For integrity, the decision to place this value in the store (possibly overwriting
a previous value) should not be influenced by information below the integrity
requirements of the store. For availability, the information should not be derived
from less available sources than the store's availability level.

%
%
Additionally, we require the current label
to flow to $l$, the label of the value that is being stored (i.e., $\canFlowTo{\lcur}{l_1}$).
Intuitively, this is because an entity that learns the labeled value also learns
that the labeled value was put in the store. Current label $\lcur$ is an upper
bound on the information that led to the decision to perform the store, and
$l_1$ bounds who may learn the labeled value.
For command $\fetch{\gv_k}{\Labeled{l_d}{\gv_d}}$,
labeled value $\Labeled{l_d}{\gv_d}$ serves double duty. First, if
 the store cannot return a suitable value (e.g., because there is no value indexed by key $\gv_k$, or because cryptographic signature verification fails), then the fetch command evaluates to the default labeled value $\Labeled{l_d}{\gv_d}$ (which might be an error value or a suitable default). Second, label $l_d$ specifies an upper bound on the label of any value that may be returned:
 if the store wants to return a labeled value $\Labeled{l}{\gv}$ where $\cantFlowTo{l}{l_d}$, then the fetch command evaluates to 
 $\Labeled{l_d}{\gv_d}$ instead. This allows programmers to specify bounds on information they are willing to read from the store.


Rule \rulename{Fetch-Valid} is used when a labeled value is successfully
fetched from the store. Store event $\fetchcmd{\gv_k}{\Labeled{l}{\gv}}$
indicates that the store was able to return labeled value $\Labeled{l}{\gv}$ indexed by the key $\gv_k$. 
Rule \rulename{Fetch-Invalid} is used when a labeled value cannot
be found indexed at the index requested or it does not safely flow
to the default labeled value (i.e., it is too secret, too untrustworthy
or not available enough), and causes the fetch to evaluate to the specified default labeled value. 
Since the label of the default value $l_d$ will be used for the label
of the fetched value in general, the availability of the store level
should be bounded above by the availabiltiy of the label of the default value
(i.e., $\canFlowToA{\projA{\lowlabel}}{\projA{l_d}}$) in both rules,
as the label of the fetched value should reflect the fact that
anyone from the store could have corrupted the value.




\subsection{Ideal Store Behavior}
\label{subsec:ideal-clio}

We informally describe the \emph{ideal} behavior of an untrusted
store from the perspective of a \clio program.\footnote{Complete formal definitions in \Cref{sec:ideal-storesem}.}
The ideal store semantics provides a specification of the
behavior that
a real 
implementation should strive for, and allows the programmer to focus on functionality and security properties of the store rather than the details of cryptographic enforcement of labeled values.
%
In Section~\ref{sec:design} we describe how
we use cryptography to achieve this ideal specification. 


We use a small-step  relation $\iconf{c}{\istore}\ito\iconf{c'}{\istore'}$ where  $\iconf{c}{\istore}$ and $\iconf{c'}{\istore'}$ are pairs of
a \clio configuration $c$ 
and an ideal store $\istore$. 
An ideal store $\istore$ 
maps ground values $\gv_k$ to labeled ground values $\Labeled{l}{\gv}$. If a store doesn't contain a mapping for an index $\gv_k$, we represent that as mapping it to 
the distinguished value
$\bot$.

Store events 
are used to communicate with the store. When a $\storecmd{\gv_k}{\Labeled{l}{\gv}}$ event is emitted, the store is updated appropriately. When the \clio computation issues a fetch command, the store provides the appropriate event (i.e., either provides event
$\missingcmd{\gv_k}$ or 
event  $\fetchcmd{\gv_k}{\Labeled{l}{\gv}}$ for an appropriate labeled value $\Labeled{l}{\gv}$). For \clio computation steps that do not interact with the store, store event $\skipcmd$ is emitted, and the store is not updated.



\subsection{Non-\clio Interaction: Threat Model}
\label{subsec:non-clio-interaction}
We assume that programs other than \clio computations may interact with the store
and may try to actively or passively subvert the security of \clio programs.
Our threat model for these adversarial programs is as follows (and uses store level $\lowlabel$ to characterize some of the adversaries' abilities).
%
%
\begin{CompactItemize}
\item All indices of the key-value store are public information, and an adversary can probe any index of the store and thus notice any and all updates to the store.
\item An adversary can read labeled values $\Labeled{l_1}{\gv}$ in the store where the confidentiality level of label $l_1$ is at least as confidential as the store level $\lowlabel$ (i.e., $\canFlowToC{\projC{l_1}}{\projC{\lowlabel}}$). 
\item An adversary can put labeled values  $\Labeled{l_1}{\gv}$ in the store (with arbitrary ground value index $\gv_k$) provided the integrity level of store level $\lowlabel$ is at least as trustworthy as the integrity of label $l_1$ (i.e., $\canFlowToI{\projI{\lowlabel}}{\projI{l_1}}$). 
\end{CompactItemize}
An adversary can adaptively interact with the store. That is, the behavior of the adversary may depend upon (possibly probabilistically) changes the adversary detects or values in the store.

We make the following restrictions on adversaries.
\begin{CompactItemize}
\item The adversary does not have access to timing information. That is, it cannot observe the time between updates to the store. We defer to orthorgonal techniques to mitigate the impact of timing channels~\cite{Askarov:2010}. For example, \clio could generate store events on a fixed schedule.
\item The adversary cannot observe termination of a \clio program, including abnormal termination due to a failed label check. This assumption can be satisfied by requiring that all \clio programs do not diverge and are checked to ensure normal termination, e.g., by requiring $\getLabel$ checks on the label of a labeled value before unlabeling it. Static program analysis can ensure these conditions, and in the rest of the paper we consider only \clio programs that terminate normally. 
\end{CompactItemize}


\begin{figure}[t]
\begin{mathpar}
\figIdealClioLowStep
\end{mathpar}
\figIdealClioInteractions
\vspace{-1.5em}
\caption{Adversary Interactions and Low Steps \label{fig:ideal-clio-low-steps}}
\end{figure}

Note that even though the
adversary might have compromised the \clio program, it can only
interact with it at runtime through the store. The adversary does not
automatically learn everything that the program learns, because data
in the \clio runtime is still subject to \clio semantics and the IFC
enforcement, which might prevent exfiltration to the store. \todo{Give
example here? Or later? Maybe when explaining fetch?}  The \clio semantics
thus gives a more precise characterization of the power of the adversary.  Rather 
than proving the security in the presence of a decryption oracle
(e.g., CCA or CCA-2 \cite{pass-shelat}), the \clio runtime prevents
system interactions from being used as a decryption oracle, by
construction.

We formally model the non-\clio interactions with the store using
sequences of \emph{adversary interactions} $I$, given in
Figure~\ref{fig:ideal-clio-low-steps}.
Adversary interactions are $\skipcmd$, $\storecmd{\gv'}{\Labeled{l_1}{\gv}}$
and $\corruptcmd{\gv_1, ... \gv_n}$, which, respectively: do nothing; put a labeled value in the store; and
%
delete the mappings
for entries at indices $\gv_1$ to $\gv_n$. 
For storing labeled values, we restrict the integrity of the labeled value stored
by non-\clio interactions to be at most at the store level.
Sequences of interactions $\iinteraction_1 \cdot ... \cdot \iinteraction_n$
are notated as $\iinteractions$.

To model the adversary actively updating the store, we define a step semantics
$\Lowstep{}$ that includes adversary interactions $\iinteractions$. We restrict
interactions to occur only at \emph{low steps}, i.e., when the current label of the \clio computation is less than or equal to the store level $\lowlabel$. (By contrast, a \emph{high step} is when the current label can not flow to $\lowlabel$.)
This captures the threat model assumption that the attacker cannot observe timing.
Rules \rulename{Low-Step}
 and \rulename{Low-To-High-To-Low-Step} in Figure~\ref{fig:ideal-clio-low-steps} express adversary interactions occurring only at low steps.





\section{Realizing \clio}
\label{sec:design}

In this section we describe how \clio uses cryptography to enforce
the policies on the labeled values through a formal model, called
the real \clio store semantics. This model serves as the 
basis for establishing strong, formally proven, computational
 guarantees of the \clio system.
We first describe how DC labels are enforced with 
cryptographic mechanisms (Section~\ref{subsec:keystore}), 
and then describe the real \clio
store semantics (Section~\ref{subsec:clio-store-semantics}). 

\subsection{Cryptographic DC Labeled Values}
\label{subsec:keystore}
\todo{This section is queueing up a lot of definitions, but it is hard to hold all of 
this in memory.  maybe we need a summary table to refer to?}

\clio, like many systems, 
identifies security principals
with the \emph{public key} of a cryptographic key pair, and associates
the authority to act as a given principal with possession of the
corresponding \emph{private key}. At a high level, \clio ensures that only those with access to a principal's private key can access information confidential to that principal and vouch for information on behalf of that principal.



\clio tracks key pairs in a \emph{keystore}.
Formally, a keystore is a mapping 
$\keystore : p \mapsto (\{0, 1\}^*, \{0, 1\}^*_{\bot}) $, where $p$ is the principal's well-known name, and the pair of bit strings contains the public and private keys for the principal.
In general, the private key for a principal may not be known---represented by $\bot$---which corresponds to knowing the identity of a principal, but not possessing its authority.
Keystores are the basis of authority and identity for \clio computations.
We use meta-functions on keystores
to describe the authority of a keystore in terms of DC labels.\footnote{Complete
definitions for these functions are in \Cref{appendix:real-syntax}.}
Conceptually, a keystore can access and vouch for any information 
for a principal for which it has the principal's 
private key. Meta-function $\authorityOf{\keystore}$ returns a label where each component (confidentiality, integrity, and availability) is the conjunction of all principals for which keystore $\keystore$ has the private key.
We also use the keystore to determine the starting label of a \clio program $\Start{\keystore}$
and the 
least restrictive
clearance $\Clr{\keystore}$, which are, respectively, the 
 most public, trusted, and available label possible
and the
 most confidential, least trusted, and least available data that the computation can compute on, given the keystore's authority.


Using the principal keystore as a basis for authority and identity
for principals,
\clio derives a cryptographic protocol 
that enforces the security policies of
safe information flows defined by DC labels. 

In the DC label model, labels are made up of triples of \emph{formulas}. Formulas are conjunctions of \emph{categories} $\category_1 \wedge ... \wedge \category_n$.
Categories are disjunctions of principals $p_1 \vee ... \vee p_n$. 
Any principal in a category can read (for confidentiality)
and vouch for (for integrity) information bounded above by the level of the category. 
We enforce that ability cryptographically by ensuring that only principals
in the category have access to the private key for that category.
\clio achieves this through the use of \emph{category keys}. 

A category key serves as the 
cryptographic basis of authority and identity for a category.
Category keys are made up of the following components:
a \emph{category public key} that is readable by all principals, 
a \emph{category private key} that is only readable by members of the category,
and a \emph{category key signature} that is a signature on the category
public key and category private key to prove the category key's authenticity.
Category keys are created lazily by \clio as needed and placed in
the store.
A category key is created using a randomized meta-function\footnote{Defined formally in \Cref{appendix:serialization}.} parameterized by the keystore.
The generated category private key is encrypted 
for each member of the category separately using each member principal's 
public key. To prevent illegitimate creation of category keys, the encrypted
category private key and category public key are together
signed using the private key of one of the category members.\footnote{The \clio runtime ensures that the first time a category key for a given category is required, it will be because data confidential to the category or vouched for by the category is being written to the store, and thus the computation has access to at least one category member's private key. Note that any computation with the authority of a category member has the authority of the category.}
When a category key is created and placed in the store, 
it can be fetched by anyone but decrypted only
by the members of the category.
When a \clio computation fetches a category key,
it verifies the signature of the category key to ensure that a  
category member actually vouches for it.\footnote{``Encrypt-then-sign''issues (e.g., \cite{anderson1995robustness}) do not apply here as the threat model (i.e., signed encrypted messages implying authorship) is different.} 
(Failing to verify the signature would allow an adversary to trick a \clio computation into using a category key that is readable by the adversary.)

A \clio computation
encrypts data confidential to a formula $\category_1 \wedge ... \wedge \category_n$
by chaining the encryptions of the value.
It first encrypts  using $\category_1$'s category public key
and then encrypts the resulting ciphertext for formula $\category_2 \wedge ... \wedge \category_n$. 
This form of layered encryption relies on a canonical ordering of categories;
we use a  lexicographic ordering of principals to ensure a canonical ordering of encryptions and decryptions.

A \clio computation signs data for a formula by signing the data with each category's private key and then concatenating the signatures together. 
Verification  succeeds only if every category signature
can be verified. 

Equipped with a mechanism to encrypt and sign data for DC labels that
conceptually respects safe information flows in \clio, we use this mechanism to serialize and deserialize labeled values to the store. Given a labeled ground value $\Labeled{\gv}{\ltrip{l_c}{l_i}{l_a}}$, the value $\gv$ is signed according to formula $l_i$. The value and signature are encrypted according to formula $l_c$, and the resulting bitstring is the serialization of the labeled value. Deserialization performs decryption and then verification. 
If deserialization fails, then \clio treats it like a missing entry, and
the fetch command that triggered the deserialization would evaluate to the default labeled value.

\paragraph{Replay Attacks}
Unfortunately, using just encryption and signatures does not faithfully implement the ideal store semantics: the adversary is able to swap entries in the store, or re-use a previous valid serialization, and thus in a limited way modify high-integrity labeled values in the store. We prevent these attacks by requiring that the encryption of the ground value and signature also includes the index value (i.e., the key used to store the labeled value) and a \emph{version number}. 
%
The real
\clio semantics keeps track of the last seen version of a labeled value
for each index of the store. When a value is serialized, the version of that
index is incremented before being put in the store. When the value is 
deserialized the version is checked to ensure that the version is not
before a previously used version for that index. In a distributed setting, this version counter
could be implemented as a vector clock between \clio computations to account
for concurrent access to the store. However, for simplicity, we use
 natural numbers for versions in the real \clio store semantics.



\subsection{\clio Store Semantics}
\label{subsec:clio-store-semantics}
\todo{reviewer: did not think store semantics accounted for category key creation}
\todo{reviewer: puzzled by signing before encrypting, don't we need CCA2-security? what about malleability?}

In this section we describe the real \clio store semantics
in terms of
a small-step probabilistic relation $\rto{p}$.
The relation models a step taken from a real 
\clio configuration to a real \clio configuration with probability $p$.
A real \clio configuration is a triple
$\rconf{c}{\istoredist_0}{\interactionsdist}{\versions}$
of a \clio configuration $c$, 
a \emph{distribution} of \emph{sequences} of \emph{real interactions} 
with the store $\interactionsdist$,
 and a version map $\versions$. 
The version map tracks version numbers for the store to prevent replay attacks, as described above. For technical reasons, instead of the configuration representing the key-value store as a map, we use
the history of store interactions (which includes interactions made both by the \clio computation and the adversary). The sequence of interactions applied to the initial store gives the current store. Because the real \clio store semantics are probabilistic (due to the use of a probabilistic cryptosystem and cryptographic-style probabilistic polynomial-time adversaries), configurations contain distributions over sequences of store interactions.

Real interactions $\interaction$ (and their sequences $\interactions$)  are defined 
in Figure~\ref{fig:real-clio-semantics} and are similar to interactions with the ideal store. However, instead of labeled values containing ground values, they contain 
bitstrings $b$ (expressing the low-level details of the cryptosystem and the ability of the adversary to perform bit-level operations). Additionally,
the interaction $\storecmd{\category}{\categorykey}$ represents storing of a category key. These
interactions arise from the serialize metafunction, which may create new category keys. 
Note that the interaction $\storecmd{\gv_k}{\LabeledO{l}{b}}$ does not need an integrity
side condition (as it did in the ideal semantics) in the real semantics since there is
no distinction between corruptions and valid store interactions.

We use notation
$$
\randomExprL f(X_1, ..., X_n) ~|~ X_1 \drawnFrom D_1; ~ ... ~ X_n \drawnFrom D_n \randomExprR
$$
to describe the probability distribution over the function $f$ with inputs of random variables $X_1, ...,  X_n$ where $X_i$ is distributed according to distribution $D_i$ for $1 \leq i \leq n$.

\begin{figure}[t]
\begin{mathpar}
\figRealClioMain
\end{mathpar}
\caption{Real \clio Semantics} \label{fig:real-clio}
\end{figure}

Figure~\ref{fig:real-clio} presents the inference rules for 
$\rto{p}$. 
Internal steps do not affect the interactions or versions. For storing 
(rule \rulename{Store}),
the version of the entry is incremented using the $\incrementnm$ function and
the real \clio configuration uses a new distribution of interactions 
$\interactionsdist'$ containing the interactions to store the labeled value. 
The new distribution contains the original interactions
(distributed according to the original distribution of interactions)
along with a concatenation of labeled ciphertexts and 
any new category keys (distributed according to the distribution
given by serialization function). Note that the label of the stored value is
not encrypted as it is public information.
The configuration steps with probability $1$ as the \rulename{Store} rule
will be used for all $\textsf{store}$ operations.

When fetching a labeled value, there are three possible rules 
that can be used depending on the
current state of the store:  \rulename{Fetch-Exists},  \rulename{Fetch-Missing},  \rulename{Fetch-Replay}. 
The premise,
%
$$
(\istore, p) \in \randomExprL \interactions(\emptyset) ~|~ 
	\interactions \drawnFrom \interactionsdist \randomExprR
$$
in each of these rules
means that store $\istore$ has probability $p$ of being produced (by drawing
interaction sequence $\interactions$ from distribution  $\interactionsdist$ and applying $\interactions$ to the empty store $\emptystore$ to give store $\istore$).

Which rule is used for a fetch operation depends on the state of the
store, and so the transitions may have probability less than one.
Rule \rulename{Fetch-Exists} is
used when the sequence of interactions drawn produces a store that has a
serialized labeled value indexed by $\gv_k$ that can be correctly deserialized
and whose version is not less than the last version seen at this index.
Rule \rulename{Fetch-Missing} is used when the sequence of interactions
drawn
produces a store that either does not have an entry indexed by $\gv_k$, or
has an entry that cannot be correctly deserialized. Finally, \rulename{Fetch-Replay}
rule is used when the sequences of interactions drawn produce a store
where an adversary has attempted to replay an old value:
the store has a labeled value that can be deserialized correctly, but whose 
recorded index is not the same as the index requested by the \clio computation
or whose version is less than the version last seen. 

\begin{figure}[t]
{\small
\begin{mathpar}
\figRealClioLowStep
\end{mathpar}
\begin{alignat*}{2}
\textrm{Interactions:} &  \quad
	\interaction & \Coloneqq~ &
	 	\skipcmd = \lambda \istore. ~ \istore \\
	&& \!\!|\ ~~ &  \storecmd{\category}{\categorykey} = \lambda \istore. ~ \istore[\category \mapsto \categorykey]\\
	&& \!\!|\ ~~ & \storecmd{\gv_k}{\Labeled{l}{\bitstream}} = \lambda \istore. ~ \istore[\gv_k \mapsto \Labeled{l}{\bitstream}]\\ 
	\textrm{Strategies:} & \quad
		\strategy & :~~~ & \interactionsdist \rightarrow \interactionsdist 
\end{alignat*}
\noindent\figStepFunctionDefNL
}
\caption{Real \clio Low Step Semantics} \label{fig:real-clio-semantics}
\end{figure}

Similar to the ideal store semantics, we use a low step relation $\Lowstep{p}$
to model adversary interactions, shown in Figure~\ref{fig:real-clio-semantics}.
The low step relation is also probabilistic as it is based on the probabilistic 
single step relation $\rto{p}$. Additionally, we use a distribution of sequences of adversarial 
interactions $\interactionsdist_A$ to model an 
adversary that behaves probabilistically. In rules \rulename{Low-Step} and 
\rulename{Low-To-High-To-Low-Step} a new distribution of interactions, $\interactionsdist'$ is created by concatenating interaction sequences drawn from the existing distribution of interactions $\interactionsdist$ and the adversary distribution $\interactionsdist_A$. This is analogous to the application of adversary interactions to the current store in the ideal semantics.
The rest of the definitions of the rules
follow the same pattern as the ideal \clio low step store semantics. 

With the low step relation, we use metafunction $\multistepnm$ to describe 
the distributions of real \clio configurations resulting from taking $j$
low steps from configuration $c_0$, formally defined in 
Figure~\ref{fig:real-clio-semantics}. The $\multistepnm$ function
is parameterized by the keystore $\keystore$ and store level $\lowlabel$.
To provide a source of adversary interactions while
running the program, the $\multistepnm$ function also takes as input a 
\emph{strategy} $\strategy$ which is a function from 
distributions of interactions to distributions of interactions, representing the
probabilities of interactions an active adversary would perform.
Before each low step, the strategy is invoked to produce a distribution of 
interactions that will affect the store that the \clio computation is using.

Strategy $\strategy$ expresses the ability of the attacker to modify the store. The
attacker chooses $\strategy$ (and $t$, $v_0$, and $v_1$), and $\strategy$
 interacts with the store during
execution. $\strategy$ is a function from (distributions of) interaction sequences to
(distributions of) interaction sequences, i.e., a function from a history of
what has happened to the store so far to the attacker’s next modifications to
the store. Note that we do not explicitly model fetching from the store as an 
adversary interaction. 
There is no need for $\strategy$ to fetch values to determine the next
modification to the store since $\strategy$ effectively observes the entire history of store
interactions. At the end of the game when the adversary continuation ($\adversary_2$) needs
to pick $v_0$ or $v_1$, it observes the history of interactions with the store via
the interaction sequence $\interactionsvar{\interaction_b}$, and thus does not need to explicitly get or fetch values.

\section{Formal Properties}
\label{sec:eval}

\subsection{Indistinguishability}
\label{subsec:indist}

A cryptosystem is \emph{semantically secure} 
if, informally, ciphertexts of messages of equal length
are \emph{computationally indistinguishable}. 
Two sequences of probability distributions
are computationally indistinguishable (written $\{X_n\}_n \approx \{Y_n\}_n$)
if for all non-uniform probabilistic polynomial time (ppt)
algorithms $\adversary$,
$$
\big| ~ \probabilityOfCond{\adversary(x) = 1}{x \drawnFrom X_n}  - \probabilityOfCond{\adversary(y) = 1}{y \drawnFrom Y_n~ } ~ \big|
$$
is \emph{negligible} in $n$~\cite{Goldwasser1982}. 


In modern cryptosystems, semantic security is defined as
indistinguishability under chosen-plaintext attacks (CPA)
\cite{pass-shelat}.
\begin{definition}[Indistinguishability under Chosen-Plaintext Attack]
Let the random variable $\mathrm{IND}_b(\mathcal{A}, n)$ denote
the output of the   experiment, where
$\mathcal{A}$ is non-uniform ppt, $n \in \varmathbb{N}$, $b \in \{0,1\}$ :
\noindent{\em
$$
\figCPAGame
$$}
{\em $\Pi = (\Gen, \Enc, \Dec)$} is Chosen-Plaintext Attack (CPA) secure if for all
non-uniform ppt $\mathcal{A}$:
$$
\Big\{~\mathrm{IND}_0(\mathcal{A}, n)~\Big\}_n \approx \Big\{~\mathrm{IND}_1(\mathcal{A}, n)~\Big\}_n
$$
\end{definition}

This definition of indistinguishability phrases the security of the
cryptosystem in terms of a game where 
an adversary receives the
public key and then produces two plaintext messages of equal length. 
One of the two messages is encrypted and 
the resulting ciphertext given to the adversary. 
The cryptosystem is CPA Secure if no adversary exists that can
produce substantially different distributions of output 
based on the choice
of message. In other words, no computationally-bounded adversary 
is able to effectively distinguish which message was encrypted.

\clio relies on a semantically secure cryptosystem, but this
is insufficient 
for \clio to protect the confidentiality of 
secret information. This is because CPA Security provides guarantees
only for individually chosen plaintext messages.
In contrast, in our setting we consider \emph{terms} (i.e., programs) 
chosen by an adversary.
There are also many principals and as a 
result many keys in a real system, so \clio must protect arbitrarily many principals' information from the adversary. Additionally, the adversary may already have 
access to some of the keys. Finally, the adversary is active: it can see interactions with the store and issue new interactions adaptively while the program is running.
It can attempt to leverage a \clio computation
to illegitimately produce a value it should not have, or could try to trick the \clio system into leaking secret information by interacting with the store.
Traditionally, these actions of the adversary are modeled by queries to a decryption oracle, as in CCA-2~\cite{pass-shelat}.
Here, they are modeled directly by the \clio language and store semantics.

We chose to formulate a new definition of security that addresses these concerns,
as many previous classical definitions of security fall short in this setting:
\begin{itemize} 
	\item Noninterference does not permit the use of computationally secure mechanisms like
cryptography. 
	\item CPA security considers only the semantics of the cryptographic
algorithms, not the system they are embedded within. 
	\item CCA and CCA2 attempt to
model system behavior using oracles, but the connection between these oracles
and an actual system is too abstract. 
\end{itemize}
In contrast, we chose to employ a computational
model of cryptography that accurately represents the power of the attacker precisely using
the semantics of the language and interactions with the store.

With these considerations in mind, we define indistinguishability
under a new form of attack: \emph{chosen-term attacks} (CTA).

\begin{definition}[Indistinguishability under Chosen-Term Attack]
Let the random variable $\IND{b}{\keystore}{\adversary}{\principals}{j}{n}$ denote 
the output of the following experiment, where 
$\cryptosys = (\Gen, \Enc, \Dec, \Sign, \Verify)$, 
$\adversary$ is non-uniform ppt, $n \in \nat$, $b \in \{0,1\}$:
$$
\figCTAGame
$$
\clio using \cryptosys is CTA Secure if for all
non-uniform ppt $\adversary$, $j \in \varmathbb{N}$, keystores \keystore, and principals $\principals$:
$$
	\distfam{\IND{0}{\keystore}{\adversary}{\principals}{j}{n}} \cequiv
	\distfam{\IND{1}{\keystore}{\adversary}{\principals}{j}{n}}	
$$
\end{definition}

The CTA game follows the same structure as the CPA game.
In addition, we allow the adversary to know certain information (by fixing it in the game),
including some part of the keystore  ($\keystore_0$),
the set of principals that \clio is protecting ($\principals$), and 
 the number of low steps the program takes ($j$).
Cryptosystem $\cryptosys$ is used implicitly in the CTA game to generate keys,
encrypt, decrypt, sign and verify\footnote{More formally, $\mathrm{IND}_b$, $\adversary$, and
the semantics are also parameterized on $\cryptosys$, and the uses of $\Gen, \Enc, \Dec, \Sign, \Verify$ should be explicitly taken from the tuple $\cryptosys$ though we elide their explicit usage in our notation for clarity.}.

In this game setup, $\Gen(\principals, 1^n)$ generates a new keystore
$\keystore'$
containing private keys for each of the principals in $\principals$,
using the underlying cryptosystem's $\Gen$ function for each keypair.
Then, the adversary receives all public keys of the keystore
$\pubkeys{\keystore}$
and
returns three well-typed \clio terms: a function $t$, and two program inputs to the function
$v_0$ and $v_1$ 
that must be confidentiality-only low equivalent $\lowEquivC$
(i.e., they may differ only on secret values)\footnote{Complete definition of low equivalence is in \Cref{subsec:complete-low-equiv}.}. It also returns a strategy
$\strategy$ that 
models the behavior of the adversary on the store while the
computation is running. 
Note that the strategy is also polynomial in the security parameter as it
is constructed from a non-uniform polynomial time algorithm.
The program $t$ is run with one of the inputs $v_0$ or $v_1$ for a fixed number of steps $j$.
The adversary receives the interactions resulting from a run
of the program and needs to use that information to determine which secret
input the program was run with.

Being secure under a chosen-term attack means that the sequences of interactions between two low-equivalent programs are indistinguishable and hence an adversary does not learn any secret information from the store despite actively interacting with it while the program it chose is running. 
Note that the adversary receives the full trace of interactions on the store 
(including its own interactions); this gives it enough information to reconstruct
the final state of the store and any intermediate state.
For any set of principals,
and any adversary store level, the interactions with the store contain
no efficiently extractable secret information for all well-typed terminating
programs.







\begin{theorem}[CTA Security]
If \cryptosys if CPA Secure, then \clio using \cryptosys is CTA Secure.
\end{theorem}

We prove this theorem in part by induction over the low step relation
$\Lowstep{p}$, to show that two low equivalent configurations will
produce low equivalent configurations, including computationally
indistinguishable distributions over sequences of interactions. A
subtlety is that we must strengthen the inductive hypothesis to
show that sequences of interactions satisfy a stronger
syntactic relation (rather than  being just computationally
indistinguishable).

More concretely, the proof follows three high-level steps. 
 First, we show how a relation $\asymp$ on families of distributions of
sequences of interactions 
preserves computational indistinguishability.
That is, if $\interactionsdist_1 \asymp \interactionsdist_2$ and $\cryptosys$ is CPA secure, then $\interactionsdist_1 \approx \interactionsdist_2$. 
Second,
we show that as two low equivalent configurations step using the low step relation
$\Lowstep{p}$, low equivalence is preserved and the interactions
they produce satisfy the relation $\asymp$. 
Third, 
we show that the use of the $\multistepnm$ metafunction on two
low equivalent configurations will produce computationally indistinguishable distributions
over distributions of sequences of interactions.
Each step of the proof relies on the previous step and the 
first step relies on the underlying assumptions on the cryptosystem. We
now describe each step of the proof in more detail.


\paragraph{Step 1: Interactions Relation}
We consider pairs of arbitrary distributions of sequences of interactions
and show that, if they are both of a certain syntactic form then they are indistinguishable.
Importantly, the indistinguishability lemmas do not refer to the \clio store semantics,
i.e., they merely describe the form of arbitrary interactions that may or may not
have come from \clio. The invariants on pairs of indistinguishable distributions of
interactions implicitly require low equivalence of the programs that generated them, 
and low equivalence circularly requires indistinguishable distributions of interactions. 
As a result, we describe the lemmas free from the 
\clio store semantics to break the circularity.

We progressively define the relation $\asymp$ on a pair of interactions.
Initially, distributions of interactions only contain secret encryptions so that
we can appeal to a standard cryptographic argument of multi-message security.
Formally,
for all keystores $\keystore_0$, and
$l_1, ..., l_k$, such that $\canFlowToC{\projC{l_i}}{\projC{\authorityOf{\keystore}}}$, and for all $m_{\{1,2\}}^1 ... m_{\{1,2\}}^n$ and all principals $\principals$, if $|m^i_1| = |m^i_2|$ for all $1 \leq i \leq k$ and $\cryptosys$ is CPA Secure, then 

\noindent
\hspace{-0mm}
\vbox{
\[
\begin{array}{l}
\big\{ ~ \storecmd{\gv^1}{\Labeled{l^1}{b^1_1}} \cdot  ~...~ \cdot \storecmd{\gv^k}{\Labeled{l^{k}}{b^{k}_1}} ~\big|~ \\ \hspace*{1cm} \keystore \drawnFrom \Gen(1^n); \\ \hspace*{1cm} (pk^i, sk^i) \in rng(\keystore); \\ \hspace*{1cm}  b^i_1 \drawnFrom \Enc(pk^i, m^i_1); ~ 1 \leq i \leq k ~ \big\}_n \\
\multicolumn{1}{c}{\asymp} \\
\big\{ ~ \storecmd{\gv^1}{\Labeled{l^1}{b^1_2}} \cdot  ~...~ \cdot~ \storecmd{\gv^k}{\Labeled{l^{k}}{b^{k}_2}} ~\big|~ \\ \hspace*{1cm} \keystore \drawnFrom \Gen(1^n); \\ \hspace*{1cm} (pk^i, sk^i) \in rng(\keystore); \\ \hspace*{1cm}  b^i_2 \drawnFrom \Enc(pk^i, m^i_2); ~ 1 \leq i \leq k ~ \big\}_n 
\end{array}
\]
}

Using multi-message security as a basis for indistinguishability, 
we then expand the relation to contain readable encryptions (i.e.,
ones for which the adversary has the private key to decrypt) where the
values encrypted are the same. 
%
In the complete definiton of $\asymp$, we expand it to also 
contain interactions from a strategy,
forming the final relationship on interactions captured by the $\asymp$ 
relation.

We establish an invariant that must hold between pairs in the relation in
order for them to be indistinguishable. For example, in the first definition,
the lengths of each corresponding message between the pair must be the same.
Each intermediate definition of $\asymp$ is used
to show that a ppt can simulate the extra information in the more generalized
definition (thus providing no distinguishing power). 
For the first definition of the relation containing only secret encryptions, 
a hybrid argument is used similar to showing multi-message
CPA security~\cite{pass-shelat}.


\paragraph{Step 2: Preservation of Low Equivalence}
We show that as two low equivalent programs $t$ and $t'$ progress, they simultaneously 
preserve low equivalence $t \lowEquivC t'$ and the distributions of sequences of interactions they produce $\interactionsdist$ and $\interactionsdist'$ are in the relation $\asymp$.

We first show that if $c_0 \ltol{\alpha} c_0'$ and $c_1 \ltol{\alpha} c_1'$ and $c_0 \lowEquivC c_1$ then $c_0' \lowEquivC c_1'$. This proof takes advantage of 
the low equivalence preservation proofs for LIO in all cases except for the
storing and fetching rules. For store events, since all
values being stored will have the same type (due to type soundness), and will be ground values,
serialized values will have the same message lengths.

We then show that if 
$$(\rconf{c_0}{\istoredist_0}{\interactionsdist_0}{\versions_0}, \interactionsdist) \Lowstep{p} \rconf{c_0'}{\istoredist_0'}{\interactionsdist_0'}{\versions_0'}$$ and 
$$(\rconf{c_1}{\istoredist_1}{\interactionsdist_1}{\versions_1}, \interactionsdist) \Lowstep{p} \rconf{c_1'}{\istoredist_1'}{\interactionsdist_1'}{\versions_1'}$$ and 
$$(c_0 \lowEquivC c_1) \wedge (\versions_0 = \versions_1) \wedge (\interactionsdist_0 \asymp \interactionsdist_1)$$ then,
$$(c_0' \lowEquivC c_1') \wedge
(\versions_0' = \versions_1') \wedge
(\interactionsdist_0' \asymp \interactionsdist_1').$$
The proof on $\Lowstep{p}$ relies on the previous preservation proof on $\ltol{\alpha}$ and the indistinguishability results on $\asymp$.


\paragraph{Step 3: Indistinguishability of the step metafunction}
We show that the $\multistepnm$ metafunction preserves low equivalence. More formally, we show that if
$c_0 \lowEquivC c_1$ and $\versions_0 = \versions_1$ and $\interactionsdist_0 \asymp \interactionsdist_1$
then
$$
\begin{array}{c}
\{ (\interactionsdist_0', p_0 \cdot ... \cdot p) ~|~ \rconf{c_0}{\istoredist_0}{\interactionsdist_0}{\versions_0} \Lowstep{p_0} ... \Lowstep{p} \rconf{c_0'}{\istoredist_0'}{\interactionsdist_0'}{\versions_0'} \}_n\\
\approx \\
\{ (\interactionsdist_1', p_0' \cdot ...\cdot p') ~|~ \rconf{c_1}{\istoredist_1}{\interactionsdist_1}{\versions_1} \Lowstep{p_0'} ... \Lowstep{p'} \rconf{c_1'}{\istoredist_1'}{\interactionsdist_1'}{\versions_1'} \}_n
\end{array}
$$
%
%
We prove this by showing that
the probabilities of traces taken by two low equivalent configurations
are equal with all but negligible probability.
As an example, Figure~\ref{fig:lowequiv-proof} shows graphically how one step of the trace is handled.
We examine the result of $\multistepnm_{\keystore}(c_1, \strategy, 1)$ and $\multistepnm_{\keystore}(c_2,$ $\strategy, 1)$ where 
$c_1 \lowEquivC c_2$. (Note that this setup matches the instantiation of
the CTA game where $j=1$.) The left rectangle shows the resulting distribution
over distributions of configurations after one step of the
$c_1$ configuration. The right circle
shows the resulting distribution over distributions of configurations
after one step of the $c_2$ configuration. Due to the results from Step 2,
we can reason that $c_1' \lowEquivC c_2'$ and that $c_1'' \lowEquivC c_2''$.
We can also conclude that $\interactionsdist_1 \asymp \interactionsdist_2$
and that $\interactionsdist_1' \asymp \interactionsdist_2'$. The final
step of the proof is to show that the interactions from the 
resulting two distributions (i.e.,
the top circle and bottom circle) are computationally indistinguishable.
%
That is, we
show that $p_1$ is equal to $p_2$ and also $p_1'$ is equal to $p_2'$ with all
but negligible probability.

\begin{figure}[t]
\tikzstyle{level 1}=[level distance=0cm, sibling distance=4cm]
\tikzstyle{level 2}=[level distance=1.5cm, sibling distance=1.9cm]

\tikzstyle{bag} = [text width=1em, text centered, text width=3cm]
\tikzstyle{bag2} = [text width=1em, text centered, text width=1.3cm]
\tikzstyle{end} = [circle, minimum width=2pt,fill, inner sep=0pt]

\hspace{-0cm}
\begin{tikzpicture}[grow=south, sloped,edge from parent/.style={draw,-latex}]
\node[bag] {~}
      child {
        node[bag] (c1 node) {\small$\big(\rconf{c_1}{\emptyset}{\skipcmd}{\versions_0}, \strategy(\skipcmd)\big)$}
            child {
                node[bag2] (c1b node) {\small$\rconf{c_1''}{\istoredist_1'}{\interactionsdist_1'}{\versions_1'}$}    
                    edge from parent 
                    node[above=-2pt] {$\Lowstep{p_1'}$}
            }
            child {
                node[bag2] (c1a node) {\small$\rconf{c_1'}{\istoredist_1}{\interactionsdist_1}{\versions_1}$}     
                edge from parent         
                node[above=-2pt] {$\Lowstep{p_1}$}
            }
        edge from parent[draw=none]
      }
      child {
      node[bag] (c2 node) {\small$\big(\rconf{c_2}{ \emptyset }{ \skipcmd }{\versions_0}, \strategy(\skipcmd)\big)$}
          child {
              node[bag2] (c2b node) {\small$\rconf{c_2''}{\istoredist_2'}{\interactionsdist_2'}{\versions_2'}$}        
                  edge from parent 
                  node[above=-2pt, rotate=0] {$\Lowstep{p_2'}$}
          }
          child {
              node[bag2] (c2a node) {\small$\rconf{c_2'}{\istoredist_2}{\interactionsdist_2}{\versions_2}$}        
              edge from parent         
              node[above=-2pt] {$\Lowstep{p_2}$}
          }
        edge from parent[draw=none]
      }; 
\draw[bend right,dashed] 
  (c2 node) to node[below=3pt, rotate=-0] { \small$\lowEquivC$ } (c1 node);
\draw[bend left=60,dashed] 
  (c2b node) to node[above=1pt, rotate=-0] { \small$\lowEquivC$ }  (c1b node);
\draw[bend left=60,dashed] 
  (c2a node) to node[above=1pt, rotate=-0] { \small$\lowEquivC$ }  (c1a node);

\node[draw,fit=(c1b node) (c1a node) ,inner sep=2mm,rounded corners=1mm] (d1) {};
\node[below=8mm of d1] {\small${\multistepnm_{\keystore}(c_1, \strategy, 1)}$}; 
\node[draw,fit=(c2b node) (c2a node) ,inner sep=2mm,rounded corners=1mm] (d2) {};
\node[below=8mm of d2] {\small${\multistepnm_{\keystore}(c_2, \strategy, 1)}$}; 
\end{tikzpicture}

\caption{Low equivalence is preserved in $\multistepnm_{\keystore}$ for two low
equivalent configurations $c_1$ and $c_2$ and a strategy $\strategy$.} \label{fig:lowequiv-proof}
\end{figure}

\subsection{Leveraged Forgery}
\label{subsec:forgery}

Whereas in the previous subsection we considered the security of encryptions, 
in this case we consider the security of the signatures.
We show that an adversary cannot leverage a \clio computation to 
illegitimately produce a signed value. 

A digital signature scheme is secure if it is difficult
to forge signatures of messages. 
%
%
\clio requires its digital signature scheme to be secure against \emph{existential forgery} under a \emph{chosen-message attack}, where the adversary is a non-uniform ppt in the size of the key. Often stated informally
in the literature~\cite{goldwasser-bellare}, 
a digital signature scheme is secure against 
\emph{existential forgery} if no adversary can succeed in forging the signature of one message, not necessarily of his choice. Further, the scheme is secure under a \emph{chosen-message attack} if the adversary is allowed to ask the signer to sign a number of messages of the adversary’s choice. The choice of these messages may depend on previously obtained signatures.

Parallel to CPA and CTA,
we adapt the definition of existential forgery for \clio, which
we call \emph{leveraged forgery}.
Intuitively, it should not be the case that a high integrity signature can be produced for a value when it is influenced by low integrity information.
We capture this intuition in the following theorem:

\begin{theorem}[Leveraged Forgery]
For a principal $p$ and all keystores $\keystore_0$, non-uniform ppts $\adversary$, and labels $l_1$, integers $j, j'$, where $\lowlabel = \authorityOf{\keystore_0}$ and $\canFlowToI{\projI{l_1}}{p}$, if $\cryptosys$ is secure against existential forgery under chosen-message attacks, then
{\em
$$
\figForgeryProb
$$} 
\end{theorem}
Intuitively, the game is structured as follows.
First, an adversary chooses a term $t$ and strategy $\strategy$ that will be run with high integrity (i.e., $\Start{\keystore}$ where $\keystore$ has $p$'s authority). The adversary sees the interactions $\interactions$ produced by the high integrity computation (which in general will include high integrity signatures). 
%

With that information, the adversary constructs a new term $t'$ and new strategy $\strategy'$ that will be run with low integrity (i.e., $\Start{\keystore_0}$). Note that the strategy may internally encode high integrity signatures learned from the high integrity run that it can place in the store.
%

The interactions produced by this low integrity computation should not contain any high integrity signatures (i.e., are signed by $p$). The adversary succeeds if it produces
a new valid labeled bitstring $\Labeled{l_1}{b}$ 
that did not exist in the first run. 
In the experiment, the $\mathrm{Values}_{\keystore}$ metafunction
extracts the set of valid labeled bitstrings (i.e., can be deserialized
correctly) using the parameterized keystore $\keystore$ to perform
the category key decryptions.

\bigskip

The proof of this theorem is in two parts.
First we show that the label of a value being stored
by a computation is no more trustworthy than the current label of
computation. Second, we show that the current label never
becomes more trustworthy than the starting label.
This means that a low integrity execution (i.e., starting from
$\conf{\Start{\keystore_0}}{\Clr{\keystore}}{t}$) cannot
produce a high integrity value (i.e., a labeled value $\Labeled{l}{b}$
such that $\projI{l} \flows^I p$).


\section{\clio in Practice}
\label{sec:impl}


\subsection{Implementation}


We implemented a \clio prototype as a Haskell library, in the same
style as LIO. Building on the LIO code base, the \clio library has
 an API for defining and running \clio programs embedded in
Haskell. The library also implements a monitor that oversees the
execution of the program and orchestrates three interdependent tasks:

\begin{CompactItemize}
\item \textbf{Information-flow control} \clio executes the usual LIO IFC
  enforcement mechanism; in particular, it adjusts the current label
  and clearance and checks that information flows according
  to the DC labels lattice.
\item\textbf{External key-value store} \clio handles all interactions with the
  store, realized as an external Redis~\cite{redis} database. This is
  accomplished by using the hedis~\cite{hedis} Haskell library, which
  implements a Redis client.
  \todo{reviewer B: I’m having a hard time imagining how a more realistic system would be used in practice. Do you imagine the Redis is simply outsourced to a generic remote service? What about practical details like distributing the public keys initially, or adding new participants?}
\item \textbf{Cryptography} \clio takes care of managing and handling
  cryptographic keys as well as invoking cryptographic operations to
  protect the security of the principals' data as it crosses the
  system boundary into/back from the untrusted store. Instead of
  implementing our own cryptographic primitives, we leverage the
  third-party cryptonite~\cite{cryptonite} library.
\end{CompactItemize}

\clio uses standard cryptographic schemes to protect the information in
the store. In particular, for efficiency reasons we use a hybrid
scheme that combines asymmetric cryptography with symmetric
encryption. The category keys in the store are encrypted and signed
with asymmetric schemes, while the entries stored by \clio programs are
encrypted with symmetric encryption and signed with an asymmetric
signature scheme.

\noindent\textbf{Asymmetric cryptography}
We use cryptonite's implementation
of RSA, specifically OAEP mode for encryption/decryption and PSS for
signing/verification, both with 1024-bit keys and using SHA256 as a
hash. We get around the message size limitation by chunking the
plaintext and encrypting the chunks separately.

\noindent\textbf{Symmetric encryption}
We use cryptonite's implementation of
AES, specifically AES256 in Counter (CTR) mode for symmetric
encryption. We use randomized initialization vectors 
for each
encryption. 
We can use AESNI if the
architecture supports it.

Storing and retrieving category keys and labeled values are
implemented as discussed in Section~\ref{subsec:keystore}. Appendix~\ref{app:impl} has more
details. 

\todo{P:Maybe move this somewhere else?}
\noindent\textbf{Performance}
LIO-style enforcement mechanisms have performed adequately in
practice, c.f. Hails~\cite{giffin:2012:hails}. We do not expect
combining this with off-the-shelf crypto to introduce more than a
constant time overhead for fetching and writing into the store. The
only additional concern is the overhead of the category key management
protocol, which is proportional to the number of distinct categories
and their size. Based on the experience obtained by Jif~\cite{jif},
Fabric~\cite{Liu:2009:FPS:1629575.1629606}, and Hails, categories are
usually small in number and size. Furthermore, creating category keys
incurs a one-time cost which can be amortized over multiple runs and
programs


\subsection{Case Study} 
%

We have implemented a simple case study to illustrate how our
prototype \clio implementation can be used to build an application. In
this case, we have built a system that models a tax preparation tool
and its interactions with a customer (the taxpayer) and the tax
reporting agency, communicating via a shared untrusted store. We model
these three components as principals $C$ (the customer), $P$ (the
preparer) and $\IRS$ (the tax reporting agency). The actions of each
of these three principals are modeled as separate \clio computations
\texttt{customerCode}, \texttt{preparerCode} and \texttt{irsCode},
respectively. We assume that the store level $\lowlabel$ restricts writes to the store in confidential contexts,
i.e. $\ell = \ltrip{\bot}{\top}{S}$, where $S$ is the principal
running as the store.
In this scenario, we consider that the principals involved ($C$, $P$
and $\IRS$) trust each other and are trying to protect their data from
all other principals in the system (i.e., from $S$). 

The customer $C$ initially makes a record with his/her personal
information, including his/her name, social security number (SSN),
declared income and bank account details, modeled as the type
\texttt{TaxpayerInfo}.  Figure~\ref{fig:customer-preparer-irs} shows the customer
code on the left, modeled as a function that takes this record as an argument,
\texttt{tpi}. The first step is to label \texttt{tpi} with the label
$\ltrip{C \vee P \vee \IRS}{C}{S}$.
The confidentiality component is a disjunction of all the principals
involved in the interaction, reflecting the fact that the customer
trusts both the preparer and the IRS with their the data and expects
them to be able to read it.
A more realistic example would also keep the customer's personal data
confidential (i.e. not readable by the IRS and to some extent by the
preparer). However, expressing those flows would require an IFC system
with declassification, a feature that we have not included in the
current version of \clio since it would introduce additional
complexity in our model, and semantic security conditions for such
systems are still an active area of research~\cite{AskarovS07,Broberg:2010,AskarovC12}.  Without
declassification, if $\IRS$ was not in the label initially, the IFC
mechanism would not allow us to release this data (or anything derived
from it) to the IRS at a later time.
The integrity component of this label is just $C$ since this
data can be vouched for only by the customer at this point, while the
availability is trusted since these values haven't been exposed to (and potentially
corrupted by) the
adversary in the store yet. The final step of the customer is to store
their labeled \texttt{TaxpayerInfo} at key \verb|"taxpayer_info"| for
the preparer to see. Note that in practice this operation creates a
category key for $C\vee P\vee \IRS$, stores it in the database and
uses it to encrypt the data, which gets signed by $C$.

\begin{figure*}[t]
\begin{minipage}[t]{0.32\textwidth}
\vspace{0pt}
\footnotesize
$
\begin{array}{l}
\texttt{customerCode} :: \texttt{TaxpayerInfo} \to \CLIOT{()}  \\
\texttt{customerCode}\;\;\; \mathit{tpi} = \mathbf{do}  \\
\begin{array}{ll}
 \quad &\mathit{info} \leftarrow \llabel{\ltrip{C\vee P\vee \IRS}{C}{S}}{\mathit{tpi}} \\
 & \store{\texttt{"taxpayer\_info"}}{\mathit{info}} \\
 & \return{()}
\end{array}
\end{array}
$
\end{minipage} \hfill
\begin{minipage}[t]{0.34\textwidth}
\vspace{0pt}
\footnotesize
$
\begin{array}{l}
\texttt{preparerCode} :: \CLIOT{()}  \\
\texttt{preparerCode} = \mathbf{do}  \\
\begin{array}{ll}
\quad & \mathit{default} \leftarrow
        \llabel{\ltrip{P\vee\IRS}{P\vee
        C}{S}}{\mathit{notFound}} \\
  & \mathit{info} \leftarrow
    \fetch{\texttt{"taxpayer\_info"}}{\mathit{default}} \\
  & r \leftarrow \toLabeled{\ltrip{P\vee\IRS}{P\vee
        C}{S}} \$\;\; \mathbf{do} \\
& \begin{array}{ll}
 & i \leftarrow \unlabel{\mathit{info}} \\
  & \return{(\texttt{prepareTaxes}\;\; i)}
\end{array} \\
 &  \store{\texttt{"tax\_return"}}{r}
\end{array}
\end{array}
$
\end{minipage} \hfill
\begin{minipage}[t]{0.32\textwidth}
\vspace{0pt}
\footnotesize
$
\begin{array}{l}
\texttt{irsCode} :: \CLIOT{\BoolT}  \\
\texttt{irsCode} = \mathbf{do}  \\
\begin{array}{ll}
\quad & \textbf{let}\;\; \mathit{l} = \ltrip{\IRS}{P\vee C\vee
        \IRS}{S} \\
 & \mathit{default}\leftarrow
        \llabel{l}{\mathit{emptyTR}} \\
 & \mathit{lv} \leftarrow
   \fetch{\texttt{"tax\_return"}}{\mathit{default}} \\
 & \mathit{tr}  \leftarrow \unlabel{\mathit{lv}} \\
 & \return{(\texttt{verifyReturn}\;\; \mathit{tr})}
\end{array}
\end{array}
$
\end{minipage}
  \caption{Customer code (left), Preparer code (middle), and IRS code (right)}
  \label{fig:customer-preparer-irs}
\end{figure*}




The next step is to run the preparer code, shown in the middle of
Figure~\ref{fig:customer-preparer-irs}. The preparer starts by fetching the
taxpayer data at key \verb|"taxpayer_info"|, using a default empty
record labeled with $l_1 = \ltrip{P\vee\IRS}{P\vee C}{S}$. The
entry in the database is labeled differently with
$l_2 =\ltrip{C\vee P\vee \IRS}{C}{S}$, but the operation succeeds
because $l_2 \sqsubseteq l_1$ and the availability
in $l_2$ is $S$, i.e., it reflects the fact that the adversary $S$ might have
corrupted this data. 
The code then starts a $\texttt{toLabeled}$
sub-computation to securely manipulate the labeled taxpayer record
without raising its current label. In the subcomputation, we unlabel
this labeled record and use function \texttt{prepareTaxes} to prepare
the tax return. Since we are only concerned with the information-flow
aspects of the example, we elide the details of how this function
works; our code includes a naive implementation but it would be
straightforward to extend it to implement a real-world tax preparation
operation. The $\text{toLabeled}$ block wraps the result in a labeled
value $r$ with label $l_1$, the argument to
$\text{toLabeled}$. Finally, the preparer stores the labeled tax
return $r$ at key \verb|"tax_return"|. Note that this operation would
fail if we had not used $\text{toLabeled}$, since in that case the
current label, raised by the $\text{unlabel}$ operation, would not
flow to $\ell$, the label of the adversary.

Figure~\ref{fig:customer-preparer-irs} shows the tax agency code
on the right. This code fetches the
tax return made by the preparer and stored at key
\verb|"tax_return"|. Analogously to the preparer code, we use the
default value of the fetch operation to specify the target label of
the result, namely $\ltrip{\IRS}{P\vee C\vee \IRS}{S}$, which in this case is
once again more restrictive than what is stored in the
database. Thereafter the labeled tax return gets unlabeled and the
information is audited in function \texttt{verifyReturn}, which
returns a boolean that represents whether the declaration is
correct. In a more realistic application, this auditing would be
performed inside a $\text{toLabeled}$ block too, but since we are not
doing any further $\text{store}$ operations we let the current label
get raised for simplicity.

These three pieces of code are put together in the main function of
the program, which we elide for brevity. This function simply
generates suitable keystores for the principals involved (using the
\clio library function \texttt{initializeKeyMapIO}) and then runs the
code for each principal using the \texttt{evalCLIO} function.

\section{Related Work}
\label{sec:related}

\todo{reviewer: expected more technical comparison to prior systems providing crypto soundness for lang-based IFC}
\todo{reviewer: proof technique seems similar to CoSP by Backes, Hofheinz, and Unruh}
\emph{Language-based approaches.}  Combining cryptography and IFC languages is
not new.
The Decentralized Label Model (DLM)~\cite{conf/sp/MyersL98} has been
extended with cryptographic
operations \cite{VaughanZS2007,SmithAR2006,ChothiaDV03,GazeauCSF17}. These
extensions, however, either use only symbolic models of cryptography
or provide no security properties for their system.


Models for secure declassification are an active area of research in the IFC
community (e.g.,~\cite{chong2006,myers2006,Askarov:2010,WayeBKCR2015}).
It is less clear, though, how such models compose with cryptographic attacker
models. Exploring the interactions between declassification and
cryptography is very interesting, but a rigorous treatment of it is beyond the
scope of this work.

Cryptographically-masked flows \cite{AskarovHS06}
account for
covert infor\-ma\-tion-flow channels due to the cryptosystem (e.g., an observer may
distinguish different ciphertexts for the same message).
%
However, this approach ignores the probability distributions for ciphertexts, which might
compromise security in some scenarios \cite{McLean1990}.
%
Laud~\cite{Laud2008} establishes conditions under which secure
programs with cryptographically-masked flows 
satisfy \emph{computational noninterference}~\cite{Laud:2001}.
%
Fournet and Rezk \cite{FournetRT2008} describe a language that directly embeds
cryptographic primitives and provide a language-based model of correctness,
where cryptographic games are encoded in the language itself so that security
can range from symbolic correctness to computational noninterference.


Information-flow availability has not been extensively studied.
%
Li et al.~\cite{LiMZ2003} discuss the relationship between availability and
integrity 
and state a (termination- and
progress-insensitive) noninterference property for availability.
Zheng and Myers \cite{zm05} extend the DLM with availability policies, which
express which principals are trusted to make data available.
%
%
In their setting, availability is, in essence, the integrity of
progress~\cite{Askarov:2008}: low-integrity
inputs should not affect the availability of high-availability outputs.
In our work, availability tracks the successful verification of signatures and
decryption of ciphertexts, and has analogies with Zheng and Myers' approach.

The problem of conducting proofs of trace-based properties of
languages with access to cryptographic operations in a computational
setting has been studied before. 
CoSP~\cite{Backes:2009} is
 a framework for writing computational soundness proofs
of symbolic models. Their approach abstracts details such as
message scheduling and corruption models and allows for proofs of
preservation of trace properties to be developed in a modular fashion. \todo{What do we have to say about the comparison?}

\emph{Cryptographic approaches.} There is much work on how to
map principals and access policies to cryptographic keys.
Attribute-Based Encryption~\cite{abe} could be used to protect the confidentiality
of data for categories and would avoid the need for category keys
when encrypting and decrypting.
Ring signatures \cite{rivest2006leak} could be used to protect the integrity
of data for categories and would similarly avoid the need for category
keys when signing and verifying. We take the approach
of using simpler cryptographic primitives as they are more
amenable to our proofs. Additionally, as a benefit of taking a
language-based approach, \clio's ideal semantics is agnostic
to the choice of cryptosystem used.
From a user's perspective
the underlying cryptographic operations could be swapped out in favor of more
efficient cryptosystems without changing the semantics of the system
(provided the real semantics was shown separately
to provide CTA security and security against leveraged forgery).

There is also work on strengthening the guarantees of existing cryptosystems
to protect against more powerful adversaries, e.g.,
Chosen Ciphertext Attack~(CCA) \cite{pass-shelat} security for adversaries that can
observe decryptions of arbitrary ciphertexts. CCA security is needed in systems where an adversary can
observe (some of) the effects of decrypting arbitrary ciphertexts.
In contrast,
\clio's security guarantees are based on a very precise definition of the
adversary's power over the system. In particular it captures
that an adversary cannot observe anything about the decryptions
of confidential values due to IFC mechanisms, since the results of
such a decryption would be
protected by a label that is more confidential than an adversary would have
access to.
As a result, \clio requires only a CPA secure cryptosystem to be CTA secure.

\emph{Systems.}
%
DStar \cite{Zeldovich:2008} extends decentralized IFC in a distributed system.
Every DStar node has an \emph{exporter} that is responsible for communicating
over the network.
Exporters also establish the security categories trusted by a node via
private/public keys.
Fabric \cite{LiuGVQWM2009} is a platform and statically-checked fine-grained IFC language.  Fabric
supports the secure transfer of data as well as code~\cite{ArdenGLVAM2012}
through, in part, the use of cryptographic mechanisms.
In contrast to Fabric, \clio provides coarse-grained IFC and uses DC labels
instead of the DLM.
In contrast to both DStar and Fabric, this work establishes a formal basis for
security of the use of cryptography in the system.
The lack of a formal proof in both DStar and Fabric is not surprising, given
that they target more ambitious and complex scenarios 
(i.e.,
decentralized information-flow control for distributed systems).
%

\emph{Remote storage.}
While data can be stored and fetched cryptographically, information
 can be still leaked through \emph{access patterns}.
Private Information Retrieval protocols aim to avoid such leaks by hiding
queries and answers from a potentially malicious server \cite{Chor:1995:PIR}
similar to \clio's threat model.
%
%
For performance reasons~\cite{Sion07onthe,OlumofinG11},
some approaches rely on a small trusted execution
environment provided by hardware \cite{Ding2010,WangDDB06} that
provides the cryptographic support needed to obliviously query the data
store~\cite{Smith:2001,asonov2005querying,WilliamsS08}.
%
%
%
This technique can be seen in oblivious computing
\cite{Maas:2013}, online advertising \cite{Backes:2012}, and credit networks
\cite{privpay}
%
for clients which are benign or follow an strict access protocol.
If clients are malicious, however, 
attacker's code may leak information though access patterns.
%
%
We force communication with the store to occur in non-sensitive contests. In addition,
our
language-based techniques could be extended to require
untrusted code to follow an oblivious protocol.


\section{Conclusion}
\label{sec:conclusion}
\clio is a computationally secure coarse-grained dynamic information-flow control
library that uses cryptography to protect the confidentiality and integrity of
data. The use of cryptography is hidden from the language operations and is controlled
instead through familiar language constructs in an existing IFC library, LIO. Further,
we present a novel proof technique that combines standard programming language 
and cryptographic proof techniques to show the interaction between the
high-level security guarantees provided by information flow control
and the low-level guarantees offered by the cryptographic mechanisms are secure.
We also provide a prototype \clio implementation in the form of a Haskell library extending LIO to evaluate its practicality. We see \clio as a way for
programmers that are non-expert cryptographers
to use cryptography securely.

\begin{acks}
This material is based upon work supported by the National Science Foundation under Grant No.s 1421770 and 1524052. Any opinions, findings, and conclusions or recommendations expressed in this material are those of the authors and do not necessarily reflect the views of the National Science Foundation. This research is also supported by the Air Force Research Laboratory and the Swedish research agencies VR and STINT.
\end{acks}


\bibliographystyle{ACM-Reference-Format}
\bibliography{bib}
\balance


\clearpage
\appendix
\onecolumn

\section{Complete Definitions}
\label{appendix:defns}

\subsection{Security Lattice Orderings and Operators}

We use three lattices (Confidentiality, Integrity, and Availability) 
whose domains (named C, I, and A) are formulas of principals in conjuctive normal form. 
We 
define an information flow ordering $\flows$ (read as ``can flow to''). We
give the definitions of each below, where $\Rightarrow, \wedge, \vee$ are the usual classical logical connectives:

\subsubsection*{Confidentiality}
\[
\begin{array}{r@{\,}c@{\;\;}l}
  \canFlowToC{C}{C'}      &\iff&     C' \Rightarrow C \\
  \canFlowToStrictC{C}{C'}      &\iff&     \canFlowToC{C}{C'} \textand C \neq C' \\
  \cjoin{C}{C'}     &\iff&     C \wedge      C' \\
  \cmeet{C}{C'}     &\iff&     C \vee        C' \\
  \bot_C \equiv True    &    &     \top_C \equiv False
\end{array}
\]

\subsubsection*{Integrity}

\[
\begin{array}{r@{\,}c@{\;\;}l}
  \canFlowToI{I}{I'}      &\iff&     I \Rightarrow I' \\
  \canFlowToStrictI{I}{I'}      &\iff&     \canFlowToI{I}{I'} \textand I \neq I' \\
  \ijoin{I}{I'}     &\iff&     I \vee        I' \\
  \imeet{I}{I'}     &\iff&     I \wedge      I' \\
  \bot_I \equiv False   &    &     \top_I \equiv True
\end{array}
\]

\subsubsection*{Availability}

\[
\begin{array}{r@{\,}c@{\;\;}l}
  \canFlowToA{A}{A'}      &\iff&     A \Rightarrow A' \\
  \canFlowToStrictA{A}{A'}      &\iff&     \canFlowToA{A}{A'} \textand A \neq A' \\
  \ajoin{A}{A'}     &\iff&     A \vee        A' \\
  \ameet{A}{A'}     &\iff&     A \wedge      A' \\
  \bot_A \equiv False   &    &     \top_A \equiv True
\end{array}
\]

We define the security lattice of DC labels as a product lattice of the three individual lattices to form a security lattice whose domain is a triple of principal formulas ($l = \ltrip{C}{I}{A}$) and whose ordering is based on safe information flows. We also define a trust ordering $\trusts$ (read as ``at least as trustworthy as''):

\[
\begin{array}{r@{\,}c@{\;\;}l}
  \canFlowTo{\ltrip{C}{I}{A}}{\ltrip{C'}{I'}{A'}}      &\iff&   
    \canFlowToC{C}{C'} \textand \canFlowToI{I}{I'} \textand \canFlowToA{A}{A'}  \\
  \canFlowToStrict{\ltrip{C}{I}{A}}{\ltrip{C'}{I'}{A'}}      &\iff&   
    \canFlowToStrictC{C}{C'} \textor \canFlowToStrictI{I}{I'} \textor \canFlowToStrictA{A}{A'}  \\
 
  \canTrustTo{\ltrip{C}{I}{A}}{\ltrip{C'}{I'}{A'}}      &\iff&   
    \canFlowToC{C}{C'} \textand \canFlowToI{I'}{I} \textand \canFlowToA{A'}{A}  \\
  
  \join{\ltrip{C}{I}{A}}{\ltrip{C'}{I'}{A'}}     &\iff&    
    \ltrip{~\cjoin{C}{C'}~}{~\ijoin{I}{I'}~}{~\ajoin{A}{A'}~} \\
  \meet{\ltrip{C}{I}{A}}{\ltrip{C'}{I'}{A'}}     &\iff&    
    \ltrip{~\cmeet{C}{C'}~}{~\imeet{I}{I'}~}{~\ameet{A}{A'}~} \\
  \bot \equiv \ltrip{\bot_C}{\bot_I}{\bot_A}   &    &     \top \equiv \ltrip{\top_C}{\top_I}{\top_A}
\end{array}
\]

For convenience of avoiding pattern matching over the components of a label when needing to inspect an individual component, we define the following projection functions:

\[
\begin{array}{r@{\,}c@{\;\;}l}
 \projC{\ltrip{C}{I}{A}} &=& C \\
 \projI{\ltrip{C}{I}{A}} &=& I \\
 \projA{\ltrip{C}{I}{A}} &=& A 
\end{array}
\]

\clearpage

\subsection{Cryptography Background Definitions}

\begin{itemize}

\item {\bf Distribution Ensemble}: An ensemble of probability distributions is a sequence $\{X_n\}_{n \in \mathcal{N}}$ of probability distributions.

\item {\bf Negligible Function}: A negligible function is a function $f(x) : \mathbb{N} \rightarrow \mathbb{R}$ such that for every positive integer $c \in \mathbb{Z}^+$ there exists an integer $N_c$ such that for all $x > N_c$:

$$
\big| f(x) \big| < \frac{1}{x^c}
$$

\item {\bf Computationally Indistinguishable}: Let $\{X_n\}_n$ and $\{Y_n\}_n$ be distribution ensembles. Then we say that they are computationally indistinguishable $\approx$ if for any non-uniform PPT $\adversary$ the following function is negligible in $n$:

$$
\big| \probabilityOfCond{\adversary(x) = 1}{x \drawnFrom X_n}  - \probabilityOfCond{\adversary(y) = 1}{y \drawnFrom Y_n} \big|
$$.

\item {\bf Hybrid Argument}:  Let $X_1, ..., X_m$ be a sequence of probability distributions, where $m$ is polynomial in $n$. Suppose
there exists a distinguisher $D$ that distinguishes $X_1$ and $X_m$ with probability $\epsilon$. Then there exists $i \in [1, m - 1]$ such that $D$ distinguishes $X_i$ and $X_{i+1}$ with probability $\frac{\epsilon}{m}$. (Proof uses Triangle Inequality)

\item {\bf Cryptosystem}: We consider a asymmetric encryption system and signature scheme $\Pi = (\Gen, \Enc, \Dec, \Sign, \Verify)$ such that $\Gen : 1^n \rightarrow (\{0, 1\}^n, \{0, 1\}^n)$ representing the public key and private key to use in the assymetric cryptographic functions.

\item {\bf Correctness of Cryptosystem} For correctness of our encryption system, we require that if $p \in \{0, 1\}^*$ and $(pk, sk) \drawnFrom \Gen(1^n)$ and $c = \Enc(pk, p)$ then $p = \Dec(sk, c)$.

For correctness of our signature scheme, we require that if $p \in \{0, 1\}^*$ and $(pk, sk) \drawnFrom \Gen(1^n)$ and $s = \Sign(sk, p)$ then $1 = \Verify(pk, s)$.

To simplify the notation while mainting easy-to-prove security properties, we require that the encryption and signature functions operate on independent parts of the key; that is, they internally use a key derivation function (e.g., encryption uses the first half, and signing uses the second half).

\item {\bf Security of Cryptosystem}: For our encryption functions, we assume the cryptosystem is CPA secure, defined as follows~\cite{pass-shelat}.
Let the random variable $\mathrm{IND}_b(\adversary, n)$ denote
the output of the   experiment, where
$\adversary$ is a non-uniform p.p.t., $n \in \varmathbb{N}$, $b \in \{0,1\}$ :
\noindent
$$
\figCPAGame
$$

Then we say that $\Pi$ is CPA (Chosen-Plaintext Attack) secure if for all
non-uniform p.p.t. $\adversary$:
$$
\Big\{~\mathrm{IND}_0(\adversary, n)~\Big\}_n \approx \Big\{~\mathrm{IND}_1(\adversary, n)~\Big\}_n
$$
Note that we consider the cryptographic functions themselves to be public.

For security of the signature scheme, we assume $\Pi$ is secure against existential forgery.

\item {\bf Indistinguishability Corrollary}: The CPA definition may be difficult to understand to some as it is phrased in the form of a game. An alternative definition of security (that is a fairly direct consequence of CPA security) is: For all $p_0, p_1$, if $|p_0| = |p_1|$ then,

$$
\Big\{ ~\Enc(pk, p_0) ~|~ (pk, sk) \drawnFrom \Gen(1^n)~ \Big\}_n \approx 
  \Big\{ ~\Enc(pk, p_1) ~|~ (pk, sk) \drawnFrom \Gen(1^n)~ \Big\}_n
$$
Informally, the results of encrypting of equal-length plaintexts are computationally indistinguishable.

\item {\bf Digital Signature Forgery}: We require our digital signature scheme to be secure against \emph{existential forgery} under a \emph{Chosen-Message Attack}, where the adversary is a non-unfiorm ppt in the size of the key~\cite{goldwasser-bellare}:

\begin{itemize}
\item \emph{Existential Foregery}: The adversary succeeds in forging the signature of one message, not necessarily of his choice.

\item \emph{Chosen-Message Attack}: The adversary is allowed to ask the signer to sign a number of messages of the adversary’s choice. The choice of these messages may depend on previously obtained signatures. For example, one may think of a notary public who signs documents on demand.
\end{itemize}

\end{itemize}

\clearpage
\subsection{\lioS Complete Syntax and Typing Rules}

\label{sec:typing}

\begin{align*}
\figLioSyntaxFull
\end{align*}

\begin{mathpar}
\figLioTypingRules
\end{mathpar}

\clearpage
\subsection{\lioS Remaining Step Rules}
\label{sec:fullsem}

The program state is $c = \conf{\lcur}{\lclr}{t}$ where $\lcur$ is the
current label, $\lclr$ is the current clearance. 
Computation is modeled as a small-step semantics $c \ltol{\event~} c'$.
We use labels $\event$ to represent interaction with the store. 

\figLioEvalContextFull

\begin{mathpar}
\figLioSemanticsFull
\end{mathpar}

Where $\lowlabel$ is the store adversary level. 


\clearpage

\subsection{Complete Low Equivalence Relation}
\label{subsec:complete-low-equiv}

\medskip

\begin{tabular}{lcll}
\figLowEquivFull
\end{tabular}
\vspace{-0em}
\subsubsection*{Confidentiality Low Equivalence}
Define \lowEquivC to be the low equivalence relation with respect to
only confidentiality. The integrity and availability parts of 
the label are ignored. 

\medskip

{\renewcommand{\arraystretch}{1.2}
\begin{tabular}{lcll}
  $\gv_1$ & $\lowEquivC$ & $\gv_2$ & if $\gv_1 = \gv_2$ \\
  $\LabeledO{l_1}{\gv_1}$ & $\lowEquivC$ & $\LabeledO{l_1}{\gv_2}$& where \\
  
\multicolumn{4}{>{\hspace{1.5em}}l}{
 $ 
\left\{   
\begin{array}{lll}
  \gv_1 = \gv_2 & \mathrm{if} & \projC{l_1} \flows^C \projC{\lowlabel} \\
   \typeOf{\gv_1} = \typeOf{\gv_2} & \multicolumn{2}{l}{\mathrm{otherwise} }
\end{array}
\right.
$} \\
  \\
  $(v_1, v_2)$ & $\lowEquivC$ & $(v_1', v_2')$& if $v_1 \lowEquivC v_1'$ and $v_2 \lowEquiv v_2'$ \\
  $t_1\ t_2$ & $\lowEquivC$ & $t_1'\ t_2'$& if $t_1 \lowEquivC t_1'$ and $t_2 \lowEquiv t_2'$\\
  \multicolumn{3}{c}{...}
  \\ 
  \multicolumn{3}{l}{{\bf Configurations:}}\\
  $\conf{\lcur}{\lclr}{t}$ & $\lowEquivC$ & $\conf{\lcur'}{\lclr'}{t'}$ & if $t \lowEquivC t'$ and 
  $\lcur = \lcur'$ and  $\lclr = \lclr'$ \\
  $\conf{\lcur}{\lclr}{t}$ & $\lowEquivC$ & $\conf{\lcur'}{\lclr'}{t'}$ & if $\projC{\lcur} \not\flows^C \projC{\lowlabel}$ and $\projC{\lcur'} \not\flows^C \projC{\lowlabel}$  and $\projC{\lclr} \not\flows^C \projC{\lowlabel}$ and $\lclr' \not\flows^C \projC{\lowlabel}$\\

\end{tabular}
}





\clearpage

\subsection{Ideal \clio Complete Syntax and Semantics}
\label{sec:ideal-storesem}
The ideal \clio state is $\iconf{c}{\istore}$ where $c$ is the
\lioS configuration and $\istore$ is a mapping 
$\istore : \gv \rightarrow \Labeled{l}{\gv}_{\missing}$, where $\missing$
represents a corrupted entry.

\begin{mathpar}
\figIdealClioFull
\end{mathpar}

Ideal interactions $\iinteraction$ are given by the following syntax:

\bigskip

\figIdealClioInteractions

The low-step relation $\Lowstep{\iinteractionsvar{\iinteraction}}{}$ from 
ideal \clio configurations to \clio configurations using adversary interaction $\iinteractions$.

\begin{mathpar}
\figIdealClioLowStep
\end{mathpar}

\clearpage

\subsection{Complete Definitions for Labeled Value Serialization}
\label{appendix:serialization}
\todo{reviewer: doesn't account for version indexes, which should be needed for Theorem 2}

\begin{itemize}
  \item $\initializecknm(\istore, C) = (\fetchcknm(\istore, C), \skipcmd)$ if $\fetchcknm(\istore, C)$ defined
  \item $\initializecknm(\istore, C) = (\fetchcknm(R(\istore), C), R)$ if $\fetchcknm(\istore, C)$ undefined and $R = \createcknm(\istore, C)$
  \item $\createcknm(\istore, p_1  \vee ... \vee p_n) = \mathrm{store}~(p_1 \vee ... \vee  p_i \vee ... \vee p_n)~(pk, m_n, s)$ where
  \\
  \hspace*{1cm}$(pk, sk) \drawnFrom \mathrm{Gen}(1^n)$ \\
  \hspace*{1cm}$m_0 = \emptyset$ \\
  \hspace*{1cm}for $i$ from $1$ to $n$:\\
  \hspace*{1.5cm}$(pk_i, sk_i) = \keystore(p_i)$\\
  \hspace*{1.5cm}$m_i \drawnFrom m_{i-1}[p_i \mapsto \Enc(pk_i, sk)]$\\
  \hspace*{1.5cm}$s \drawnFrom \Sign(sk_i, (pk, m))$ if $sk_i \neq \bot$
\item $\fetchcknm(\istore, p_1 \vee ... \vee p_i \vee ... \vee p_n) = (pk, sk)$ where
  \\
  \hspace*{1cm}$(pk, m, s) = \istore(p_1 \vee ... \vee p_i \vee ... \vee p_n)$ \\
  \hspace*{1cm}$(pk_i, sk_i) = \keystore(p_i)$ where $i$ chosen s.t. $sk_i \neq \bot$\\
  \hspace*{1cm}$sk = \Dec(sk_i, m(p_i))$ \\
  \hspace*{1cm}$\Verify(pk_j, (pk, m), s) = 1$ for some $pk_j$ in the category.

\end{itemize}

\begin{itemize}
\item $\Enc_{\keystore}(\istore_0, \langle C_1 \wedge ... \wedge C_i \wedge ... \wedge C_n, l_i \rangle, v_0) = (R_n \cdot ... \cdot R_1, v_n)$ where
  \\
  \hspace*{1cm} for $i$ from $1$ to $n$:\\
  \hspace*{1.5cm} $((pk_i, sk_i), R_i) \drawnFrom \initializecknm(\istore_{i-1}, C_i)$ \\
  \hspace*{1.5cm} $\istore_i = R_i(\istore_{i-1})$\\
  \hspace*{1.5cm} $v_i \drawnFrom \Enc(pk_i, v_{i-1})$

\item $\Sign_{\keystore}(\istore_0, \langle l_c, C_1 \wedge ... \wedge C_n \rangle, v) = (R_n \cdot ... \cdot R_1, s_1, ..., s_n)$ where
  \\
  \hspace*{1cm} for $i$ from $1$ to $n$:\\
  \hspace*{1.5cm} $((pk_i, sk_i), R_i) \drawnFrom \initializecknm(\istore_{i-1}, C_i)$ \\
  \hspace*{1.5cm} $\istore_i = R_i(\istore_{i-1})$\\
  \hspace*{1.5cm} $s_i \drawnFrom \Sign(sk_i, v)$

\item $\Dec_{\keystore}(\istore, \langle C_1  \wedge ... \wedge C_n, l_i \rangle, v_n) = v_0$ where
  \\
    \hspace*{1cm} for $i$ from $n$ to $1$:\\
  \hspace*{1.5cm} $(pk_i, sk_i) = \fetchcknm(\istore, C_i)$ \\
  \hspace*{1.5cm} $v_{i-1} = \Dec(sk_i, v_{i})$ \\

\item $\Verify_{\keystore}(\istore, \langle l_c, C_1  \wedge ... \wedge C_n \rangle, v, s_1, ..., s_n) = 1$ if
  \\
    \hspace*{1cm} for $i$ from $n$ to $1$:\\
  \hspace*{1.5cm} $(pk_i, sk_i) = \fetchcknm(\istore, C_i)$ \\
  \hspace*{1.5cm} $1 = \Verify(pk_i, v, s_i)$

\end{itemize}

\bigskip
\bigskip

Similar to the category key meta-functions, we also annotate the results of the meta-functions with the interactions made on the store so that we can track what actions are being taken on the crypto store.

With these cryptographic functions operating on labels, we are now ready to describe the meta-functions which convert a labeled value to a bit string and vice-versa.

\begin{itemize}
\item $\serializenm_{\keystore}(\istore, \LabeledO{l}{\gv}) = \randomExpr{(\overline{R_2} \cdot \overline{R_1} , \LabeledO{l}{\bitstream})}{(\overline{R}, b) \drawnFrom \Enc_{\keystore}(\overline{R_1}(\istore), l, (\gv, s_1, ..., s_n));~ (\overline{R_1}, s_1, ..., s_n) \drawnFrom \Sign_{\keystore}(\istore, l, \gv) }$
\item $\deserializenm_{\keystore}(\istore, \LabeledO{l}{\bitstream}, \tau) =  \LabeledO{l}{\gv}$ if $\Verify_{\keystore}(\istore, l, \gv, s_1, ..., s_n) = 1$ and $\Dec_{\keystore}(\istore, l, b) = (\gv, s_1, ..., s_n)$ and $\typeOf{\gv} = \tau$
\end{itemize}

We use the convenience function $\mathrm{pub}(\keystore)$ to represent the projection of the keystore that only contains the public key parts of the keystores, and no private keys.

\clearpage

\subsection{Real \clio Complete Syntax and Semantics}
\label{appendix:real-syntax}
\figRealClioSyntax

\figRealClioInteractions

\clearpage
The interaction concatenation operation $\concat{\interaction}{\interaction}$ sequences interactions. We use the notation $\interactions = \concat{\interaction}{\concat{...}{\interaction}}$ to denote a sequence of interactions. 

The real \clio state is $\rconf{c}{\istore}{\interactionsdist}{\versions}$
where $c$ is the \lioS configuration, $\istore$ is the store.

\begin{mathpar}
\figRealClioFull
\end{mathpar}

\bigskip

The low-step relation $\Lowstep{p}$ from 
real \clio configurations and adversary interactions to \clio configurations with probability $p$.

\begin{mathpar}
\figRealClioLowStep
\end{mathpar}

\noindent The step function encodes the distribution of real \clio states after taking $j$ low steps:

\noindent
\figStepFunctionDef
\medskip

Note that, we consider only configurations and strategies that can and always will take at least $j$ low steps for all strategies. That is, there is no possibility for a trace to fail to make a low step before $j$ low steps. As a result, the step function will always produce a distribution (i.e., their probabilities will add up to 1). The program should be written in such a way that it is defensively written to ensure that it can take at least $j$ low steps.

\clearpage

\section{Complete Theorems and Proofs}
\label{appendix:proofs}

\subsection{\clio Interaction Indistinguishability Lemmas}

\begin{lemma}[Round 1: Multi-Message Security]
For all $m_{\{1,2\}}^1 ... m_{\{1,2\}}^n$ and all principals
$\principals$, if $|m^0_i| = |m^1_i|$ for all $1 \leq i \leq k$, and
$\cryptosys$ is CPA Secure, then

\scalebox{0.9}{%
\vbox{
\[
\begin{array}{c}
\big\{ \Enc(pk^1_1, m^1_1),  ~...,~ \Enc(pk^{\,i}_1, m^{\,i}_1), ~  ..., ~\Enc(pk^k_1, m^k_1) ~\big|~ \keystore \drawnFrom \Gen(\principals, 1^n); ~ (pk^i_2, sk^i_2) \in rng(\keystore);  ~ 1 \leq i \leq k \big\}_n \\
\approx \\
\big\{ \Enc(pk^1_2, m^1_2),  ~...,~ \Enc(pk^{\,i}_2, m^{\,i}_2), ~  ..., ~\Enc(pk^k_2, m^k_2) ~\big|~ \keystore \drawnFrom \Gen(\principals, 1^n); ~ (pk^i_2, sk^i_2) \in rng(\keystore); ~ 1 \leq i \leq k \big\}_n \\
\end{array}
\]
}}
\end{lemma}

\begin{proof}

We perform a proof by contradiction: we assume the consequent does not hold
and construct a counter-example to the CPA-security of $\cryptosys$.



The lengths of the sequences of encryptions are equal by setup. The sequences are also polynomial in $n$ as each low step produces a polynomial number of messages and the number of low steps is polynomial in $n$ as $k$ is fixed.

\medskip

This setup is equivalent to the multi-message CPA security problem. We use the same general technique to show that single-message CPA security. gives rise to multi-message CPA security by using a hybrid argument.

\medskip

We can define a hybrid sequence of messages

$$
H_i = \big\{ m^1_1;  ~...;~ m^{\,i}_1; ~   m^{\,i+1}_2; ~  ...; ~m^k_2 \} 
$$

such that at the point $i$ we switch from using the sequence of messages from the first run to using the sequence of messages from the second run. By the hybrid argument, there must exist an $i$ that distinguishes $H_i$ and $H_{i+1}$ with non-negligible probability in $n$. We will fix on that $i$. So we can now construct the following CPA adversary:
\begin{enumerate}
\item By assumption of our proof by contradiction, there exists a Round 1 adversary that can distinguish $H_i$ and $H_{i+1}$ for a particular $i$ and particular $m^{\{0,1\}}_1, ..., m^{\{0,1\}}_n$, and call it $\adversary_{R1}$.
\item In our construction, we generate a new keystore $\keystore$ as defined in the lemma statement and give the plaintext messages $m^{i+1}_2$ and $m^{i+1}_1$ to the CPA game and it will then provide us back a ciphertext $c$ of one of the messages.
\item We create a new sequence of encrypted messages in the following way:

\scalebox{0.8}{%
\vbox{
$$
H = \big\{ \Enc(pk^1_1, m^1_1);  ~...;~ \Enc(pk^i_1, m^{\,i}_1); ~  c; ~ \Enc(pk^{i+2}_2, m^{\,i+2}_2); ~  ...; ~\Enc(pk^k_2, m^k_2) ~\big|~ pk^o_{\{1,2\}} \drawnFrom \Gen(1^n); 1 \leq o \leq k \} \} 
$$
}}

In the case where the CPA game chose message $m^{i+1}_2$ we have that $H = H_i$ and in the other case where $m^{i+1}_1$ we have that $H = H_{i+1}$. Since $\adversary_{R1}$ can distinguish exactly this case and that the choice of message encrypted determines which sequence of messages was used, we can as a result distinguish which plain-text message was chosen by the CPA game with non-negligible probability.

\end{enumerate}

As a result of constructing a CPA adversary that can distinguish plain-text messages with non-negligible probability, we have shown a contradiction, and can conclude that the above sequences of encryptions is indistinguishable. 

\end{proof}

\clearpage

\begin{definition}[Low Equivalent Interactions]
\normalfont
Let $L_{\keystore}(\interactions)$ to be a fixed function (i.e., it does not change its behavior based on its inputs) from interactions to interactions such that the result contains the original sequence of interactions
with \emph{low interactions} added at statically fixed locations in the sequence. Let the resulting sequences of interactions be called \emph{low equivalent interactions}.

\bigskip

A low interaction is a $\skipcmd$ command or a $\storecmd{\gv_k}{\Labeled{l_1}{b}}$ such that \\ $\projC{l_1} \flows^C \projC{\authorityOf{\keystore}} \textand
 \Labeled{l_1}{(\gv, \gv_k, n)} = \deserializenm_{\keystore}({\interactions},{\Labeled{l_1}{b}})$ or $\storecmd{\category}{\principals'}$.
\end{definition}

\begin{lemma}[Round 2: Secret and Low Equivalent Interactions]
For all keystores $\keystore_0$, and
$l_1, ..., l_k$, such that $\canFlowToC{\projC{l_i}}{\projC{\authorityOf{\keystore}}}$, and for all $m_{\{1,2\}}^1 ... m_{\{1,2\}}^n$ and all principals $\principals$, if $|m^i_1| = |m^i_2|$ for all $1 \leq i \leq k$ and $\cryptosys$ is CPA Secure, then 

\scalebox{0.8}{%
\vbox{
\[
\begin{array}{c}
\big\{ ~ L_{\keystore_0}\Big(\storecmd{\gv^1}{\Labeled{l^1}{b^1_1}} \cdot  ~...~ \cdot \storecmd{\gv^k}{\Labeled{l^{k}}{b^{k}_1}} \Big) ~|~ \keystore \drawnFrom \Gen(1^n); ~ (pk^i, sk^i) \in rng(\keystore); ~ b^i_1 \drawnFrom \Enc(pk^i, m^i_1); ~ 1 \leq i \leq k ~ \big\}_n \\
\approx \\
\big\{ ~ L_{\keystore_0}\Big(\storecmd{\gv^1}{\Labeled{l^1}{b^1_2}} \cdot  ~...~ \cdot~ \storecmd{\gv^k}{\Labeled{l^{k}}{b^{k}_2}} \Big) ~|~ \keystore \drawnFrom \Gen(1^n); ~ (pk^i, sk^i) \in rng(\keystore); ~ b^i_2 \drawnFrom \Enc(pk^i, m^i_2); ~ 1 \leq i \leq k ~ \big\}_n 
\end{array}
\]
}}
\label{lemma:round2}
\end{lemma}

\begin{proof}
We perform a reduction to the Round 1 adversary. That is, if there exists a Round 2 adversary $\adversary_{R2}$ then there also exists a Round 1 adversary, which will provide a contradiction.

We construct our adversary as follows. Given a sequence of secret encryptions (from Round 1), we can construct the input to the $L$ function by constructing a constant set of labels $l_1$ to $l_n$ arbitrarily so long as they flow to $\projC{\authorityOf{\keystore}}$, which is a static property and also arbitrary input keys. We can then also construct the entry keys $v^i$ in the same static fashion. When then just apply the deterministic $L$ function and pass that to $\adversary_{R2}$.

Since we have shown how to construct a Round 1 adversary from a Round 2 adversary, we have a contradiction of the Round 1 lemma, so we conclude our proof. 

\end{proof}

\clearpage

\begin{definition}[\clio Interactions]
\normalfont
Let $m^1_{\{1,2\}}, ..., m^k_{\{1,2\}}$ be sequences of messages such that $|m^i_1| = |m^i_2|$. Further, let $\interactions^1_{\{1,2\}}, ..., \interactions^j_{\{1,2\}}$ be sequences of low equivalent interactions (from Definition 2) whose ciphertexts are based on slices of the underlying message. For example

\[
\begin{array}{lcl}
\interactions^1_1 &=& \storecmd{\category}{(b, \{ p \mapsto \Enc(pk, m^1_1) \}, b_s)}  ~ \cdot ~ \storecmd{\gv_1}{\Enc(pk, m^2_1)} \\
\interactions^2_1 &=& \skipcmd \\
\interactions^3_1 &=& \storecmd{\gv_3}{\Enc(pk', m^3_1)} \\
\multicolumn{2}{c}{...} \\
\interactions^j_1 &=& \storecmd{\gv_j}{\Enc(pk'', m^k_1)} 
\end{array}
\]

Then we say two sequences of interactions are \emph{\clio equivalent} $\asymp$ iff they are of the form, with all but negligible probability, 

\scalebox{0.8}{%
\vbox{
\[
\begin{array}{rcl}
\big\{ ~ \interactions^{1}_{1} ~ \cdot ~ \interactions'^{2}_{1} \cdot ... \cdot ~ 
\interactions^{j}_{1}  ~ \cdot ~ \interactions'^j_1
 & ~\Big|~ & \keystore \drawnFrom \Gen(1^n); ~ (pk^i, sk^i) \in rng(\keystore); ~ b^i_1 \drawnFrom \Enc(pk^i, m^i_1); ~ 1 \leq i \leq k; ~ \\
 && t, v_{\{0,1\}}, \strategy \drawnFrom \adversary(1^n); ~
 \interactions'^{s}_1 \drawnFrom \strategy(\interactions^{s}_1 \cdot ... \cdot \interactions^{j}_1 \cdot \interactions'^{j}_1); ~ 1 < s < j; ~ \interactions'^{j}_1 \drawnFrom \strategy(\skipcmd)
   ~ \big\}_n \\
\multicolumn{3}{c}{\asymp} \\
\big\{ ~ \interactions^{1}_{2} ~ \cdot ~ \interactions'^{2}_{2} \cdot ... \cdot ~ 
\interactions^{j}_{2}  ~ \cdot ~ \interactions'^j_2
 & ~\Big|~ & \keystore \drawnFrom \Gen(1^n); ~ (pk^i, sk^i) \in rng(\keystore); ~ b^i_2 \drawnFrom \Enc(pk^i, m^i_2); ~ 1 \leq i \leq k; ~ \\
 && t, v_{\{0,1\}}, \strategy \drawnFrom \adversary(1^n); ~
 \interactions'^{s}_2 \drawnFrom \strategy(\interactions^{s}_2 \cdot ... \cdot \interactions^{j}_2 \cdot \interactions'^{j}_2); ~ 1 < s < j; ~ \interactions'^{j}_2 \drawnFrom \strategy(\skipcmd)
   ~ \big\}_n \\
\end{array}
\]
}}
\end{definition}

\begin{lemma}[Round 3: \clio Interactions Indistinguishability]
For all families of distributions $\interactionsdist_1$ and  $\interactionsdist_2$, if $\interactionsdist_1 \asymp \interactionsdist_2$ and $\cryptosys$ is CPA secure, then $\interactionsdist_1 \approx \interactionsdist_2$.
\label{lemma:clio-indist}
\end{lemma}

\begin{proof}
Similar to the previous rounds, we will reduce this problem to the Round 2 Secret and Low Equivalent Interactions indistinguishability problem. If there exists a Round 3 adversary $\adversary_{R3}$ then there also exists a Round 2 adversary, which will provide a contradiction.

We construct our adversary as follows. Given a sequence of low equivalent interactions (from Round 2) $\interactions_1$ and $\interactions_2$, we subdivide the interactions into sequences of interactions $\interactions^1_{\{1,2\}} ... \interactions^j_{\{1,2\}}$. We also add the categories in storing category keys arbitrarily in a static way and also add the category key signatures (as it was not in the previous round) by just performing the signing process according to the \initializecknm function. 

For the strategy interactions, we just perform the draws from the strategy starting from the end of the sequences of interactions, working backwards and place them in their corresponding positions, i.e., $\interactions'^{s} \drawnFrom \strategy(\interactions^{s} \cdot ... \cdot \interactions'^{j} \cdot \interactions^{j}$ for $1 < s < j$. We note that, although the interactions may differ they are indistinguishable. That is because the first sequence of interactions $\interactions^j_{\{1,2\}}$ is a sub-problem of the Round 2 sequences of interactions. As a result, since the two sequences of interactions are computationally indistinguishable, then their corresponding draws from the strategy are also computationally indistinguishable.

As a result, we can then pass the final sequence of interactions to $\adversary_{R2}$ to distinguish the distributions. 
Since we have shown how to construct a Round 2 adversary from a Round 3 adversary, we have a contradiction of the Round 2 lemma, so we conclude our proof. 

\end{proof}

\clearpage
\subsection{\clio Preservation of Low Equivalence}

\bigskip

\begin{lemma}[Preservation of Low Equivalence]
For \figInteractionLemma{}
\label{lemma:interaction}
\end{lemma}

\begin{proof} We will prove this lemma in two steps: first by showing that the invariant is preserved across \clio steps $\rto{p}^*$ and then using that fact, we can show that the invariant is also preserved across \clio low steps $\Lowstep{p}$. 

\bigskip
\bigskip

{\bf Proof on Step relation $\rto{p}$:} We will perform induction on the derivation of the steps (which will be finite when used with the low-step rules, i.e., it is well-founded) with the number of steps being $k$ being 1 less than the total number of (possibly high) steps in the context of a single low step.

\bigskip


Our inductive hypothesis will be if $\rconf{c_1}{\emptyset}{\skipcmd}{\versions_1} \lowEquivC \rconf{c_2'}{\emptyset}{\skipcmd}{\versions_2}$ and $\cryptosys$ is CPA Secure, then,




$ 
  \hspace*{1cm} \textbf{Pr}\Big[  
    \rconf{c_1'}{\istoredist_1}{\interactionsdist_1}{\versions_1}
    \not\lowEquivC
    \rconf{c_2'}{\istoredist_2}{\interactionsdist_2}{\versions_2}
    \textor
    \versions_1 \neq \versions_2
  ~~\Big| ~\newline \hspace*{2cm}
     ~\keystore \drawnFrom \Gen(\principals, 1^n); ~
      \keystore_1 = \keystore_0 \uplus \keystore; ~
   \rconf{c_1}{\emptyset}{\skipcmd}{\zeroes} \rto{p_1} ... \rto{p_k} 
    \rconf{c_1'}{\istoredist_1}{\interactionsdist_1}{\versions_1};  \newline
  \hspace*{2cm}
    \keystore' \drawnFrom \Gen(\principals, 1^n); ~
      \keystore_2 = \keystore_0 \uplus \keystore'; ~
    \rconf{c_2}{\emptyset}{\skipcmd}{\zeroes} \rto{p_1'} ... \rto{p_{k'}'} 
     \rconf{c_2'}{\istoredist_2}{\interactionsdist_2}{\versions_2}
  \Big]$ \\ ~\\
  \hspace*{0.5cm} is negligible in $n$, and  \newline ~ \newline
  $\hspace*{0cm} \Big\{  ~~ \interactions_1 ~~ \Big| ~~
   \keystore' \drawnFrom \Gen(\principals, 1^n); ~
   \keystore = \keystore_0 \uplus \keystore'; ~
  \rconf{c_1}{\emptyset}{\skipcmd}{\zeroes} \rto{p_1} ... \rto{p_k} 
    \rconf{c_1'}{\istoredist_1}{\interactionsdist_1}{\versions_1}; ~ \interactions_1 \drawnFrom \interactionsdist_1 ~ ~ \Big\}_n \newline
    \hspace*{7cm} \asymp \newline
    \hspace*{0.5cm} \Big\{  ~~ \interactions_2 ~~ \Big| ~~
   \keystore' \drawnFrom \Gen(\principals, 1^n);~ 
   \keystore = \keystore_0 \uplus \keystore';~
  \rconf{c_2}{\emptyset}{\skipcmd}{\zeroes} \rto{p_1'} ... \rto{p_{k'}'} 
     \rconf{c_2'}{\istoredist_2}{\interactionsdist_2}{\versions_2}; ~ \interactions_2 \drawnFrom \interactionsdist_2 ~ \Big\}_n $

\bigskip

. We must show as our inductive step that if our inductive hypothesis is true,
that the following is true:

\bigskip

$ 
  \hspace*{0.5cm} \textbf{Pr}\Big[  
    \rconf{c_1''}{\istoredist_1'}{\interactionsdist_1'}{\versions_1'}
    \not\lowEquivC
    \rconf{c_2''}{\istoredist_2'}{\interactionsdist_2'}{\versions_2'}
  ~~\Big| \newline \hspace*{2cm}
     ~\keystore \drawnFrom \Gen(\principals, 1^n); ~
      \keystore_1 = \keystore_0 \uplus \keystore; ~
   \rconf{c_1}{\emptyset}{\skipcmd}{\zeroes} \rto{p_1} ... \rto{p_k} 
    \rconf{c_1'}{\istoredist_1}{\interactionsdist_1}{\versions_1}  \rto{p_{k+1}} 
  \rconf{c_1''}{\istoredist_1'}{\interactionsdist_1'}{\versions_1'};  \newline
  \hspace*{2cm}
    \keystore' \drawnFrom \Gen(\principals, 1^n); ~
      \keystore_2 = \keystore_0 \uplus \keystore'; ~
    \rconf{c_2}{\emptyset}{\skipcmd}{\zeroes} \rto{p_1'} ... \rto{p_{k'}'} 
     \rconf{c_2'}{\istoredist_2}{\interactionsdist_2}{\versions_2}
      \rto{p'_{k'+1}} 
   \rconf{c_2''}{\istoredist_2'}{\interactionsdist_2'}{\versions_2'}
  \Big]$ \newline~ \newline~
  \hspace*{0.5cm} is negligible in $n$, and \newline~ \newline~
  $\hspace*{0cm} \Big\{  ~~ \interactions_1' ~~ \Big| ~~
   \keystore' \drawnFrom \Gen(\principals, 1^n); ~
   \keystore = \keystore_0 \uplus \keystore'; ~
  \rconf{c_1}{\emptyset}{\skipcmd}{\zeroes} \rto{p_1} ... \rto{p_k} 
    \rconf{c_1'}{\istoredist_1}{\interactionsdist_1}{\versions_1}  \rto{p_{k+1}} 
   \rconf{c_1''}{\istoredist_1'}{\interactionsdist_1'}{\versions_1'}; ~ \interactions'_1 \drawnFrom \interactionsdist'_1 ~ \Big\}_n \newline~
    \hspace*{8cm} \asymp \newline~
    \hspace*{0.5cm} \Big\{  ~~ \interactions'_2 ~~ \Big| ~~
   \keystore' \drawnFrom \Gen(\principals, 1^n);~ 
   \keystore = \keystore_0 \uplus \keystore';~
  \rconf{c_2}{\emptyset}{\skipcmd}{\zeroes} \rto{p_1'} ... \rto{p_{k'}'} 
     \rconf{c_2'}{\istoredist_2}{\interactionsdist_2}{\versions_2} 
      \rto{p'_{k'+1}} 
   \rconf{c_2''}{\istoredist_2'}{\interactionsdist_2'}{\versions_2'}; ~ \interactions'_2 \drawnFrom \interactionsdist'_2  ~ \Big\}_n $

\bigskip


The base case $k=0$ is direct as the initial interactions are a special case of \clio interactions (i.e., $\skipcmd = \skipcmd$) and $\zeroes = \zeroes$ and we already know by supposition that $c_1 \lowEquivC c_2$.

For the inductive case, we now consider the derivation rule used for the $k+1$'th step and show that it preserves the inductive hypothesis, assuming it for the $k$'th step. 

We note that the single steps may take differing numbers of steps (i.e., $k$ and $k'$). Due to the \rulename{Low-to-High-to-Low} step rule, though, these differences only occur when $\lcur$ is high. As a result, the only invariant we need to preserve is confidentiality-only low equivalence between configurations as the high steps do not change the versions, stores, or interactions. We can appeal to the preservation of low equivalence of LIO proved by Stefan et al.~\cite{stefan:2011:flexible} to conclude the preservation confidentiality-only low equivalence of the standard (i.e., non-store and non-fetch) LIO internal steps. We now consider the low steps that affect the non-standard parts of LIO (i.e., the store and fetch commands).

We also note that since there are only a polynomial number of steps that the resulting sequences of configurations from single steps that do not preserve low equivalent in each step will still together be negligible. As a result, we only need to show that the probability of each step not preserving low equivalence is negligible in $n$. To that end, we will ignore the traces of steps with negligible probabilities that are not low equivalent.

\begin{itemize}

\item {\bf Case Fetch (\rulename{Fetch-Exists}, \rulename{Fetch-Missing}, or \rulename{Fetch-Replay}):}

In this case, both configurations have a term with an evaluation context hole that is at a fetch command. That is, they are both attempting to fetch an entry from the store with key $\gv_k$ and $\gv_k'$. Due to low equivalence they are both fetching the same key, $\gv_k = \gv_k'$. As a result, they must each be using one of the following rules: \rulename{Fetch-Exists}, \rulename{Fetch-Missing}, or \rulename{Fetch-Replay}. 

We know by our inductive hypothesis that the distributions of interactions are \clio equivalent. Because the distributions are equal, we can consider steps where the draws are equivalent for values in the erased distributions. 
We now consider now each case that $c_1'$  transitions with.

\begin{itemize}
\item {\bf Case } \rulename{Fetch-Exists} and $\canFlowToC{\projC{l_1}}{\projC{\lowlabel}}$: In this case the labeled value will be deserialized the same way and the same labeled value will be fetched. Since the labeled value is readable by the adversary it must be syntactically equivalent with all but negligible probability, otherwise the interactions would be distinguishable which would be a counter-example to Lemma~\ref{lemma:clio-indist}. As a result, low equivalence will be preserved and both configurations will transition in the same way with all but negligible probability.

\item {\bf Case } \rulename{Fetch-Exists} and $\cantFlowToC{\projC{l_1}}{\projC{\lowlabel}}$: 
In general the value fetched from the store will vary, or it may be the case that only some of the time a value can even be deserialized.
In these cases, the configurations may transition using this rule and other interactions from the distribution may result in it using another rule. 
However, if $c_2'$ transitions using another fetch rule (\rulename{Fetch-Missing} or \rulename{Fetch-Replay}), 
the default value will be used. Since secret values can differ and still be low equivalent, 
the resulting two configurations will still be low equivalent.  In each of these cases, 
no new interactions are produced so the resulting distributions of interactions are still valid \clio interactions by our inductive hypothesis (as they did not change). 
As a result, low equivalence and the valid interactions invariant is preserved.

\item {\bf Case } \rulename{Fetch-Missing} or \rulename{Fetch-Replay}: In these cases, the default labeled value will be used, which by our inductive hypothesis is already low equivalent (due to the configurations being low equivalent).
The other configurations will transition in a symmetric way described for the \rulename{Fetch-Exists} rule.

\end{itemize}

\item {\bf Case \rulename{Store}:}

In this case the distribution of stores and interactions will change
so we must show that they remain equivalent. That is, we must show that $\{ ~ \interactions'_1 ~|~ \interactions'_ \drawnFrom \interactionsdist'_1 \}_n \asymp \{ ~ \interactions'_2 ~|~ \interactions'_2 \drawnFrom \interactionsdist'_2 \}$ where 

$$
\begin{array}{l<{\hspace{-2mm}}c<{\hspace{-2mm}}l<{\hspace{-2mm}}l}
    \interactionsdist_1' &=& \big\{ ~ & \concat{\concat{\storecmd{\gv_k}{\Labeled{l_1}{\bitstream_1}}}{\interactionsprime_1}}{\interactions_1} ~ ~ \Big|~ ~
      \interactions_1 \drawnFrom \interactionsdist_1;  \arcr
        &&&
        (\interactionsprime_1, \Labeled{l_1}{\bitstream_1}) \drawnFrom \serialize{\istore_1}{\Labeled{l_1}{(\gv, \gv_k, \version_1)}}
         ~ \big\} \arcr
      \\
      \interactionsdist_2' &=& \big\{ ~ & \concat{\concat{\storecmd{\gv_k}{\Labeled{l_1}{\bitstream_2}}}{\interactionsprime_2}}{\interactions_2} ~ ~ \Big|~ ~
      \interactions_2 \drawnFrom \interactionsdist_2;  \arcr
        &&&
        (\interactionsprime_2, \Labeled{l_1}{\bitstream_2}) \drawnFrom \serialize{\istore_2}{\Labeled{l_1}{(\gv, \gv_k, \version_2)}}
         ~ \big\} \arcr
    \end{array} 
$$

We first note that the entry keys are the same from low equivalence. The versions are equal from \clio low equivalence, i.e., $\version_1 = \version_2$. We also note that the distributions of interactions are valid \clio interactions from our inductive hypothesis. For readable labeled values, we can conclude that they are  syntactically equivalent values from low equivalence. For non-readable values, the types of the secret values will be the same due to low equivalence (and so the serialized plaintext message will have the same length). As a result the $\storecmd{\gv_k}{\Labeled{l_1}{b_{\{0,1\}}}}$ will be a valid extension of valid \clio interactions.

We next consider the creation of category keys (i.e.,
$\interactionsprime_1$ and $\interactionsprime_2$). The initialization of category keys will behave the same way as described for fetching a labeled value: it will either create new keys (if they were corrupted or not there), or $\skipcmd$. It will do this in the same way as the resulting interactions are indistinguishable. For the contents of the category keys, we can divide the parts of the category into deterministic parts (i.e., from Lemma~\ref{lemma:round2}) and secret encryptions. 

With these considerations, we conclude that with all but negligible probability the resulting interactions will be valid \clio interactions.


Finally, for the versions mappings, we note that they are both updated equivalently (i.e., incremented by the version in the mapping) and the versions mappings originally were equal, so the resulting versions are equal.




\end{itemize}


With all cases of the reduction shown to satisfy the proof obligation, we can conclude the inductive hypothesis is true for all steps used in the context of a single low step. We next show the low equivalence invariant on the low step relation.

\bigskip

{\bf Proof on Low-step relation $\Lowstep{p}$:} by induction on the number of low steps $j$. Our inductive hypothesis will match our lemma. 
For \figInteractionLemma{\figInteractionLemmaIH}

\bigskip

\begin{itemize}

\item {\bf Base Case: $j = 1$:} That is, we will prove the following:

\figInteractionProbStatement{1}{\figInteractionLemmaIH}



We must show $\rconf{c_1'}{\istoredist_1}{\interactionsdist_1}{\versions_1}
    \lowEquivC
    \rconf{c_2'}{\istoredist_2}{\interactionsdist_2}{\versions_2} \textor \versions_1 \neq \versions_2$. 
    There are two cases we must consider in the low step relation, the \rulename{Low-Step} rule and the \rulename{Low-to-High-to-Low-Step} rule. In the \rulename{Low-Step}, we must show that the inductive hypothesis holds after a single \clio step $\rto{p}$, and for the \rulename{Low-to-High-to-Low-Step} rule must hold for many (finite) \clio steps $\rto{p}^*$. Note that the \rulename{Low-Step} rule is a special case of the \rulename{Low-to-High-to-Low-Step} rule so we only consider the more general case of preserving the invariant across many steps. To show this, we appeal to the previous proof made to show that the invariant is preserved across \clio steps. 




Unlike the single step relation, this includes a strategy interaction on the distribution of stores. Since both interactions receive the indistinguishable  distributions of interactions (from Lemma~\ref{lemma:clio-indist} and the inductive hypothesis) so the resulting distributions from the strategy will also be computationally indistinguishable. That is because if they were not, then the strategy itself could be used as a counter-example for Lemma~\ref{lemma:clio-indist}. In sum, the resulting strategy invocation results in a valid sequence of \clio interactions. 

From the previous proof on the single-step relation, we can conclude that \\ \rconf{c_1'}{\istoredist_1}{\interactionsdist_1}{\versions_1}
    \lowEquivC
    \rconf{c_2'}{\istoredist_2}{\interactionsdist_2}{\versions_2}. 
    As a result, we satisfy the inductive hypothesis.

\item {\bf Inductive Case} $j=k+1$: That is, we will prove the following:

\figInteractionProbStatement{k+1}{\figInteractionLemmaIH}




We now must show that, for any low equivalent configurations that the resulting single step will remain low equivalent. 
We can use the same reasoning from the base case to show that the adversary interaction preserves equivalence on distributions of stores.
%
After this adversary interaction, we can invoke the single-step lemma result here to conclude that $c_1'' \lowEquiv c_2''$. 

\end{itemize}

With the single low step relation handled we now must consider the distribution of distributions of interactions from the $\multistepnm$ function. For example, it may be the case that a particular distribution of interactions generated from one trace of low steps may be much more likely than another distribution of interactions. However, we can use the preservation of low equivalence to reason about the probabilities of corresponding low equivalent distributions of interactions. 
We consider the distribution formed from the step function after 1 low step as a running example to make our arguments concrete, shown graphically in the main matter in \Cref{subsec:indist}.

From our inductive hypothesis we know that corresponding low equivalent configurations have indistinguishable distributions of interactions. We now must consider the relationship between the probabilities that led to the corresponding low equivalent configurations (e.g., from the diagram $p_1$ and $p_2$, and also $p_1'$ and $p_2'$). If they are similar, then the resulting draws from the distributions will be similar (from low equivalence). 

Consider the pairs of low equivalent configurations and the probabilities that led to those configurations (e.g., from the diagram $\rconf{c_1'}{\istoredist_1}{\interactionsdist_1}{\versions_1}$ with probability $p_1$ and $\rconf{c_2'}{\istoredist_2}{\interactionsdist_2}{\versions_2}$ with probability $p_2$). 
Consider the ways the configurations can differ probabilistically (e.g., from the diagram, how $c_1$ steps to both $c_1'$ and $c_1''$ and how $c_2$ steps to both $c_2'$ and $c_2''$). 
The low step relation is just the probability of the trace of single steps leading to the next low \clio configuration. The \rulename{Store} and \rulename{Internal-Step} rules take steps with probability $1$ so they will not cause the low step to differ probabilistically. 

Indeed, only the fetching rules \rulename{Fetch-Exists}, \rulename{Fetch-Missing}, and \rulename{Fetch-Replay} rules will cause the configurations to differ probabilistically. In particular they will differ based on the interactions drawn,  and as a result differ on how those interactions affect the fetch: if the entry is missing or not deserializeable (\rulename{Fetch-Missing}), if the value can be deserialized but the version is old (\rulename{Fetch-Replay}), or if it was successfully deserialized and the version is not old (\rulename{Fetch-Valid}). 

Due to our inductive hypothesis we know that the distributions of interactions are valid \clio interactions and as a result are indistinguishable from Lemma~\ref{lemma:clio-indist}. For readable labeled values, the configurations will step with the same probability in lock-step with all but negligible probability, as the readable labeled values will be syntactically equivalent (as the distributions of erased stores are equivalent).

In the case where the label of the labeled value is not readable, the rules used to step may not be the same as they are the results of encrypted values. 
For example, in one configuration a labeled value may be successfully fetched (using \rulename{Fetch-Valid}) but not in the corresponding configuration (e.g., \rulename{Fetch-Missing} was used). However, as noted above and by our inductive hypothesis, the different rules used will all step to a low equivalent configuration. 
In addition, though, to the configurations being low equivalent, it is also the case that the sums of the probabilities of all steps taken will be equivalent with all but negligible probability. For example, if $c_1$ steps using \rulename{Fetch-Missing} with probability $p_1$, and \rulename{Fetch-Valid} with probability $p_2$, it is also the case that $c_2$ will use the same rules \rulename{Fetch-Missing} with probability $p_1$ and \rulename{Fetch-Valid} with probability $p_2$ due to indistinguishability of the interactions. That is because if it did not, then an adversary could be constructed to distinguish the interactions based on the proportions of rules used by the \clio semantics. Intuitively, the draws of indistinguishable interactions will produce distributions of indistinguishable steps.

With this reasoning, we conclude that the probabilities of each corresponding single step taking place will be equal (e.g., in the diagram above, $p_1 = p_2$ and $p_2 = p_2'$). So, the resulting distribution of distributions over interactions will be still be valid \clio interactions and so the $\asymp$ relation holds (and, by Lemma~\ref{lemma:clio-indist}, they are also indistinguishable as a result).

\end{proof}

\clearpage

\subsection{Indistinguishability Proof}
\label{appendix:indist}


\bigskip

\begin{definition}[Chosen-Term Attack (CTA) Game]
Let the random variable $\IND{b}{\keystore}{\adversary}{\principals}{j}{n}$ denote 
the output of the following experiment, where 
$\cryptosys = (\Gen, \Enc, \Dec, \Sign, \Verify)$, 
$\adversary$ is a non-uniform ppt, $n \in \nat$, $b \in \{0,1\}$:

$$
\figCTAGame
$$

We say that \clio using \cryptosys is CTA (Chosen-Term Attack) Secure if for all
non-uniform ppt $\adversary$, $j \in \nat$, keystores \keystore, and principals $\principals$:
$$
  \distfam{\IND{0}{\keystore}{\adversary}{\principals}{j}{n}} \cequiv
  \distfam{\IND{1}{\keystore}{\adversary}{\principals}{j}{n}} 
$$.
\end{definition}

\begin{theorem}[Indistinguishability Theorem]
If \cryptosys if CPA Secure, then \clio using \cryptosys is CTA Secure.
\end{theorem}

\begin{proof} Direct result of low equivalence (interactions are valid \clio interactions, i.e., they satisfy the $\asymp$ relation) and Lemma~\ref{lemma:clio-indist} (indistinguishability of valid \clio interactions).

\end{proof}

\clearpage
\subsection{Leveraged Forgery Lemmas}
\label{appendix:forgery}

\begin{lemma}[Starting Label is a Floor]
For all keystores $\keystore_0$ and terms $t$ and strategies $\strategy$ and principals $\principals$ and $j$,
{\em $$
\begin{array}{l}
  \textbf{Pr}\Big[  
    ~\canFlowToStrictI{\projI{\pcOf{c}}}{\projI{\Start{\keystore_0}}}
  ~~\Big| \\ 
    \hspace*{0.5cm} \keystore' \drawnFrom \Gen(\principals, 1^n); \\
  \hspace*{0.5cm} \keystore = \keystore_0 \uplus \keystore'; \\
    \hspace*{0.5cm} \rconf{c}{\istoredist}{\interactionsdist}{\versions} \drawnFrom \multistep{\keystore}{\conf{\Start{\keystore_0}}{\Clr{\keystore_0}}{t}}{\strategy}{j}
    ~
  \Big] = 0
  \end{array}
$$}
\end{lemma}

\begin{proof}
We will prove this lemma in two steps: first by showing that the invariant is preserved across \clio steps $\rto{p}^*$ and then using that fact, we can show that the invariant is also preserved across \clio low steps $\Lowstep{p}*$. Our invariant will serve as our inductive hypothesis in both cases.

{\bf Proof on Step relation $\rto{p}$:} We will perform induction on the derivation of the steps (which will be finite when used with the low-step rules, i.e., it is well-founded) with the number of steps being $k$ being 1 less than the total number of (possibly high) steps in the context of a single low step, and our inductive hypothesis will be if  $\canFlowToStrictI{\projI{\pcOf{c}}}{\projI{\Start{\keystore_0}}}$ and,

$ 
\hspace*{4cm} \textbf{Pr}\Big[  
   ~\canFlowToStrictI{\projI{\pcOf{c'}}}{\projI{\Start{\keystore_0}}}
~~\Big| \newline 
   \hspace*{5cm} \keystore' \drawnFrom \Gen(\principals, 1^n); \newline~
  \hspace*{5cm} \keystore = \keystore_0 \uplus \keystore'; \newline~
\hspace*{5cm}
  \rconf{c}{\emptyset}{\skipcmd}{\zeroes} \rto{p_1} ... \rto{p_k} 
  \rconf{c'}{\istoredist'}{\interactionsdist'}{\versions'}
\Big] = 0
$

then,

$ 
\hspace*{4cm} \textbf{Pr}\Big[  
   ~\canFlowToStrictI{\projI{\pcOf{c''}}}{\projI{\Start{\keystore_0}}}
~~\Big| \newline~ 
 \hspace*{5cm} \keystore' \drawnFrom \Gen(\principals, 1^n); \newline~
  \hspace*{5cm} \keystore = \keystore_0 \uplus \keystore'; \newline~
\hspace*{5cm}
  \rconf{c}{\emptyset}{\skipcmd}{\zeroes} \rto{p_1} ... \rto{p_k} 
  \rconf{c'}{\istoredist}{\interactionsdist}{\versions} \rto{p_{k+1}} 
  \rconf{c''}{\istoredist'}{\interactionsdist'}{\versions'}
\Big] = 0
$

\bigskip

The base case $k=0$ is trivial as it is true by supposition.

For the inductive case, we now consider the derivation rule used for the $k+1$'th step and show that it preserves the inductive hypothesis, assuming it for the $k$'th step. We now perform a case analysis on the step used.

\begin{itemize}

\item {\bf Case \rulename{Internal-Step}:} 

In this derivation we have that:
\begin{mathpar}
\inferrule[internal step]{
  c' \lto c''
}{
  \rconf{c'}{\istoredist}{\interactionsdist}{\versions} \rto{1} \rconf{c''}{\istoredist'}{\interactionsdist'}{\versions_1}
}
\end{mathpar}

By inspection of each of the LIO rules, the label is manipulated in the following ways:
\begin{itemize}
\item In \rulename{Unlabel}, the current label is joined with the level of the labeled value, so the flows relation between the current label and the starting label is preserved.
\item In \rulename{Reset} the label is returned to its original label. However, from the \rulename{ToLabeled} rule, the label is based on the current label. As a result, since the label is based on a previous step's current label, and the inductive hypothesis assumes it was true for that point, then the label it is reset to is also satisfies the flow relation to the starting label.
\item In all other rules, the current label is not changed, which by supposition satisfies the flow relation.
\end{itemize}

\end{itemize}

{\bf Proof on Low-Step relation \Lowstep{}{p}}
By induction on $j$.
\begin{itemize}
\item {\bf Base Case:} $j = 1$: That is, we will prove the following:
$$
\begin{array}{l}
  \textbf{Pr}\Big[  
    ~\canFlowToStrictI{\projI{\pcOf{c}}}{\projI{\Start{\keystore_0}}}
  ~~\Big| \\~ 
 \hspace*{0.5cm} \keystore' \drawnFrom \Gen(\principals, 1^n); \\~
  \hspace*{0.5cm} \keystore = \keystore_0 \uplus \keystore'; \newline~
    \hspace*{0.5cm} \rconf{c}{\istoredist}{\interactionsdist}{\versions} \drawnFrom \multistep{\keystore}{\conf{\Start{\keystore}}{\Clr{\keystore}}{t}}{\strategy}{1}
    ~
  \Big] = 0
  \end{array}
$$



There are two cases we must consider in the low step relation, the \rulename{Low-Step} rule and the \rulename{Low-to-High-to-Low-Step} rule. In the \rulename{Low-Step}, we must show that the inductive hypothesis holds after a single \clio step $\rto{p}$, and for the \rulename{Low-to-High-to-Low-Step} rule must hold for many (finite) \clio steps $\rto{p}^*$. Note that the \rulename{Low-Step} rule is a special case of the \rulename{Low-to-High-to-Low-Step} rule so we only consider the more general case of preserving the invariant across many steps. To show this, we appeal to the previous proof made to show that the invariant is preserved across \clio steps. As a result, we have that:

\medskip

$ 
\hspace*{4cm} \textbf{Pr}\Big[  
   ~\canFlowToStrictI{\projI{\pcOf{c'}}}{\projI{\Start{\keystore_0}}}
 ~~\Big| \newline~ 
 \hspace*{5cm} \keystore' \drawnFrom \Gen(\principals, 1^n); \newline~
  \hspace*{5cm} \keystore = \keystore_0 \uplus \keystore'; \newline~
\hspace*{5cm}
  \rconf{c}{\emptyset}{\skipcmd}{\zeroes} \rto{p_1} ... \rto{p_k} 
  \rconf{c'}{\istoredist}{\interactionsdist}{\versions}
\Big] = 0$

From the previous proof on the single-step relation, we can conclude that $\canFlowToStrictI{\projI{\pcOf{c}}}{\projI{\Start{\keystore_0}}}$. As a result, we satisfy the inductive hypothesis.

\item {\bf Inductive Case:} j = k+1: That is, we will prove the following:
$$
\begin{array}{l}
  \textbf{Pr}\Big[  
    ~\canFlowToStrictI{\projI{\pcOf{c}}}{\projI{\Start{\keystore_0}}}
  ~~\Big| \newline~ 
 \hspace*{0.5cm} \keystore' \drawnFrom \Gen(\principals, 1^n); \\~
  \hspace*{0.5cm} \keystore = \keystore_0 \uplus \keystore'; \newline~
    \hspace*{0.5cm} \rconf{c}{\istoredist}{\interactionsdist}{\versions} \drawnFrom \multistep{\keystore}{\conf{\Start{\keystore}}{\Clr{\keystore}}{t}}{\strategy}{k+1}
    ~
  \Big] = 0
  \end{array}
$$

We can expand the $\multistepnm$ metafunction to be

$ \hspace*{4cm}
\big(\rconf{c''}{\istoredist''}{\interactionsdist''}{\versions'}, p \cdot p'\big)
    \newline~ \hspace*{4.5cm} \mathrm{where} \newline~ \hspace*{5cm}
    \big(\rconf{c'}{\istoredist'}{\interactionsdist}{\versions}, p\big) \in \multistep{\keystore}{c}{\strategy}{k}; \newline~ \hspace*{5cm}
    \interactionsdist' = \strategy(\interactionsdist); \newline~ \hspace*{5cm}
    (\rconf{c'}{\istoredist}{\interactionsdist}{\versions}, 
    {\interactionsdist'})
      \Lowstep{p'}
    \rconf{c''}{\istoredist''}{\interactionsdist''}{\versions'}
    $

    The strategy on the stores does not affect the current label.
After this adversary interaction, we can invoke the single-step lemma result here to conclude that $\canFlowToStrictI{\projI{\pcOf{c''}}}{\projI{\Start{\keystore_0}}}
  $. As a result, the inductive hypothesis is true.

\end{itemize}
With all cases accounted for in the low step relation, and the single-step relation, we can conclude the proof.

\end{proof}

\bigskip

\clearpage
\subsection{Leveraged Forgery Security Proof}

\begin{definition}[Values function]

Define the Values function as follows:

\scalebox{0.8}{%
\vbox{
\[
\begin{array}{lcll}
\mathrm{Values}_{\keystore}(\concat{\storecmd{\gv_k}{\Labeled{l_1}{b}}}{\interactions}) &=&   \concat{\storecmd{\gv_k}{\Labeled{l_1}{b}}}{\mathrm{Values}(\interactions)} & \textif \Labeled{l_1}{(\gv, \gv_k, n)} = \deserialize{\interactions}{\Labeled{l_1}{b}} \\
\mathrm{Values}_{\keystore}(\concat{\interaction}{\interactions}) &=&   \mathrm{Values}(\interactions) & \mathrm{otherwise} 
\end{array}
\]
}}
\end{definition}

\begin{theorem}[Existential forgery under chosen message attack]
For all keystores $\keystore_0$, principals $p$ in principal sets $\principals$, and $j$ if $\cryptosys$ is secure against existential forgery under chosen message attacks, then \\ for all $~\canFlowToI{\projI{l_2}}{\canFlowToI{\projI{l_1}}{\projI{\authorityOf{\keystore_0}} \wedge p}}$,

$$
\figForgeryProb
$$
is negligible in $n$. 
\end{theorem}

\begin{proof}

We consider the level required to produce a valid signature during a store operation. The signature must be valid for a $p \in \principals$. By inspection of the \clio  semantics, the only way to a valid signature would occur in the interaction is during the store operation, which uses the labeled value's label, or by the strategy. 

According to the store operation, current label must be bounded above by the label
of the labeled value (i.e., $\canFlowTo{\lcur}{l_1}$). For integrity, this means
the current label's integrity component $I$ must be at least as trustworthy as the principal $p$.


By the previous lemma, we can conclude that the current integrity label will never be at a level $I'$ such that 
$$
\canFlowToStrictI{I'}{\projI{\Start{\keystore_0}}}$$ 
 This means that the level of the \clio computation would need to be at least 
$$
I = \projI{\authorityOf{\keystore_0 \uplus \{ p \mapsto \Gen(1^n) \} }}$$

\noindent By unpacking the definition of $\Start{\keystore \uplus \{ p \mapsto \Gen(1^n) \}}$, this integrity label satisfies the following relation
$$
\canFlowToStrictI{ I }{ \projI{\Start{\keystore_0} }}
$$

Which we have shown is impossible to reach. As a result, the current label's integrity level will never be at a level where it can sign the value using $p$'s private signing key.  

As a result, the only way a high integrity value could be in the challenge store $\istore'$ and not in the original store $\istore$ would require the strategy to forge a signature itself without \clio's assistance. Since this occurs with
only negiglible probability, we have satisfied the proof obligation.
\end{proof}

\clearpage
\section{Implementation Details}
\label{app:impl}

In this section we discuss additional implementation details of our
prototype.

\subsection{Storing and Fetching}

For each category used in the
program we generate a symmetric key and two RSA key pairs: an
encryption/decryption key pair and a signing/verification key
pair. This information is stored in the database after being
asymmetrically encrypted and then signed as described in
Section~\ref{subsec:keystore}. Category key generation relies on the
RSA key pairs for each principal involved, which should be supplied by
the user in the form of an initial keystore when the \clio computation
starts.


After the relevant IFC effects have been performed, storing a labeled
value involves fetching the symmetric key for each
category in its confidentiality clause as well as the signature keys
that correspond to each category in its integrity clause, potentially
generating these on the fly. The labeled value is serialized to a
bitstring, then RSA-PSS signed by at least one principal per integrity
category, and finally AES256-CTR onion-encrypted using the symmetric
key for each confidentiality category. Fetching involves the dual
operations, i.e., symmetric decryption and RSA-PSS signature
verification.

In order to avoid problems with improperly
escaped strings, we encode every bitstring in
base64.

\subsection{User API}

Our library provides all the \clio{} operations described in the
paper, plus a few extra functions that are necessary to glue \clio code
with the rest of the program. Here are some of the most important ones.

\clio code can be run using the \texttt{evalCLIO} function. This
function takes two arguments: a record $\text{initialState}$ of type
\texttt{CLIOState} and a \clio computation $m$. The record
$\text{initialState}$ provides initial values for the current label,
the current clearance, the keystore, the version map and the store
label. The function simply establishes a connection with a Redis
server 
 and executes $m$ using
that database as the store and $\text{initialState}$ as the local
state. 

In order to generate keystores, we provide the utility function
\texttt{initializeKeyMapIO}. This function takes a list of principals
as argument, and produces a keystore with fresh asymmetric key pairs
for all of them. Our prototype does not provide means to store these
keystores beyond the execution of the program, but it would be
straightforward for users to implement this functionality in their own
programs, or with their own PKI.


\end{document}
